\documentclass[11pt,a4paper]{amsart}
\usepackage[utf8]{inputenc}
\usepackage{amsmath}
\usepackage{amsfonts}
\usepackage{amssymb}
\usepackage{graphicx}
\usepackage{multirow}
\usepackage{booktabs}
\usepackage{array}
\usepackage[ruled,vlined]{algorithm2e}
\usepackage{amsaddr}
\usepackage{subfig}
\usepackage{tablefootnote}

\usepackage[left=2.2cm,right=2.2cm,top=2cm,bottom=2cm]{geometry}

\newcommand{\D}{~{\mathop{}\!\mathrm{d}}} 

\newcommand{\R}{\mathbb{R}}
\newcommand{\Q}{\mathbb{Q}}
\newcommand{\C}{\operatorname{C}}
\newcommand{\G}{\mathcal{G}}
\newcommand{\HM}{\mathcal{H}_{B,L}}

\newcommand{\N}{\mathbb{N}}

\newcommand{\PP}{\mathbb{P}}

\newcommand{\E}{\mathbb{E}}

\newcommand{\one}{ 1 \hspace{-3pt} \mathrm{l}} %

\numberwithin{equation}{section}  
\newtheorem{defn}{Definition}[section]

\newtheorem{rem}[defn]{Remark}
\newtheorem{thm}[defn]{Theorem}
\newtheorem{prop}[defn]{Proposition}

\newtheorem{lem}[defn]{Lemma}
\newtheorem{asu}[defn]{Assumption}

\usepackage{tikz}
\usetikzlibrary{decorations.pathreplacing}
\usetikzlibrary{fadings}



\usepackage{hyperref}
\usepackage{xcolor}
\hypersetup{
    colorlinks,
    linkcolor={red!50!black},
    citecolor={blue!50!black},
    urlcolor={blue!80!black}
}

\begin{document}
\title[Detecting Robust Statistical Arbitrage Strategies with Deep Neural Networks]{Detecting data-driven robust statistical arbitrage strategies with deep neural networks}

\author[A. Neufeld, J. Sester, D. Yin, ]{ Ariel Neufeld$^{1}$, Julian Sester$^{2}$, Daiying Yin$^{3}$}
\thanks{\textit{E-Mail addresses:}~
$^{1}$\texttt{ariel.neufeld@ntu.edu.sg},~
$^{2}$\texttt{jul$\_$ses@nus.edu.sg},~
$^{3}$\texttt{yind0004@e.ntu.edu.sg}}

\maketitle
\date{\today}

\begin{center}
\normalsize{\today} \\ \vspace{0.5cm}
\small\textit{$^{1,3}$NTU Singapore, Division of Mathematical Sciences,\\ 21 Nanyang Link, Singapore 637371.\\[2mm]
$^{2}$NUS Singapore, Department of Mathematics,\\ 21 Lower Kent Ridge Road, Singapore 119077}                                                                                                                              
\end{center}

\begin{abstract}
We present an approach, based on deep neural networks, that allows identifying robust statistical arbitrage strategies in financial markets. Robust statistical arbitrage strategies refer to  trading strategies that enable profitable trading under model ambiguity. The presented novel methodology allows to consider a large amount of underlying securities simultaneously and does not depend on the identification of cointegrated pairs of assets, hence it is applicable on high-dimensional financial markets or in markets where classical pairs trading approaches fail. Moreover, we provide a method to build an ambiguity set of admissible probability measures that can be derived from observed market data. Thus, the approach can be considered as being model-free and entirely data-driven. We showcase the applicability of our method by providing empirical investigations with highly profitable trading performances even in $50$ dimensions, during financial crises, and when the cointegration relationship between asset pairs stops to persist.
\end{abstract}

\section{Introduction}
In this paper we present an empirically tractable method using neural networks that allows identifying robust statistical arbitrage strategies in financial markets.

The term \emph{statistical arbitrage} is commonly used in finance to describe trading strategies which are profitable on average, but, in contrast to pure arbitrage strategies (see e.g., \cite{cohen,cui,neufeld2020model} for their detection), not necessarily in every market scenario that is deemed to be possible. This generalized notion therefore forms the foundation for a systematic approach enabling to trade profitably even in markets in which it is difficult or impossible to detect pure arbitrage strategies.

A popular class of trading strategies that are often referred to as statistical arbitrage strategies are  \emph{pairs trading} strategies, which all rely on the fundamental idea that, given two financial assets are strongly related, e.g., through a cointegration property (compare e.g.~\cite{caldeira2013selection}) or through a low level of the variance of the spread between the two assets (compare e.g.~\cite{gatev2006pairs}), then deviations of the spread are assumed to last only for a short period of time and, thus, eventually the spread of the asset pair will return to its long-term equilibrium. This mean-reversion property is exploited by trading in the opposite direction of the deviation after the deviation has exceeded a certain threshold and by clearing the position when the spread has reached again a level close to its long-term equilibrium, see also \cite{avellaneda2010statistical, do2006new, elliott2005pairs, gatev2006pairs, krauss2017statistical, krauss2017deep, rad2016profitability, vidyamurthy2004pairs}.

The apparent drawback of applying pairs trading strategies is the strong dependence on the underlying mean-reversion property of the spread process. Indeed, if the mean-reversion relation breaks down and thus the spread does not converge to its equilibrium, then the pairs trading approach will in general not be profitable. Moreover, in an empirical study, \cite{do2010does} provide evidence that the profitability of pairs trading strategies has declined over the recent years, mainly as it has become increasingly difficult to identify profitable asset pairs.

In contrast, our contribution builds on another class of statistical arbitrage strategies that do not rely on a mean-reversion property. We will show that these strategies are profitable even in periods where the pairs trading approach fails. Our approach relies on the idea of Bondarenko (\cite{bondarenko2003statistical}), who introduced and characterized statistical arbitrage strategies as  strategies which are profitable on average given any terminal value of the underlying securities at a fixed maturity $t_n>0$, i.e., if $\Phi(S)$ denotes the profit of a trading strategy investing in underlying securities $S=(S)_{ t \leq t_n}$, then statistical arbitrage strategies fulfil $\E_{\PP}[\Phi(S)~|~S_{t_n}] \geq 0 $ as well as $\PP\big(\E_{\PP}[\Phi(S)~|~S_{t_n}] > 0\big)>0$. Based on this idea, \cite{kassberger2017additive} generalized the notion of statistical arbitrage from \cite{bondarenko2003statistical} by introducing $\mathcal{G}$-arbitrage defined through zero-cost payoffs $Y$ (which are not necessarily the payoffs of  trading strategies) fulfilling \begin{equation}\label{eq_g_arbitrage}
\E_{\PP}[Y~|~\mathcal{G}] \geq 0 \text{ and } \PP \big(\E_{\PP}[Y~|~\mathcal{G}]> 0\big)>0
\end{equation} for $\mathcal{G}$ being a $\sigma$-algebra $\mathcal{G}\subseteq \sigma(S)$, which allows, in particular, to take into account more flexible choices of trading strategies, possibly adjusted to available information. 

Building on the definition of $\mathcal{G}$-arbitrage, the results from \cite[Proposition 1]{bondarenko2003statistical}, \cite[Proposition 6]{kassberger2017additive}, and \cite[Theorem 3.3]{rein2021generalized} characterize the existence of $\mathcal{G}$-arbitrage strategies by relating the absence of  strategies fulfilling \eqref{eq_g_arbitrage} to the existence of $\mathcal{G}$-measurable Radon---Nikodym densities, a result which can be considered as an extension of the fundamental theorem of asset-pricing~(compare \cite{acciaio2016model,bouchard2015arbitrage,BurzoniRiedelSonerFTAP,delbaen1994general,harrison1981martingales,schachermayer2004fundamental} for several versions of the fundamental theorem of asset-pricing in different underlying settings) which connects the absence of arbitrage with the existence of pricing measures. The authors from \cite{rein2021generalized} further propose and validate empirically an embedding-methodology to exploit statistical arbitrage on financial markets. All these mentioned contributions assume that the underlying securities behave according to a previously fixed underlying probability measure $\PP$. To account for ambiguity with respect to the choice of an underlying probability measure \cite{lutkebohmert2021robust} recently introduced the notion of $\mathcal{P}$-robust $\mathcal{G}$-arbitrage referring to  trading strategies that allow for $\mathcal{G}$-arbitrage independent which measure $\PP$ from the considered ambiguity set of measures $\mathcal{P}$ is the "correct" underlying probability measure, as the strategy is required to be profitable for all measures from the ambiguity set. Hence, this notion allows, in particular, to take into account probability measures that would incur losses for conventional pairs trading strategies, and thus to determine trading strategies that are profitable even in scenarios where pairs trading fails.

To compute  $\mathcal{P}$-robust statistical arbitrage strategies, i.e., $\mathcal{P}$-robust $\mathcal{G}$-arbitrage strategies for the choice $\mathcal{G}=\sigma(S_{t_n})$, where $S_{t_n}$ denotes the terminal values of the underlying securities, a methodology in \cite{lutkebohmert2021robust} is provided that relies on a linear programming approach. This routine cannot be applied in high-dimensional settings since it relies on linear programming. To overcome this limitation, we establish a numerical method involving deep neural networks to determine  $\mathcal{P}$-robust statistical arbitrage strategies. To this end, we introduce a functional which penalizes  trading strategies that are not statistical arbitrage strategies for each of the considered measures from the ambiguity set $\mathcal{P}$. We show as a main result in Theorem~\ref{thm_summary} that minimizing this functional among trading strategies that can be represented as the outputs of deep neural networks is, for a sufficiently large penalization parameter, equivalent to the determination of  $\mathcal{P}$-robust statistical arbitrage strategies. In particular, through the representation of trading strategies by neural networks, our approach is applicable even in high dimensions, in contrast to approaches relying on linear programming, and is hence tractable even when the trading strategy depends on a large amount of underlying assets. We provide empirical evidence in Section~\ref{exa_large_number} of the applicability of our approach in a high-dimensional setting, where we consider strategies investing in $50$ underlying assets from the S\&P $500$ universe.

As a choice of the ambiguity set $\mathcal{P}$, we propose a purely data-driven approach as follows. First, relying on a historical time-series of stock returns, a probability measure is constructed that considers each sequence of observed historical returns as equally likely to occur again in the future. Second, since this measure allows no deviation of future returns from historical returns, we consider a ball around this measure (where the ball is taken with respect to the Wasserstein-distance). Finally, we provide a methodology to explicitly generate measures from this ball.

In contrast to the determination of statistical arbitrage strategies with respect to a fixed underlying probability measure, taking into account the ambiguity set  allows, particularly, that the future developments of the underlying process differ to a certain degree (this degree can be specified by the applicant and is expressed in terms of the Wasserstein-distance) from the historical evolution. Indeed, by using this data-driven ambiguity set of measures $\mathcal{P}$ we provide empirical evidence for the profitability of our approach in overall bad market scenarios (Section~\ref{exa_bad_market}), in high-dimensional financial markets (Section~\ref{exa_large_number}), and in scenarios where classical pairs trading approaches fail (Section~\ref{exa_pairs_trading}).

The remainder of this paper is as follows. In Section~\ref{sec_main} we introduce the underlying setting and present our main theoretical results. Section~\ref{sec_application} presents an approach to first derive an ambiguity set of measures $\mathcal{P}$ from financial data and then use a numerical method to compute $\mathcal{P}$-robust statistical arbitrage strategies when considering the data-implied set $\mathcal{P}$ as the ambiguity set of physical measures. In Section~\ref{sec_real_world_examples} we then provide empirical evidence for the applicability of our approach by studying several relevant examples involving real-world data. Section~\ref{sec_conclusion} concludes and provides an outlook to future research.
The proofs of all mathematical statements are provided in Section~\ref{sec_proofs}.

\section{Setting and main results}\label{sec_main}

We consider a financial market on which we examine the evolution of $d \in \N$ underlying securities at $n\in \N$ future times $t_1 < \dots < t_n$. The future values of all securities are assumed to be bounded from below and from above by some (possibly very large) constants\footnote{Compare also Figure~\ref{fig_NN_vs_LP} where we illustrate that if the width of the interval $[\underline{K}^j,\overline{K}^j]$, $j \in \{1,\dots,d\}$, spanned by the bounds $\underline{K}^j,\overline{K}^j$ of the securities, exceeds  a certain threshold, then the performance of our approach does not further improve. Thus, assuming a setting considering large bounds for the underlying securities turns out to be no restriction in practice when applying our approach.}. More precisely, for each $j \in \{1,\dots,d\}$ there exist security-specific bounds $\underline{K}^j,\overline{K}^j\in \R$  satisfying $\underline{K}^j<\overline{K}^j$. We then define
$$
\Omega_i: = {[\underline{K}^1,\overline{K}^1]}^i\times \cdots \times {[\underline{K}^d,\overline{K}^d]}^i \subset \R^{id},\qquad i=1,\dots,n,
$$
which refers to the bounds for the evolution of the underlying assets until time $t_i$. Then, we model the evolution of the values of the underlying securities until time $t_n$ through the canonical process on $
\Omega:=\Omega_n$ and we denote this process by $S:=(S_{t_i}^j)_{i=1,\dots,n}^{j=1,\dots,d}$, which is for $i=1,\dots,n$; $j=1,
\dots,d$ defined by
\[
S_{t_i}^j:\Omega \to [\underline{K}^j,\overline{K}^j], ~(x_1^1,\dots,x_n^d)\mapsto x_i^j.
\]
Moreover, we introduce for the values of the $d$ securities at each time $t_i$ the notation $S_{t_i}:=(S_{t_i}^j)^{j=1,\dots,d}$ and denote for $(S_{t_0}^j)^{j=1,\dots,d} \in [\underline{K}^1,\overline{K}^1]\times \cdots \times {[\underline{K}^d,\overline{K}^d]}$ the observed, and therefore deterministic, spot prices at initial time $t_0$. Then we define the \emph{gross profit} of a  trading strategy $\Delta:=(\Delta_i^j)_{i=0,\dots,n-1}^{j=1,\dots,d}$, consisting of Borel measurable functions $\Delta_i^j: \R^{id}\to \R$, (where $\Delta_0^j\in \R$ is a constant) via
\[
(\Delta \cdot S)_n:=\sum_{j=1}^d\sum_{i=0}^{n-1} \Delta_i^j(S_{t_1},\dots,S_{t_i})(S_{t_{i+1}}^j-S_{t_i}^j).
\]
{Additionally, we assume that trading in a self-financing strategy $\Delta$ induces three types of trading costs, namely, transaction fees, borrowing costs when holding a short position, and liquidity costs arising from bid-ask spreads. If a trader adjusts at time $t_i$ her trading position $\Delta_{i-1}^j$ and changes it to a new position $\Delta_i^j$, then this adjustment
causes transactions costs $${c_{\operatorname{trans}}}_i^j\left( S_{t_i}^j ,~\Delta_i^j(S_{t_1},\dots,S_{t_i})-\Delta_{i-1}^j(S_{t_1},\dots,S_{t_{i-1}})\right)$$ for some measurable non-negative function ${c_{\operatorname{trans}}}_i^j:{ \R \times }\R\rightarrow \R_+$ as well as liquidity costs associated to the bid-ask spread of size
$$
{c_{\operatorname{spread}}}_i^j\left( S_{t_i}^j ,~\Delta_i^j(S_{t_1},\dots,S_{t_i})-\Delta_{i-1}^j(S_{t_1},\dots,S_{t_{i-1}})\right)
$$
for some measurable non-negative ${c_{\operatorname{spread}}}_i^j:{ \R \times }\R\rightarrow \R_+$.
Moreover, if the held position after trading is negative, then this induces additional borrowing costs of size
$$
{c_{\operatorname{short}}}_i^j\left( S_{t_i}^j ,~\Delta_i^j(S_{t_1},\dots,S_{t_i})\right)
$$
for some measurable non-negative function ${c_{\operatorname{short}}}_i^j:{ \R \times }\R\rightarrow \R_+$ being zero  if the agent holds a long position. Thus, given some trading strategy $\Delta$, the overall trading costs until time $t_n$ are measured by
\begin{align*}
	\C_n(\Delta)&:= \sum_{j=1}^d\sum_{i=0}^{n} \bigg[{c_{\operatorname{trans}}}_i^j\left(S_{t_i}^j ,\Delta_i^j(S_{t_1},\dots,S_{t_i})-\Delta_{i-1}^j(S_{t_1},\dots,S_{t_{i-1}})\right)\\
	&+{c_{\operatorname{{spread}}}}_i^j\left(S_{t_i}^j ,\Delta_i^j(S_{t_1},\dots,S_{t_i})-\Delta_{i-1}^j(S_{t_1},\dots,S_{t_{i-1}})\right)+{c_{\operatorname{{short}}}}_i^j\left(S_{t_i}^j ,\Delta_i^j(S_{t_1},\dots,S_{t_i})\right)\bigg]
\end{align*}

with the convention $\Delta_{-1}^j\equiv\Delta_n^j \equiv 0 $ for all $j=1,\dots,d$. Therefore, the \emph{net overall profit} of a self-financing trading strategy $\Delta$ is given by 
\[
(\Delta \cdot S)_n - \C_n(\Delta).
\]

\begin{asu}\label{asu_transaction_costs}
We assume for all $i=0,\dots,n$, $j=1,\dots,d$ and $S_{t_i}^j \in \R$ that the functions 
\begin{align*}
\R \ni x \mapsto &{c_{\operatorname{trans}}}_i^j\left( S_{t_i}^j,x\right) \in \R,\\
\R \ni x \mapsto &{c_{\operatorname{spread}}}_i^j\left( S_{t_i}^j,x\right) \in \R,\\
\R \ni x \mapsto &{c_{\operatorname{short}}}_i^j\left( S_{t_i}^j,x\right) \in \R
\end{align*}
are continuous.
\end{asu}

\begin{rem}
\begin{itemize}
\item[(a)]
Typical forms of transaction costs ${c_{\operatorname{trans}}}_i^j$ which satisfy Assumption~\ref{asu_transaction_costs} are given by the following two specifications.
\begin{itemize}
\item[(i)] \emph{Per share transaction costs}, where ${c_{\operatorname{trans}}}_i^j(S_{t_i}^j,x):={\lambda_{\operatorname{trans}}}_i^j\cdot |x|$ for some ${\lambda_{\operatorname{trans}}}_i^j>0$ for all $i,j$.
\item[(ii)] \emph{Proportional transaction costs}, where ${c_{\operatorname{trans}}}_i^j(S_{t_i}^j,x):={\lambda_{\operatorname{trans}}}_i^j S_{t_i}^j\cdot |x|$ for some \mbox{${\lambda_{\operatorname{trans}}}_i^j>0$} for all $i,j$.
\end{itemize}
Compare also \cite[Section 2.2]{buehler2019deep} and \cite[Section 3.1]{cheridito2017duality}, where additional examples are discussed.
\item[(b)] To account for the cost induced by the bid-ask spreads, we assume that the spreads are symmetric around the stock prices $S_{t_i}^j$, and hence ${c_{\operatorname{{spread}}}}_i^j(S_{t_i}^j,x):=0.5\cdot{\lambda_{\operatorname{spread}}}_i^j\cdot |x|$, where ${\lambda_{\operatorname{spread}}}_i^j>0$ corresponds to the size of the bid-ask spread.
\item[(c)]  Finally, to account for borrowing costs, we consider ${c_{\operatorname{{short}}}}_i^j(S_{t_i}^j,x):={\lambda_{\operatorname{short}}}_i^j\cdot\max\{-x,0\}\cdot S_{t_i}^j$, where ${\lambda_{\operatorname{short}}}_i^j>0$ are  daily borrowing costs.
\end{itemize}
\end{rem}}

\subsection{$\mathcal{P}$-robust $\G$-arbitrage strategies}

The notion of (statistical) arbitrage depends crucially on the choice of an underlying probability measure. To take into account uncertainty with respect to the choice of the correct underlying model - and the associated probability measure - we consider a set $\mathcal{P}\subseteq \mathcal{M}_1(\Omega)$, where $\mathcal{M}_1(\Omega)$ denotes the set of Borel probability measures on $\Omega$. The set $\mathcal{P}$ can thus be considered as an ambiguity set of admissible physical probability measures. Our approach relies on the notion of $\mathcal{P}$-robust $\G$-arbitrage, which is borrowed from \cite{lutkebohmert2021robust} and adjusted to the present multi-asset setting with { additional trading costs}.
\begin{defn}[$\mathcal{P}$-robust $\G$-arbitrage]\label{defn_stat_arb}
Let $\mathcal{P}\subseteq \mathcal{M}_1(\Omega)$, and let $\G \subseteq\sigma(S)$ be a sigma-algebra on $\Omega$. Then, a self-financing strategy $\Delta$ is called a \emph{$\mathcal{P}$-robust $\G$-arbitrage strategy} if the following conditions are fulfilled.
\begin{itemize}
\item[(i)] $\E_\PP\left[(\Delta \cdot S)_n-\C_n(\Delta)~\middle|~\G\right] \geq 0 ~~~\PP\text{-a.s. for all } \PP \in \mathcal{P},$
\item[(ii)]$\E_\PP[(\Delta \cdot S)_n-\C_n(\Delta)]>0 \text{ for some } \PP \in \mathcal{P}$.
\end{itemize}
\end{defn}
Let $B>0$, $L >0$ and consider the following set of strategies which are bounded by $B$ and which are $L$-Lipschitz.\footnote{Let $a,b\in \N$. Then, we say a function $f:\R^a \rightarrow \R$ is $L$-Lipschitz if $|f(x)-f(y)|\leq L \|x-y\|_a$ for all $x,y \in \R^a$, with $\|\cdot \|_{a}$ denoting the Euclidean norm on $\R^{a}$. Note that for any $X \subseteq \R^a$ and $f:X \rightarrow \R^b$ we denote by $\|f\|_{\infty,X,b}:=\sup_{x\in X} \|f(x)\|_{b}$ the supremum-norm restricted to the set $X$. Moreover, to simplify notations, we write $\|f\|_{\infty,X}$ for $\|f\|_{\infty,X,b}$  if the dimension $b$ is clear.}
\begin{equation}\label{eq_defn_H_m}
\begin{aligned}
\HM:=\bigg\{ h_{c,\Delta}:\Omega\rightarrow \R~\bigg|~&h_{c,\Delta}(S)=c+(\Delta\cdot S)_n-\C_n(\Delta) \\
&\text{for some } c \in \R,~ \Delta={(\Delta_i^j)}_{i=0,\dots,n-1}^{j=1,\dots,d},\\
&\text{with } \Delta_0^j\in \R,~\Delta_i^j:\R^{id} \rightarrow \R  ~L\text{-Lipschitz for all }i=1,\dots,n,~j=1,\dots,d\\ &\text{and with}~ |c| \leq B, ~|\Delta_0^j|\leq B,\text{ and}~\|\Delta_i^j\|_{\infty,\Omega_i}\leq B\\
&\hspace{6.5cm}\text{ for all } i=1,\dots,n,~j=1,\dots,d \bigg\} .
\end{aligned}
\end{equation}
The constant $B$, which restricts the amount an investor is able to invest, possesses therefore a natural interpretation as a given budget constraint. Moreover, we obtain from the following lemma that $\HM$ is compact.
\begin{lem}\label{lem_compactness}
Let Assumption~\ref{asu_transaction_costs} hold true. Then, for all $B >0$ and all $L>0$ the space $\HM$ is compact in the uniform topology on $\Omega$.
\end{lem}
Given some measurable function $\Phi:\Omega\rightarrow \R$ and some $\sigma$-algebra $\G\subseteq \sigma(S)$, we are interested in solving the following conditional super-replication problem
\[
\Gamma_{B,L}(\Phi,\G):=\inf_{h_{c,\Delta} \in \HM} \bigg\{c~\bigg|~\E_{\PP}[h_{c,\Delta}(S)~|~\G] \geq \E_{\PP}[\Phi(S)~|~\G] ~~~\PP\text{-a.s. for all }\PP \in \mathcal{P}\bigg\}.
\]

Solutions of $\Gamma_{B,L}$ are interesting for two reasons. For a financial derivative with payoff function $\Phi$, the value $\Gamma_{B,L}(\Phi,\G)$ describes the upper bound of prices for $\Phi$ which do not allow for $\mathcal{P}$-robust $\G$-arbitrage, given the market admits no $\mathcal{P}$-robust $\G$-arbitrage opportunities, see also the discussion in \cite[Section 4.2]{lutkebohmert2021robust}. On the other hand, if the market admits $\mathcal{P}$-robust $\G$-arbitrage, the definition of $\Gamma_{B,L}$  allows to determine $\mathcal{P}$-robust $\G$-arbitrage strategies. Indeed, when considering $\Phi \equiv 0$, then $\Gamma_{B,L}(0,\G)< 0$ leads, according to Definition~\ref{defn_stat_arb}, directly to a $\mathcal{P}$-robust $\G$-arbitrage strategy, since $\Gamma_{B,L}(0,\G)< 0$ implies the existence of some strategy with negative price which has, conditional on $\mathcal{G}$, a non-negative payoff.

\subsection{An approximation of $\Gamma_{B,L}$ via penalization}
In the following we establish a numerical method to compute $\Gamma_{B,L}$ as well as associated $\mathcal{P}$-robust $\G$-arbitrage strategies relying on a penalization approach, similar to the approaches pursued in \cite{benamou,cominetti,degennaro2020,eckstein2021robust,
eckstein2019computation,eckstein2020robust,henry2019martingale}
. To facilitate a numerical implementation, we impose the standing assumption that $\mathcal{P}$ contains a finite amount of measures.
\begin{asu}\label{asu_P}
The set $\mathcal{P} \subset \mathcal{M}_1(\Omega)$ contains a finite amount of measures.
\end{asu}
Let $M>0$, let $\Phi:\Omega\rightarrow \R$ be some Borel measurable function, and let $\G\subseteq \sigma(S)$ be some $\sigma$-algebra on $\Omega$, then to approximate $\Gamma_{B,L}$ numerically, we introduce for each $k \in \N$ the following functional
\begin{equation*}
\Gamma_{B,L,k}(\Phi,\G):=\inf_{h_{c,\Delta} \in \HM} \bigg\{c~+ k \sum_{\PP \in \mathcal{P}} \int_\Omega \beta\big(\E_{\PP}[\Phi(S)-h_{c,\Delta}(S) ~|~ \G]\big) \D \PP \bigg\},
\end{equation*}
where $\beta:\R \rightarrow \R_+$ is some penalization function which penalizes trading strategies $h_{c,\Delta} \in \HM$ for which the conditional super-replication constraint
\[
\E_{\PP}[h_{c,\Delta}(S)~|~\G] \geq \E_{\PP}[\Phi(S)~|~\G] ~~~\PP\text{-a.s. for all } \PP \in \mathcal{P}
\]
is violated. To ensure that $\beta$ penalizes as intended, we impose additionally the following assumption on the geometry of $\beta$.
\begin{asu}\label{asu_beta}
Let $\beta:\R \rightarrow \R_+$ be a continuous function with the property that 
$\beta(x) =  0, \text{if } x \leq 0$, and that $\beta(x) > 0, \text{if } x > 0$.
\end{asu}
A natural choice for the function $\beta$ is therefore given by
\[
\beta(x) := \lambda \cdot \max\{0,x\}^p
\]
for some $\lambda>0$, $p > 0$. As we will show below, the compactness of $\HM$, stated in Lemma~\ref{lem_compactness}, induces the existence of some strategy attaining the value of $\Gamma_{B,L,k}$.
\begin{lem}\label{lem_attainment}
Let Assumptions~\ref{asu_transaction_costs}, \ref{asu_P}, \ref{asu_beta} hold true. Let $B >0$, let $L>0$, and let $\Phi:\Omega\rightarrow \R$ be Borel measurable with $\|\Phi\|_{\infty,\Omega} \leq B$. \\
Then  for every $\sigma$-algebra $\G \subseteq\sigma(S)$ on $\Omega$ and for all $k \in \N$ the infimum of $\Gamma_{B,L,k}$ is attained, i.e., there exists $h_{c^*,\Delta^*} \in \HM$ such that
\[
\Gamma_{B,L,k}(\Phi,\G)=c~+ k \sum_{\PP \in \mathcal{P}} \int_{\Omega} \beta\big(\E_{\PP}[\Phi(S)-h_{c^*,\Delta^*} (S) ~|~ \G]\big) \D \PP.
\]
\end{lem}
By employing Lemma~\ref{lem_compactness} and Lemma~\ref{lem_attainment}, we can derive the following result which shows that $\Gamma_{B,L}$ can indeed be approximated arbitrarily well by $\Gamma_{B,L,k}$ by choosing $k$ sufficiently large.
\begin{prop}\label{prop_convergence}
Let Assumptions~\ref{asu_transaction_costs}, \ref{asu_P}, \ref{asu_beta} hold true. Let $B >0$, let $L>0$, and let $\Phi:\Omega\rightarrow \R$ be Borel measurable with $\|\Phi\|_{\infty,\Omega} \leq B$. \\
Then, we have for every $\sigma$-algebra $\G \subseteq\sigma(S)$ on $\Omega$ that
\[
\lim_{k \rightarrow \infty} \Gamma_{B,L,k}(\Phi,\G) = \Gamma_{B,L}(\Phi,\G).
\]
\end{prop}
\subsection{An approximation of $\Gamma_{B,L}$ through neural networks}\label{sec_nn}
Next, we show how the functional  $\Gamma_{B,L,k}$ can be approximated with trading strategies whose positions are represented by an appropriate class of neural networks.\\
To this end, we present a problem-tailored class of fully-connected feed-forward neural networks which we will use to approximate the class of bounded and Lipschitz-continuous trading strategies $\HM$.

We refer to \cite{bengio2009learning}, \cite{dixon2020machine} \cite{goodfellow2016deep}, or \cite{hassoun1995fundamentals}, which are excellent monographs on neural networks, for a general introduction to networks, as well as to  \cite{baes2021low}, \cite{buehler2019deep}, \cite{degennaro2020}, \cite{eckstein2021robust}, \cite{eckstein2019computation}, \cite{eckstein2020robust},  \cite{lutkebohmert2021deephedging}, \cite{neufeld2021deep}, \cite{ruf2020neural} for applications in financial mathematics.

More specifically, for any $m\in \N$ we consider fully-connected neural networks which have input dimension $m\in \N$, output dimension $d$, and number of layers $l+1 \in \N$ defined as functions of the form
\begin{equation}\label{eq_nn_function}
\begin{aligned}
\R^{m} &\rightarrow \R^d\\
{x} &\mapsto {A_l} \circ {\varphi}_l \circ {A_{l-1}} \circ \cdots \circ {\varphi}_1 \circ {A_0}({x}),
\end{aligned}
\end{equation}
where $({A_i})_{i=0,\dots,l}$ are affine functions of the form 
\begin{equation}\label{eq_A_i_def}
{A_0}: \R^{{m}} \rightarrow \R^{h_1},\qquad {A_i}:\R^{h_i}\rightarrow \R^{h_{i+1}}\text{ for } i =1,\dots,l-1, \text{(if } l>1), \text{ and}\qquad {A_l} : \R^{h_l} \rightarrow \R^d,
\end{equation}
where for  $i=1,\dots,l$ we have ${\varphi}_i(x_1,\dots,x_{h_i})=\left(\varphi(x_1),\dots,\varphi(x_{h_i})\right)$ with $\varphi:\R \rightarrow \R$ being a non-constant function which is referred to as \emph{activation function}.
 The vector ${h}=(h_1,\dots,h_{l}) \in \N^{l}$ contains the number of neurons of the hidden layers. We call $h$ also the \emph{hidden dimension} of the neural network.
 
\begin{rem}\label{rem_activation_functions}
Frequent choices of activation functions that have turned out to perform well in practice include the \emph{sigmoid} function
\[
\varphi(x):=\frac{1}{1+\exp(-x)}
\]
and the \emph{ReLU} function
\[
\varphi(x):=\max\{x,0\}.
\]
For an extensive overview on different activation functions we refer the reader to \cite{goodfellow2016deep}.
\end{rem}
\begin{asu}\label{asu_activation_function_phi}We assume that the activation function $\varphi:\R \rightarrow \R$ is either one time continuously differentiable and not polynomial, or that $\varphi$ is the ReLU. \\
\end{asu}
For a fixed activation function, we consider the set of all neural networks with input dimension ${m}\in \N$, $l\in \N$ hidden layers, hidden dimension ${h}\in \N$ , and output dimension $d\in \N$ and we call this set $\mathfrak{N}_{m,d}^{l,{h}}$. Next, we introduce an additional degree of freedom by considering all neural networks where the hidden dimension is arbitrary and the number of hidden layers is not specified, i.e., 
\[
 \mathfrak{N}_{m,d}:=\bigcup_{l \in \N}\bigcup_{{h} \in \N^l}\mathfrak{N}_{m,d}^{l,{h}}.
\]
We further restrict to those neural networks with outputs that are componentwise bounded by a constant $B>0$ and $L$-Lipschitz, when restricted to $\Omega_i$, $i \in \N$.
\begin{equation}\label{eq_definition_set_neural_net_strategies}
\begin{aligned}
\mathfrak{N}_{i,B,L}:= \big\{f \in \mathfrak{N}_{id,d}~\big|~ &\pi_j\circ f|_{\Omega_i}\text{  is  }L \text{-Lipschitz} \\
&\text{and }\|\pi_j\circ f\|_{\infty,\Omega_i}\leq B\text{ for all } j=1,\dots,d\big\},
\end{aligned}
\end{equation}
where $\pi_j:\R^d \rightarrow \R$ denotes the canonical projection onto the $j$-th component, i.e., we have for every function $f:\R^m \rightarrow \R^d$ the representation $f(x)=(f_1(x),\dots,f_d(x))=(\pi_1\circ f(x),\dots,\pi_d\circ f(x))$, $x \in \R^m$.
One reason for the importance and popularity of deep neural networks is their universal approximation property (compare e.g., for the classical result \cite[Theorem 2]{hornik1991approximation}, \cite[Theorem 3.2]{kidger2020universal} for arbitrary input and output dimensions, and \cite[Theorem 1]{eckstein2020lipschitz} for neural networks which are Lipschitz continuous with respect to some pre-specified constant). We establish the following form of a universal approximation theorem tailored to our setting which is useful to approximate trading positions from $\HM$.
\begin{lem}\label{lem_universal_approx}
Let Assumption~\ref{asu_activation_function_phi} hold true.
Then, for all $i =1,\dots,n$, for all $L>0$ and for all $B>0$ the set $\mathfrak{N}_{i,B,L}|_{\Omega_i}$ is dense in 
\begin{equation}\label{eq_Lipschitz_M_class}
\left\{f:\Omega_i\rightarrow \R^d~\middle|~ \pi_j\circ f\text{  is  }L \text{-Lipschitz and } \| \pi_j\circ f\|_{\infty,\Omega_i} \leq B \text{ for all } j=1,\dots,d\right\}
\end{equation}
with respect to the uniform topology on $\Omega_i$.
\end{lem}
We remark that the requirement on the activation function formulated in Lemma~\ref{lem_universal_approx} is in particular fulfilled by those activation functions mentioned in Remark~\ref{rem_activation_functions}.
The reduction of $\HM$ to those strategies that can be represented by neural networks leads to the set
\begin{equation*}
\HM^{\mathcal{N}\mathcal{N}}:=\left\{h_{c,\Delta} \in \HM ~\middle|~ \left(\Delta_{i}^1,\dots,\Delta_{i}^d\right) \in \mathfrak{N}_{i,B,L}\text{ for all } i=1,\dots,n-1\right\}.
\end{equation*}
The associated functional is then defined accordingly by
\begin{equation*}
\Gamma_{B,L,k}^{\mathcal{N}\mathcal{N}}(\Phi,\G):=\inf_{h_{c,\Delta}  \in \HM^{\mathcal{N}\mathcal{N}}} \bigg\{c~+ k \sum_{\PP \in \mathcal{P}} \int_{\Omega} \beta\big(\E_{\PP}[\Phi(S)-h_{c,\Delta}(S) ~|~ \G]\big) \D \PP \bigg\},
\end{equation*}
which relies on the same penalization as $\Gamma_{B,L,k}$ expressed through the function $\beta$.
The following result employing Lemma~\ref{lem_universal_approx} establishes that $\Gamma_{B,L,k}$ can indeed be approximated arbitrarily well by appropriate neural networks.
\begin{prop}\label{prop_neuralnetworks} 
Let Assumptions~\ref{asu_transaction_costs}, \ref{asu_P}, \ref{asu_beta}, \ref{asu_activation_function_phi} hold true. Let $L > 0$, $B>0$, and let $\Phi:\Omega\rightarrow \R$ be Borel measurable with $\|\Phi\|_{\infty,\Omega} \leq B$. \\
Then  for every $\sigma$-algebra $\G \subseteq\sigma(S)$ on $\Omega$ and for all $k\in \N$ we have
\[
\Gamma_{B,L,k}^{\mathcal{N}\mathcal{N}}(\Phi,\G) = \Gamma_{B,L,k}(\Phi,\G).
\]
\end{prop}

\subsection{Computation of $\mathcal{P}$-robust statistical arbitrage strategies with neural networks}
We now focus on the computation of \emph{$\mathcal{P}$-robust statistical arbitrage strategies}. These are $\mathcal{P}$-robust $\G$-arbitrage strategies with the specific choice $\G:=\sigma(S_{t_n})$. These strategies have a particular importance as they can be interpreted as \emph{profitable trading strategies on average given any terminal value}, compare also \cite{bondarenko2003statistical},~\cite{kassberger2017additive},~\cite{lutkebohmert2021robust}, and \cite{rein2021generalized}. As $\sigma(S_{t_n})$ contains an infinite number of sets, it is a priori unclear how to compute numerically a conditional expectation with respect to $\sigma(S_{t_n})$. However, recall that conditional expectations with respect to any $\sigma$-algebra $\mathcal{F}$ generated by a finite partition can be computed efficiently.\footnote{Recall that if $\G=\sigma(\mathcal{E})$ is generated by a finite partition $\mathcal{E}=(\mathcal{F}_i)_{i=1,\dots,N}$, for some $N \in \N$ with $\PP\left(\mathcal{F}_i \right)>0 $ for all $i=1,\dots,N$, then $\E_\PP[X~|~\G]=\sum_{i=1}^N \tfrac{\E_\PP[X \one_{\mathcal{F}_i}]}{\PP(\mathcal{F}_i)} \one_{\mathcal{F}_i}$ for all $\PP$-integrable random variables $X$.}

Thus, with Pseudo-Algorithm~\ref{algo_generation_sigma} we establish a routine enabling the generation of \emph{random} sets from the $\sigma$-algebra $\sigma(S_{t_n})$.
To this end, we consider for every $i \in \N$ the finite partition $\mathcal{E}_i$ of $[\underline{K}^1,\overline{K}^1]\times \cdots \times [\underline{K}^d,\overline{K}^d]$ generated according to \eqref{defn_E_i} in Pseudo-Algorithm~\ref{algo_generation_sigma}. 

\begin{algorithm}[h!]
\SetAlgorithmName{Pseudo-Algorithm}{Pseudo-Algorithm}{Pseudo-Algorithm}
\SetAlgoLined
\SetKwInOut{Input}{Input}
\SetKwInOut{Output}{Output}

\Input{The bounds $\underline{K}^j<\overline{K}^j$, $j =1,\dots,d$ of possible stock prices;}
\Output{For each $i\in \N$ a randomly generated partition $\mathcal{E}_i$ of $[\underline{K}^1,\overline{K}^1]\times \cdots \times [\underline{K}^d,\overline{K}^d]$ containing $2^i$ sets;}
We set $\mathcal{E}_0=\{\emptyset,\Omega\}$;\\
\For {$i\in \N$}{
 Generate a random set
 \begin{equation}\label{eq_pseudo_algo_A_i_generation}
 A_i:=(a_1^{(i)},b_1^{(i)}]\times \cdots \times (a_d^{(i)}, b_d^{(i)}] \subset [\underline{K}^1,\overline{K}^1]\times \cdots \times [\underline{K}^d,\overline{K}^d];
 \end{equation}
 }
  \For {$i\in \N$}{
 Define 
 \begin{equation}\label{defn_E_i}
 \mathcal{E}_i:=\left\{E \cap A_i, E \cap A_i^{C} \text{ for all } E \in \mathcal{E}_{i-1}\right\},
\end{equation}
where $A_i^C:=[\underline{K}^1,\overline{K}^1]\times \cdots \times [\underline{K}^d,\overline{K}^d] \backslash A_i$;}

We obtain, by construction, that for each $i\in \N$ the set $\mathcal{E}_i$ contains (at most) $2^i$ sets and fulfils
\[
\bigcup_{E \in \mathcal{E}_i} \{E\}=[\underline{K}^1,\overline{K}^1]\times \cdots \times [\underline{K}^d,\overline{K}^d], \text{ and } E \cap F = \emptyset\text{ for all } E,F \in \mathcal{E}_i\text{ with } E \neq F.
\]
 \caption{Generation of a partition of $[\underline{K}^1,\overline{K}^1]\times \cdots \times [\underline{K}^d,\overline{K}^d]$.}\label{algo_generation_sigma}
\end{algorithm}
For any realization of the Pseudo-Algorithm~\ref{algo_generation_sigma}, we set
\begin{equation}\label{eq_defn_F_i}
\mathcal{F}_i:= \left\{S_{t_n}^{-1}(A)~\middle|~ A \in \sigma(\mathcal{E}_i)\right\},~i \in \N.
\end{equation}
One possible choice of distribution according to which the random sets 
$ A_i:=(a_1^{(i)},b_1^{(i)}]\times \cdots \times (a_d^{(i)}, b_d^{(i)}] \subset [\underline{K}^1,\overline{K}^1]\times \cdots \times [\underline{K}^d,\overline{K}^d]$ from \eqref{eq_pseudo_algo_A_i_generation} are generated is the following: For each $i \in \N$ and for all $j=1,\dots,d$, let 
\begin{equation}\label{defn_A_i}
 a_j^{(i)} \sim \mathcal{U}\left([\underline{K}^j,\overline{K}^j]\right), \qquad  b_j^{(i)}\equiv \overline{K}^j,
\end{equation}
with $ a_j^{(i)}$ being independent of  $a_k^{(l)}$ for $(j,i) \neq (k,l)$. Moreover, we denote by $\PP^U$ the distribution of $U:=\left(\left(a_j^{(i)},b_j^{(i)}\right)_{j=1,\dots,d}\right)_{i \in \N}$.

The following proposition then shows that for $\PP^U$- almost all realizations of Pseudo-Algorithm~\ref{algo_generation_sigma}, conditional expectations with respect to the sigma-algebra $\sigma(S_{t_n})$ can be approximated arbitrarily well by conditional expectations with respect to $\mathcal{F}_i$, defined in \eqref{eq_defn_F_i}, for large enough $i \in \N_0$. Since each $\mathcal{F}_i$ is generated by a finite partition, this allows to compute the corresponding conditional expectation efficiently.
\begin{prop}\label{prop_filtration}
Let $L>0$, $B>0$, let $\Phi:\Omega\rightarrow \R$ be Borel measurable with $\|\Phi\|_{\infty,\Omega} \leq B$, and let $h_{c,\Delta} \in \HM$. Moreover, consider Pseudo-Algorithm~\ref{algo_generation_sigma},  where $a_j^{(i)},b_j^{(i)}$ are distributed for $i\in \N$, $j\in \{1,\dots,d\}$ according to $\PP^U$.
Then we have

\begin{itemize}
\item[(i)]$\sigma(S_{t_n})=\sigma \left(\bigcup_{i=1}^{\infty}\mathcal{F}_{i}\right)~~\PP^U$-a.s.\footnote{More precisely, that $\sigma(S_{t_n})=\sigma \left(\bigcup_{i=1}^{\infty}\mathcal{F}_{i}\right)$ holds for $\PP^U$-almost all realizations of $(\mathcal{E}_i)_{i \in \N}$ following \eqref{defn_E_i} and \eqref{defn_A_i}.}
\item[(ii)] Let $\sigma(S_{t_n})=\sigma \left(\bigcup_{i=1}^{\infty}\mathcal{F}_{i}\right)$ hold, then for all $\PP\in\mathcal{P}$ we have that
\[
\lim_{i \rightarrow \infty}\E_{\PP}[\Phi(S)-h_{c,\Delta}(S)~|~\mathcal{F}_i] = \E_{\PP}[\Phi(S)-h_{c,\Delta}(S)~|~\sigma(S_{t_n})],
\]
where the convergence holds both $\PP$-almost surely and in $L^1(\PP)$.
\end{itemize}
\end{prop}

We build on Proposition~\ref{prop_convergence}, Proposition~\ref{prop_neuralnetworks}, and Proposition~\ref{prop_filtration} to establish the following main theorem.

\begin{thm}\label{thm_summary}
Let Assumptions~\ref{asu_transaction_costs}, \ref{asu_P}, \ref{asu_beta}, \ref{asu_activation_function_phi} hold true. Further, assume that $\beta$ is convex. \\
Then  for all $B>0$, for all $L>0$, and for all $\Phi:\Omega\rightarrow \R$ Borel measurable with $\|\Phi\|_{\infty,\Omega} \leq B$, the following equality holds $\PP^U$-almost surely:
\begin{equation}\label{eq_approx}
\Gamma_{B,L}(\Phi,\sigma(S_{t_n}))= \lim_{k\rightarrow\infty}\lim_{i \rightarrow \infty}\Gamma_{B,L,k}^{\mathcal{N}\mathcal{N}}(\Phi,\mathcal{F}_i).
\end{equation}
\end{thm}
In particular, in view of Equation~\eqref{eq_approx} we have for large enough $k$ and $i$ that 
$
\Gamma_{B,L}(\Phi,\sigma(S_{t_n})) \approx \Gamma_{B,L,k}^{\mathcal{N}\mathcal{N}}(\Phi,\mathcal{F}_i),
$
where $\Gamma_{B,L,k}^{\mathcal{N}\mathcal{N}}(\Phi,\mathcal{F}_i)$ is an optimization problem involving neural networks which can be solved efficiently. We will apply this approximation on real-world data in the next section.

\begin{rem}\label{rem_partition}
\begin{itemize}
\item[(i)]
Proposition~\ref{prop_filtration} and Theorem~\ref{thm_summary} rely on the construction from Pseudo-Algorithm~\ref{algo_generation_sigma} that allows to generate random partitions $(\mathcal{E}_i)_{i \in \N}$. Note that while \eqref{eq_pseudo_algo_A_i_generation} and \eqref{defn_A_i} are numerically tractable, the creation of sets involving complements as in \eqref{defn_E_i} grows exponentially with the dimension and is therefore not tractable in higher dimensions. For this reason, we call this methodology a pseudo-algorithm and provide with Algorithm~\ref{algo_generation_indicator_matrix} a numerical feasible methodology to generate indicator functions of sets of the corresponding partition $(\mathcal{E}_i)_{i \in \N}$ avoiding the explicit creation of complements. As it turns out in Algorithm~\ref{algo_training_nn}, considering indicator functions instead of computing explicit sets  is sufficient to compute expressions of the form 
\[
\E_{\PP}[\Phi(S)-h_{c,\Delta}(S)~|~\mathcal{F}_i].
\]
Hence, we consider Pseudo-Algorithm~\ref{algo_generation_sigma} as a theoretical routine, while Algorithm~\ref{algo_generation_indicator_matrix} and Algorithm~\ref{algo_training_nn} can be used in practice to compute $\mathcal{P}$-robust statistical arbitrage strategies.
\item[(ii)]
Proposition~\ref{prop_filtration} and Theorem~\ref{thm_summary} remain still valid  if we replace $(\mathcal{E}_i)_{i \in \N}$ by any family $(\widetilde{\mathcal{E}}_i)_{i \in \N}$, for which each $\widetilde{\mathcal{E}}_i$ forms a partition of $[\underline{K}^1,\overline{K}^1]\times \cdots \times [\underline{K}^d,\overline{K}^d]$ such that $\widetilde{\mathcal{E}}_i \subseteq \widetilde{\mathcal{E}}_{i+1}$ for each $i \in \N$ and $\sigma\left(S_{t_n}\right)= \sigma \left(\cup_{i\in \N} \widetilde{\mathcal{E}}_i\right)$. This is satisfied, e.g., by the (deterministic) dyadic partition of $[\underline{K}^1,\overline{K}^1]\times \cdots \times [\underline{K}^d,\overline{K}^d]$.
For our numerical experiments we decided however to use the partition $(\mathcal{E}_i)_{i \in \N}$ from Pseudo-Algorithm~\ref{algo_generation_sigma}, mainly because Algorithm~\ref{algo_generation_indicator_matrix} admits a computationally efficient algorithm to construct indicators of sets in $\mathcal{E}_i$.
\item[(iii)] If the sigma-algebra $\mathcal{G}\subseteq \sigma(S)$ is countably generated, i.e., $\mathcal{G} = \sigma\left(\cup_{i\in \N} A_i\right)$, then a similar procedure as in Proposition~\ref{prop_filtration} and Theorem~\ref{thm_summary}, where the case $\mathcal{G}=\sigma(S_{t_n})$ is considered, can be obtained, by generating for each $i \in \N$ a finite partition $\mathcal{F}_1^i,\dots,\mathcal{F}_{N_i}^i$ such that $\sigma\left(\mathcal{F}_j^i,j \leq N_i\right) = \sigma\left(A_j, j \leq i\right) $ for all $i \in \N$. However, note that for $\mathcal{G} \neq \sigma\left(S_{t_n}\right)$ it is not guaranteed that one can efficiently generate elements $A_{i_1},\dots,A_{i_\mathfrak{B}} \subseteq (A_i)_{i\in \N}$ as in \eqref{eq_pseudo_algo_A_i_generation} of Pseudo-Algorithm~\ref{algo_generation_sigma}. However, if it is possible to do so, one can apply the same Algorithm~\ref{algo_generation_indicator_matrix} to efficiently compute indicators of the partition obtained via \eqref{defn_E_i}.

\end{itemize}
\end{rem}
\section{Determination of the ambiguity set and the numerical algorithm}\label{sec_application}
We next use the approximation $
\Gamma_{B,L}(\Phi,\sigma(S_{t_n})) \approx \Gamma_{B,L,k}^{\mathcal{N}\mathcal{N}}(\Phi,\mathcal{F}_i),
$ for large $k,i \in \N$ justified by Theorem~\ref{thm_summary} and introduce an approach enabling to detect $\mathcal{P}$-robust statistical arbitrage strategies on financial markets. First, in Section~\ref{sec_ambiguity_set_p}, we propose an approach to determine an entirely data-driven ambiguity set $\mathcal{P}$, then, in Section~\ref{sec_numerical_algorithm}, we provide a neural network-based numerical algorithm.

\subsection{Determination of an ambiguity set of physical measures}\label{sec_ambiguity_set_p}
We assume that we are able to observe a time series 
\[
Y:=(Y_{s_i})_{i=1,\dots,N} \text{ with } Y_{s_i}=\left(Y_{s_i}^j\right)_{j=1,\dots,d} \in (0,\infty)^d\text{ for all } i=1,\dots,N
\]
consisting of $N\in \N$ observed past prices of each of the underlying $d$ securities on which we rely our trading strategy on. The past observation dates are denoted by $s_1,\dots,s_N$ with $s_N \leq t_0$, compare also Figure~\ref{fig_dates}.
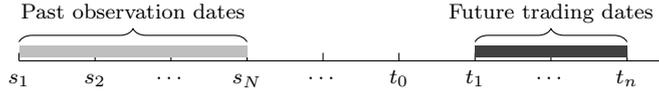
\begin{figure}[h!]
\begin{center}
\begin{tikzpicture}[%
    every node/.style={
        font=\scriptsize,
        text height=1ex,
        text depth=.25ex,
    },
]
\draw[->] (0,0) -- (8.5,0);

\foreach \x in {0,1,...,8}{
    \draw (\x cm,3pt) -- (\x cm,0pt);
}

\node[anchor=north] at (0,0) {$s_1$};
\node[anchor=north] at (1,0) {$s_2$};
\node[anchor=north] at (2,0) {$\cdots$};
\node[anchor=north] at (3,0) {$s_N$};
\node[anchor=north] at (4,0) {$\cdots$};
\node[anchor=north] at (5,0) {$t_0$};
\node[anchor=north] at (6,0) {$t_1$};
\node[anchor=north] at (7,0) {$\cdots$};
\node[anchor=north] at (8,0) {$t_n$};

\fill[lightgray] (0,0.05) rectangle (3,0.2);
\fill[darkgray] (6,0.05) rectangle (8,0.2);

\draw[decorate,decoration={brace,amplitude=5pt}] (0,0.25) -- (3,0.25)
    node[anchor=south,midway,above=4pt] {Past observation dates};
\draw[decorate,decoration={brace,amplitude=5pt}] (6,0.25) -- (8,0.25)
    node[anchor=south,midway,above=4pt] {Future trading dates};
\end{tikzpicture}
\end{center}
\caption{Illustration of the observation dates $(s_i)_{i=1,\dots,N}$ and the future dates $(t_i)_{i=0,\dots,n-1}$.}\label{fig_dates}
\end{figure}

Based on the time series $Y$, our aim is to construct an entirely data-driven ambiguity set $\mathcal{P}$ such that we are able to detect $\mathcal{P}$-robust statistical arbitrage strategies which turn out to be profitable. Recall that we aim to determine a trading strategy which trades at future dates $(t_i)_{i=0,\dots,n-1}$. Assume for sake of simplicity that 
\[
t_{i+1}-t_{i} = s_{j+1}-s_j\text{ for all } i=0,\dots,n-1,~~~ j = 1,\dots,N-1.
\]
and that $N > n$. Further, assume that at time $t_0$ we observe the spot values
\[
S_{t_0}:=(S_{t_0}^1,\dots,S_{t_0}^d)\in (0,\infty)^d.
\]
This allows us to determine the empirical measure $\widehat{\PP} \in \mathcal{M}_1(\R^{nd})$, relying on the observations from $Y$ and on the current spot values $S_{t_0}$, through
\begin{equation}\label{eq_definition_dirac_measure}
\widehat{\PP}:= \frac{1}{N-n}\sum_{\ell =1}^{N-n} \delta_{\left(S_{t_0}\cdot \frac{Y_{s_{\ell+1}}}{Y_{s_\ell}},\cdots,S_{t_0}\cdot \frac{Y_{s_{\ell+n}}}{Y_{s_\ell}}\right)},
\end{equation}
where $\delta_x$ denotes the Dirac measure with unit mass at $x\in \R^{nd}$, and where the multiplication with $S_{t_0}$ as well as division with respect to $Y_{s_i}$ is meant component-wise. This means, according to \eqref{eq_definition_dirac_measure}, the measure $\widehat{\PP}$ assigns equal probability to all $N-n$ scaled paths of length $n$ that were observed at times $(s_\ell)_{\ell=1,\dots,N-n}$ to occur in the future, compare also the left and middle panel of Figure~\ref{fig_illustration_empirical_measure}. Note that the construction in \eqref{eq_definition_dirac_measure} ensures particularly that the future paths evolve according to observed paths while the scaling ensures that the spot value of all considered future paths coincides with $S_{t_0}$.
\begin{figure}
\begin{center}
    {
    \begin{tikzpicture}[xscale=.6,yscale=.6]
        \draw[black,thick] (0,1) -- (1,0.8);
        \draw[orange,thick](1,0.8)--(2,3);
        \draw[purple,thick](2,3)--(3,3.5);
        \draw[olive,thick](3,3.5)--(4,3.6);
        \draw[pink,thick](4,3.6)--(5,2);
        \draw[teal,thick](5,2)--(6,5); 
        \filldraw(0,1)circle(1pt);
        \filldraw(1,0.8)circle(1pt);
        \filldraw(2,3)circle(1pt);
        \filldraw(3,3.5)circle(1pt);
        \filldraw(4,3.6)circle(1pt);
        \filldraw(5,2)circle(1pt);
        \filldraw(6,5)circle(1pt);              
        \draw [->] (0,0) -- (6,0);
        \draw [->] (0,0) -- (0,6)node[anchor=east]{$Y_{s_i}$};
        \draw(0,1pt)--(0,-1pt)node[anchor=north] {$s_{1}$};
        \draw(1,1pt)--(1,-1pt)node[anchor=north] {$s_{2}$};
        \draw(2,1pt)--(2,-1pt)node[anchor=north] {$s_{3}$};
        \draw(3,1pt)--(3,-1pt)node[anchor=north] {$s_{4}$};
        \draw(4,1pt)--(4,-1pt)node[anchor=north] {$s_{5}$};
        \draw(5,1pt)--(5,-1pt)node[anchor=north] {$s_{6}$};
        \draw(6,1pt)--(6,-1pt)node[anchor=north] {$s_{7}$};
        
		\fill[lightgray] (0,0.05) rectangle (6,0.2);    
    \end{tikzpicture}
    }    
    {
    \begin{tikzpicture}[xscale=.6,yscale=.6]
    \draw [->] (0,0) -- (3,0);
    \draw [->] (0,0) -- (0,6)node[anchor=east]{$S_{t_i}$};
    \draw(0,1pt)--(0,-1pt)node[anchor=north] {$t_{0}$};
	\draw(1,1pt)--(1,-1pt)node[anchor=north] {$t_{1}$};
	\draw(2,1pt)--(2,-1pt)node[anchor=north] {$t_{2}$};
        \draw[black,thick] (0,3) -- (1,2.8);
        \draw[orange,thick](1,2.8)--(2,5);

        \draw[orange,thick](0,3)--(1,5.2);
        \draw[purple,thick](1,5.2)--(2,5.7);
	
        \draw[purple,thick](0,3)--(1,3.5);
        \draw[olive,thick](1,3.5)--(2,3.6);	
        
        \draw[olive,thick](0,3)--(1,3.1);
        \draw[pink,thick](1,3.1)--(2,1.5);	

        \draw[pink,thick](0,3)--(1,1.4);
        \draw[teal,thick](1,1.4)--(2,4.4); 
                
        \filldraw(0,3)circle(1pt);
        \node[anchor=east] at (0,3){$S_{t_0}$};
        \filldraw(1,2.8)circle(1pt);
        \filldraw(2,5)circle(1pt);
        \filldraw(1,5.2)circle(1pt);
        \filldraw(2,5.7)circle(1pt);
        \filldraw(1,3.5)circle(1pt);
        \filldraw(2,3.6)circle(1pt);     
        \filldraw(1,3.1)circle(1pt);
        \filldraw(2,1.5)circle(1pt);
        \filldraw(1,1.4)circle(1pt);
        \filldraw(2,4.4)circle(1pt); 
        
        \fill[darkgray] (1,0.05) rectangle (2,0.2); 
    \end{tikzpicture}
    }
 {
    \begin{tikzpicture}[xscale=.6,yscale=.6]
    \draw [->] (0,0) -- (3,0);
    \draw [->] (0,0) -- (0,6)node[anchor=east]{$S_{t_i}$};
    \draw(0,1pt)--(0,-1pt)node[anchor=north] {$t_{0}$};
	\draw(1,1pt)--(1,-1pt)node[anchor=north] {$t_{1}$};
	\draw(2,1pt)--(2,-1pt)node[anchor=north] {$t_{2}$};

        \draw[black,very thick,dotted] (0,3) -- (1,3.2);
        \draw[orange,very thick,dotted](1,3.2)--(2,5.1);
        
        \draw[orange,very thick,dotted](0,3)--(1,5.8);
        \draw[purple,very thick,dotted](1,5.8)--(2,5.6);        
		
        \draw[purple,very thick,dotted](0,3)--(1,3.7);
        \draw[olive,very thick,dotted](1,3.7)--(2,3.2);	
        
        \draw[olive,very thick,dotted](0,3)--(1,3.4);
        \draw[pink,very thick,dotted](1,3.4)--(2,1.2);	

        \draw[pink,very thick,dotted](0,3)--(1,0.8);
        \draw[teal,very thick,dotted](1,0.8)--(2,4.7); 
                
        \filldraw(0,3)circle(1pt);
        \node[anchor=east] at (0,3){$S_{t_0}$};
        \filldraw(1,3.2)circle(1pt);
        \filldraw(2,5.1)circle(1pt);
        \filldraw(1,5.8)circle(1pt);
        \filldraw(2,5.6)circle(1pt);
        \filldraw(1,3.7)circle(1pt);
        \filldraw(2,3.2)circle(1pt);     
        \filldraw(1,3.4)circle(1pt);
        \filldraw(2,1.2)circle(1pt);
        \filldraw(1,0.8)circle(1pt);
        \filldraw(2,4.7)circle(1pt); 
        
         \fill[darkgray] (1,0.05) rectangle (2,0.2); 
    \end{tikzpicture}
    }
\end{center}
\caption{The figure illustrates how the measure $\widehat{\PP}$ and $\widehat{\PP}_{\tau}$, as defined in \eqref{eq_definition_dirac_measure} and \eqref{eq_defn_p_tau} respectively, are constructed. \\
\underline{Left:} An exemplary observed path $(Y_{s_i})_{i=1,\dots,7}$ is displayed. \underline{Centre:} In the case $n=2$ the future paths with positive probability under $\widehat{\PP}$ are shown. To each of the paths equal probability is assigned. \underline{Right:} A possible realization of $\widehat{\PP}_{\tau}$ is depicted, the paths from $\widehat{\PP}$ are deviated according to a normally distributed realization of $\tau$.}\label{fig_illustration_empirical_measure}
\end{figure}
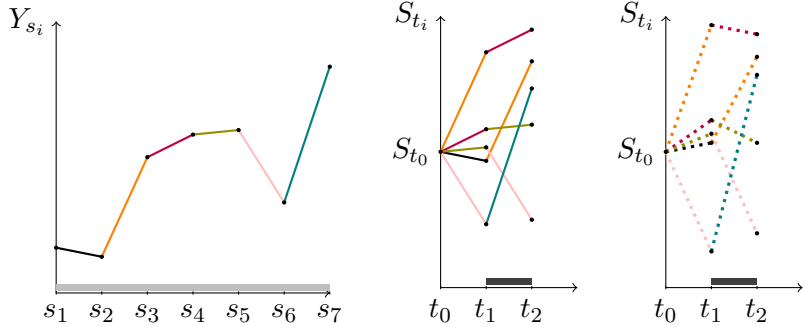 

The empirical measure $\widehat{\PP}$ relies entirely on past data and on the spot value $S_{t_0}$, whereas for the future development of the securities we want to take into account uncertainty with respect to the choice of the underlying probability measure, but still consider only such measures which are in a certain sense close to the observed data. To measure the distance between two probability measures $\PP, \Q \in \mathcal{M}_1(\R^{nd})$, we  therefore introduce the $1$-Wasserstein-distance\footnote{Alternatively, one could have used, e.g., the Kullback-Leibler divergence (\cite{kullback1951information}) to measure the distance between two measures. However, the Kullback-Leibler divergence requires absolute continuity of the considered measures, which means in our case that considered measures would have been restricted to the support determined by the observed prices $(Y_{s_i})_{i\in \N}$. The $1$-Wasserstein distance does not suffer from this restriction.} which is defined as
\[
\mathcal{W}(\PP,\Q):=\inf_{\pi \in \Pi(\PP,\Q)}\int_{\R^{nd} \times \R^{nd}} \| u-v \|_{nd} \,\D\pi(u,v),
\]
with $\Pi(\PP,\Q)\subset \mathcal{M}_1\left(\R^{nd} \times \R^{nd}\right)$ denoting the set of all joint distributions of $\PP$ and $\Q$, compare e.g.\,\cite{ambrosio2003lecture,ruschen2007monge,villani2008optimal} for more details on the Wasserstein distance and the related topic of optimal transport. 
Let $\varepsilon>0$, then we consider some ambiguity set $\mathcal{P}\subset \mathcal{M}_1(\R^{nd} )$ such that the therein contained measures lie within the Wasserstein ball around  $\widehat{\PP}$ with radius $\varepsilon>0$, i.e.,
\[
\mathcal{P} \subset \mathcal{B}_\varepsilon(\widehat{\PP}):=\left\{\PP\in \mathcal{M}_1(\R^{nd})~\middle|~ \mathcal{W}(\widehat{\PP},\PP)<\varepsilon\right\}.
\]
In that way each $\PP \in \mathcal{P}$ can be considered to be close to the empirical measure $\widehat{\PP}$. Next, we introduce for some $\tau = \left(\tau_1,\dots,\tau_{N-n}\right) \in \R^{nd(N-n)}$ the perturbed measure
\begin{equation}\label{eq_defn_p_tau}
\widehat{\PP}_\tau:=\frac{1}{N-n}\sum_{\ell=1}^{N-n} \delta_{\left(\left(S_{t_0}\cdot \frac{Y_{s_{\ell+1}}}{Y_{s_\ell}},\cdots,S_{t_0}\cdot \frac{Y_{s_{\ell+n}}}{Y_{s_\ell}}\right)+\tau_\ell\right)},
\end{equation}
compare also the right panel of Figure~\ref{fig_illustration_empirical_measure}.
The following lemma asserts that choosing $\tau$ appropriately ensures that the measure $\widehat{\PP}_\tau$ is indeed contained in $\mathcal{B}_\varepsilon(\widehat{\PP})$.

\begin{lem}\label{lem_wasserstein}
Let $\varepsilon>0$ and let $\tau = \left(\tau_1,\dots,\tau_{N-n}\right) \in \R^{nd(N-n)}$ such that\footnote{Here, and in the rest of the paper $\|\cdot \|_{a}$ denotes the Euclidean norm on $\R^{a}$ for any $a \in \mathbb{N}$.}
\begin{equation}\label{eq_ineq_tau_epsilon}
\max _{\ell =1, \dots, N-n} \|\tau_\ell\|_{nd} < \varepsilon.
\end{equation}
Moreover, let $\widehat{\PP}$ be defined by \eqref{eq_definition_dirac_measure}, and $\widehat{\PP}_\tau$ by \eqref{eq_defn_p_tau}.
Then we have that
\[
\widehat{\PP}_\tau\in \mathcal{B}_\varepsilon(\widehat{\PP}).
\]
\end{lem}

Then, based on Lemma~\ref{lem_wasserstein}, we are able to determine a set $\mathcal{P}\subset \mathcal{B}_\varepsilon(\widehat{\PP})$ which contains a finite number of $N_{\operatorname{measures}}\in \N$. We summarize this procedure in Algorithm~\ref{algo_generation_P}.

\begin{rem}\label{rem_choice_Omega}
Given a time series $Y$ with $\left(Y_{s_i}^j\right)_{j=1,\dots,d} \in (0,\infty)^d$ for all $i=1,\dots,N$ and spot prices $S_{t_0}\in (0,\infty)^d$, we can construct $\Omega={[\underline{K}^1,\overline{K}^1]}^n\times \cdots \times {[\underline{K}^d,\overline{K}^d]}^n$ the following way. Consider for $j=1,\dots,d$
\begin{equation}\label{eq_s_underline_overline}
\underline{S}^j:=\min_{\substack{\ell =1,\dots,N-n\\ i=0,\dots,n}} S_{t_0}^j \cdot \frac{Y_{s_{\ell+i}}^j}{Y_{s_\ell}^j} ,\qquad \overline{S}^j:=\max_{\substack{\ell =1,\dots,N-n\\ i=0,\dots,n}} S_{t_0}^j \cdot \frac{Y_{s_{\ell+i}}^j}{Y_{s_\ell}^j}.
\end{equation}
Based on the above bounds $(\underline{S}^j)_{j=1,\dots,d}$ and $(\overline{S}^j)_{j=1,\dots,d}$, we define, for some $\delta>0$, the underlying space   by setting
\begin{equation}
\label{eq_definition_omega_empirical}
\underline{K}^j:= \underline{S}^j-\delta,~~\overline{K}^j:= \overline{S}^j+\delta,~~\text{ for }j=1,\dots,d.
\end{equation}

Then, we see that $\widehat{\PP}\in \mathcal{M}_1(\Omega)$ and for each $0<\varepsilon\leq  \delta$ and $\tau \in \R^{nd(N-n)}$ satisfying \eqref{eq_ineq_tau_epsilon} we have $\widehat{\PP}_{\tau} \in \mathcal{M}_1(\Omega)$.
\end{rem}

\begin{algorithm}[h!]
\SetAlgoLined
\SetKwInOut{Input}{Input}
\SetKwInOut{Output}{Output}

\Input{Time series $Y \in (0,\infty)^{dN}$;  Spot prices $S_{t_0} \in(0,\infty)^d$; Radius of the Wasserstein ball $\varepsilon>0$;  Amount of measures $N_{\operatorname{measures}} \in \N$ to generate;}
\Output{Set $\mathcal{P}\subset \mathcal{B}_\varepsilon(\widehat{\PP})$ of probability measures;}
\For {$m =1,\dots, N_{\operatorname{measures}}$}{ 
\begin{equation}\label{eq_generate_p}
\begin{aligned}
\tau_\ell^{(m)}:&=({\tau_{\ell}^{1,1}}^{(m)},\dots,{\tau_{\ell}^{n,d}}^{(m)}), \\
&\text{ with } {\tau_\ell^{k,j}}^{(m)} \sim \mathcal{N}(0,1) \text{ i.i.d. for all } j =1,\dots,d,~k=1,\dots,n,~\ell =1,\dots,N-n;\\
U_\varepsilon^{(m)}:&\sim \mathcal{U}((0,\varepsilon));\\
\widetilde{\tau}^{(m)}:&={ U_\varepsilon^{(m)}} \cdot \left(\frac{\tau_1^{(m)}}{\|\tau_1^{(m)}\|_{nd}}, \cdots, \frac{\tau_{N-n}^{(m)}}{\|\tau_{N-n}^{(m)}\|_{nd}} \right)\in \R^{nd(N-n)};\\
\widehat{\PP}_{\widetilde{\tau}^{(m)}}:&=\frac{1}{N-n}\sum_{\ell=1}^{N-n} \delta_{\left(\left(S_{t_0}\cdot \frac{Y_{s_{\ell+1}}}{Y_{s_\ell}},\cdots,S_{t_0}\cdot \frac{Y_{s_{\ell+n}}}{Y_{s_\ell}}\right)+\widetilde{\tau}_\ell^{(m)}\right)};
\end{aligned}
\end{equation}
}
Then, we have by construction $\max_{\ell =1,\dots,N-n} \|\widetilde{\tau}^{(m)}_\ell\|_{nd} < \varepsilon$ for all $m$ and thus, according to Lemma~\ref{lem_wasserstein}, it holds
\begin{equation}\label{eq_definition_set_p}
\mathcal{P}:= \left\{\widehat{\PP}_{\widetilde{\tau}^{(m)}}, m=1,\dots,N_{\operatorname{measures}} \right\}\subset \mathcal{B}_{\varepsilon}(\widehat{\PP}).
\end{equation}
 \caption{Generation of a set $\mathcal{P}\subset \mathcal{B}_{\varepsilon}(\widehat{\PP})$.}\label{algo_generation_P}
\end{algorithm}

\begin{algorithm}[h!]
\SetAlgoLined
\SetKwInOut{Input}{Input}
\SetKwInOut{Output}{Output}

\Input{
 Number of iterations $\mathfrak{B} \in \N$; Bounds $\underline{K}^j<\overline{K}^j$, $j=1,\dots,d$;
 Terminal prices $S_{t_n}=(S_{t_n}^1,\dots,S_{t_n}^d)\in  [\underline{K}^1,\overline{K}^1]\times \cdots \times [\underline{K}^d,\overline{K}^d]$;}
\Output{The indicators $(\one_{\left\{S_{t_n}\in \mathfrak{F}_b\right\}})_{b=1,...,2^{\mathfrak{B}}}$ corresponding to the finite partition $\mathcal{E}_\mathfrak{B}=\{\mathfrak{F}_1,\dots,\mathfrak{F}_{2^{\mathfrak{B}}}\}$ of $\Omega$ as defined in Pseudo-Algorithm~\ref{algo_generation_sigma};}
Initialize a vector $s=(s_1,\dots,s_\mathfrak{B})\in \R^{\mathfrak{B}}$;.\\
Initialize a vector $v\in \R^{2^\mathfrak{B}}$, with all entries equal to $0$;\\
\For {$i=1,\dots,\mathfrak{B}$}{
 Generate
 \[
 A_i=(a_1^{(i)},b_1^{(i)}]\times \cdots \times (a_d^{(i)}\times b_d^{(i)}] \subset [\underline{K}^1,\overline{K}^1]\times \cdots \times [\underline{K}^d,\overline{K}^d];
 \]
 }
\For {$i=1,\dots,\mathfrak{B}$}{
Set the $i$-th element of $s$ to 1 if $S_{t_n} \in A_i$ and 0 otherwise;
}
 Treat $s$ as a binary number and convert it into a decimal number $\widetilde{s}:= \sum_{i=1}^{\mathfrak{B}} s_i2^{i-1}$, with $0\leq \widetilde{s}\leq2^{\mathfrak{B}}-1$;\\
Set the $(\widetilde{s}+1)$-th entry of $v$ to 1;\\
 
Then, by construction, the indicator vector $(\one_{\left\{S_{t_n}\in \mathfrak{F}_b\right\}})_{b=1,...,2^{\mathfrak{B}}}$ coincides with $v$. 

 \caption{Generation of indicators corresponding to given partition and terminal prices.}\label{algo_generation_indicator_matrix}
\end{algorithm}

\subsection{The numerical algorithm to compute statistical arbitrage strategies}\label{sec_numerical_algorithm}
We rely on Theorem~\ref{thm_summary}, which ensures that $\Gamma_{B,L}(0,\sigma(S_{t_n})) \approx \Gamma_{B,L,k}^{\mathcal{N}\mathcal{N}}(0,\mathcal{F}_i)$ for sufficiently large values $i,k$, i.e., we choose $\Phi\equiv 0$ as discussed below Lemma~\ref{lem_compactness}. We remark that Pseudo-Algorithm~\ref{algo_generation_sigma} allows to generate finite partitions $(\mathcal{E}_i)_{i \in \N}$ of $[\underline{K}^1,\overline{K}^1]\times \cdots \times [\underline{K}^d,\overline{K}^d]$ contained in the $\sigma$-algebra $\sigma(S_{t_n})$. Based on this methodology, we propose a method specified in Algorithm~\ref{algo_generation_indicator_matrix} that allows to compute efficiently indicator functions of the form $\one_{\{S_{t_n} \in E\}}$ for $E \in \mathcal{E}_i$, $i \in \N$.
Note that this method provides an extremely efficient method that reduces the exponential complexity of the brute force approach to indicator computation to a linear time complexity.
Moreover, according to the considerations from Section~\ref{sec_ambiguity_set_p}, we propose the following Algorithm~\ref{algo_training_nn} to detect $\mathcal{P}$-robust statistical arbitrage strategies, where $\mathcal{P}$ is an ambiguity set defined as in \eqref{eq_definition_set_p}; see also Algorithm~\ref{algo_generation_P}.

\begin{algorithm}[h!]
\SetAlgoLined
\KwData{Time series $Y \in (0,\infty)^{dN}$; Spot values $S_{t_0}\in (0,\infty)^d$;}
\SetKwInOut{Input}{Input}
\SetKwInOut{Output}{Output}

\Input{Budget constraint $B>0$; Bounds $\underline{K}^j<\overline{K}^j, j=1,\dots,d$ to define $\Omega$; Penalization parameter $k\in \N$; Penalization function $\beta$; Hyper-parameters of the neural networks; Activation function $\varphi$; Number of iterations $N_{\operatorname{iter}}$; Number of measures $N_{\operatorname{measures}}$ contained in $\mathcal{P}$; The radius $\varepsilon>0$ of the Wasserstein-ball $\mathcal{B}_\varepsilon(\widehat{\PP})$;  Size of partition of $\Omega$: $2^\mathfrak{B}$ for some $\mathfrak{B}\in \N$; { Functions $\left({c_{\operatorname{trans}}}_i^j\right)_{i=0,\dots,n,\atop j=1,\dots,d}, \left({c_{\operatorname{spread}}}_i^j\right)_{i=0,\dots,n,\atop j=1,\dots,d}, \left({c_{\operatorname{short}}}_i^j\right)_{i=0,\dots,n,\atop j=1,\dots,d}$ to measure trading costs};}
\Output{Cash position $c\in \R$ of the strategy; \\
Trading positions $\left(\Delta_0^1,\dots,\Delta_0^d\right)\in [-B,B]^d$;\\  Trading strategies $\left(\Delta_i^1,\dots,\Delta_i^d\right)\in \mathfrak{N}_{id,B,L}$, $i=1,\dots,n-1$;}
Initialize parameter $c =0$;\\
Initialize parameter $\left(\Delta_0^1,\dots,\Delta_0^d\right)=(0,\dots,0)$;\\
\For{$i=1,\dots,n-1$}{
Initalize the parameters of the neural networks $\left(\Delta_i^1,\dots,\Delta_i^d\right)\in \mathfrak{N}_{id,B,L}$;
}
\For{$\operatorname{iter} =1,\dots,N_{\operatorname{iter}}$}{
Generate $(\tau^{(m)}_\ell)_{m=1,\dots,N_{\operatorname{measures}} \atop \ell =1,\dots,N-n}$, $\left(U_\varepsilon^{(m)}\right)_{m=1,\dots,N_{\operatorname{measures}}}$, and $\mathcal{P}=\left\{\widehat{\PP}_{\widetilde{\tau}^{(m)}}, m=1,\dots,N_{\operatorname{measures}} \right\} \subset \mathcal{B}_\varepsilon(\widehat{\PP})$ according\footnotemark   to Algorithm~\ref{algo_generation_P};\\
Set  $\widetilde{\tau}_{\ell,n}^{(m)}:=  U_\varepsilon^{(m)} \cdot \left(\frac{{\tau_\ell^{n,1}}^{(m)}}{\|{\tau_\ell}^{(m)}\|_{nd}},\cdots,\frac{{\tau_\ell^{n,d}}^{(m)}}{\|\tau_\ell^{(m)}\|_{nd}}\right)$ for $m=1,\dots,N_{\operatorname{measures}}$, $\ell=1,\dots,N-n$;\\
Set $\widetilde{S}_\ell^{(m)}:=\left(S_{t_0}\cdot \frac{Y_{s_{\ell+1}}}{Y_{s_\ell}},\dots,S_{t_0}\cdot \frac{Y_{s_{\ell+n}}}{Y_{s_\ell}}\right)+\widetilde{\tau}^{(m)}_\ell$ for $m=1,\dots,N_{\operatorname{measures}}$, $\ell=1,\dots,N-n$;\\
Set $\widetilde{S}_{\ell,n}^{(m)}:={S_{t_0}}\cdot \frac{Y_{s_{\ell+n}}}{Y_{s_{\ell}}}+\widetilde{\tau}_{\ell,n}^{(m)}$ for $m=1,\dots,N_{\operatorname{measures}}$, $\ell=1,\dots,N-n$;\\
Calculate, according to Algorithm~\ref{algo_generation_indicator_matrix}, the indicators $(\one_{\left\{\widetilde{S}_{l,n}^{(m)}\in \mathfrak{F}_b\right\}})_{b=1,...,2^{\mathfrak{B}}}$ for $m=1,\dots,N_{\operatorname{measures}}$, $\ell=1,\dots,N-n$;

Apply stochastic gradient descent / back-propagation to minimize the following\footnotemark
\begin{equation}\label{eq_algo_nn_equation_minimization}
\begin{aligned}
&c~+ k \sum_{\PP \in \mathcal{P}} \int \beta\left(\sum_{b=1}^{2^\mathfrak{B}}  \frac{\E_{\PP}\left[-h_{c,\Delta}(S)\one_{\mathfrak{F}_b}(S_{t_n})\right]}{\PP ({\mathfrak{F}_b})} \one_{\mathfrak{F}_b}(S_{t_n})\right) \D \PP \\
=&c+k \sum_{m=1}^{N_{\operatorname{measures}}} \sum_{i=1}^{N-n} \frac{1}{N-n} \cdot \beta \left(\sum_{b=1}^{2^\mathfrak{B}}\frac{\frac{1}{N-n}\displaystyle{\sum\limits_{\ell=1}^{N-n}} -h_{c,\Delta}\left(\widetilde{S}_\ell^{(m)}\right)\one_{\left\{\widetilde{S}_{\ell,n}^{(m)} \in \mathfrak{F}_b\right\}}}{\frac{1}{N-n}\displaystyle{\sum\limits_{\ell=1}^{N-n}}\one_{\left\{\widetilde{S}_{\ell,n}^{(m)} \in \mathfrak{F}_b\right\}}} \one_{\left\{\widetilde{S}_{i,n}^{(m)}\in \mathfrak{F}_b\right\}}\right)
\end{aligned}
\end{equation}
with respect to the parameters of the neural networks $\left(\Delta_i^j\right)_{i=1,\dots,n-1}^{j=1,\dots,d}$ and $(\Delta_0^j)^{j=1,\dots,d}$, $c$;}

 \caption{Computation of  $\mathcal{P}$-robust statistical arbitrage strategies.}\label{algo_training_nn}
\end{algorithm}
\section{Real-world examples}\label{sec_real_world_examples}
\footnotetext[7]{While we use Algorithm~\ref{algo_generation_P}  to generate samples from the ambiguity set $\mathcal{P}\subseteq \mathcal{B}_\varepsilon(\widehat{\PP})$, any methodology creating samples from $\mathcal{P} \subseteq \mathcal{B}_\varepsilon(\widehat{\PP})$ could be applied here.}
\footnotetext{Note that we define $\tfrac{0}{0}: = 0$ to avoid numerical instabilities when minimizing the expression in \eqref{eq_algo_nn_equation_minimization}.}
In this section we show that the presented approach can be applied successfully to real-world-data. In particular, the experiments provide empirical evidence for the applicability of the determined set of physical measures in Section~\ref{sec_ambiguity_set_p} even in a multi-asset environment. Moreover, we highlight the extraordinary performance of our robust trading approach in crisis-periods. 

\subsection{Implementation}
To apply Algorithm~\ref{algo_training_nn}, we use the framework provided by \emph{PyTorch} (\cite{NEURIPS2019_9015}) and train neural networks with a Batch-normalization layer for the input and $3$ fully-connected layers with $32d$ neurons in the first layer, $64d$ neurons in the second layer, and $128d$ neurons in the third layer, where $d$ is the number of considered assets  (the coefficient $d$ varies across the provided examples). We consider $n=9$ future times\footnote{We chose $n=9$ since this number provides a numerically tractable quantity of trading days that can be taken into account while still allowing to consider a large number of stocks. This leads to a trading period of $10$ days (recall that there is no trading on the last day of the period) and indeed, empirical findings indicate that most of the autocorrelation of daily stock returns is contained in the most recent $10$ stock returns (compare \cite[Table 3.1]{ding1993long}). 
} $t_1,\dots, t_9$ and hence the network architecture consists of $n-1$ independent fully-connected subnetworks $\left(\Delta_{i}^1,\dots,\Delta_{i}^d\right) \in \mathfrak{N}_{id,B,L}$ for $i=1,
\dots,n-1$, whereas $\Delta_{0}^1,\dots,\Delta_{0}^d$ are constants which are trained, too. Each neural network $\left(\Delta_{i}^1,\dots,\Delta_{i}^d\right)$, $i=1,\dots,n-1$  takes the past realizations of the $d$ assets of the $i$ trading days after $t_0$ as inputs and outputs the trading positions of all assets at day $t_i$. The activation functions of all hidden layers are \emph{ReLU} functions. To only consider neural networks from the set $\mathfrak{N}_{m,B,L}$ for $m \in \N$ and, in particular, to bound the output of the neural network by the constant $B:=10$, we use in the notation of \eqref{eq_nn_function} for $\varphi_l$ a $\tanh$-activation function (which is bounded by the absolute value 1) and matrices $A_i$ with biases $b_i$ which are randomly initialized for $i=1,\dots,l-1$ (using standard initialization in \emph{PyTorch}), whereas $A_l: = \operatorname{diag}(B,\dots,B)$ and $b_l:= 0$ . A predetermined Lipschitz constant $L>0$ could be enforced through per-layer Batch-normalization (\cite{gouk2021regularisation}), weight restrictions (\cite{anil2019sorting}), or through regularization (\cite{petzka2017on}), however we decided to not use any of these methods specifically for this purpose and to allow, a priori, $L$ to be arbitrarily large.
 However, note that a posteriori $L<\infty$, since the parameters of the neural networks in practice remain bounded over the training period (see e.g., \cite[Fig.~4]{baes2021low}).
 Moreover, a large Lipschitz constant $L>0$ is preferred to not exclude (too many) tradings trategies, and in order to keep our approach as simple as possible it is sufficient to know about the existence of \emph{some} constant $L$ which can be arbitrarily large but finite.

 We apply Algorithm~\ref{algo_training_nn} with $N_{\operatorname{iter}}:= 100$ training iterations and choose a penalization parameter\footnote{Note that a small penalization parameter such as $k=1$ turns out to be sufficient in practice. In contrast, choosing $k$ too large may result in numerical instabilities as observed in \cite[Section 4.1]{henry2019martingale} regarding a similar penalization approach. Compare also the relatively small penalization parameters chosen in \cite[Section 4]{eckstein2019computation}.} of $k:=1$, a radius $\varepsilon:=d$ of the Wasserstein Ball around the empirical measure, $N_{\operatorname{measures}}:=5$, {either with zero transaction costs,} per share transaction costs\footnote{These numbers represent conservative upper estimates of transaction costs that apply in practice. Typically, transaction costs are even smaller. Compare, e.g., \cite{interactivebrokes} where the fees for trading in different markets are indicated. {Indeed, note that according to \cite{interactivebrokes}, trading  in stocks incurs transaction costs of not more than $0.0035 \$$ per share. Also note that the average stock price of the $50$ constituents of the S\&P~500 considered in Section~\ref{exa_large_number} is $71.54 \$$.
Choosing }{proportional transaction costs with $1~bp$ however leads with an average stock value of $71.54 \$ $ to average transactions costs of $71.54 \$\cdot 0.0001 = 0.007154\$ $ being significantly higher than $0.0035\$$. }} with {${\lambda_{\operatorname{trans}}}_i^j:=0.01$,
 or proportional transaction costs with ${\lambda_{\operatorname{trans}}}_i^j:=0.0001$. Moreover we consider bid-ask spreads with ${\lambda_{\operatorname{spread}}}_i^j:=0.0002$, daily borrowing costs with ${\lambda_{\operatorname{short}}}_i^j:=0.1/252$},  a penalization function $\beta(x):=\max\{x,0\}^2$ for all the examples below, and we  report a value of $\mathfrak{B}:=12$ leading to $2^{12}$ sets that are considered when applying Algorithm~\ref{algo_generation_indicator_matrix} in all of the following examples. 
Moreover, we use the bounds $\underline{K}^j:= \underline{S}^j-\delta,~~\overline{K}^j:= \overline{S}^j+\delta,~~\text{ for }j=1,\dots,d,$ as elaborated in Remark~\ref{rem_choice_Omega}, where we set $\delta: = \varepsilon = d$. Further, we adopt a learning rate of $1e^{-3}$ for low dimensional cases (1 or 2 assets) and a learning rate of $1e^{-4}$ for high dimensional cases (10 assets and more), as empirical tests have shown that choosing different learning rates for low and high dimensional cases, respectively, significantly improves the training speed. In the sequel (Table~\ref{tbl_profit_example_1},~\ref{tbl_profit_example_3},~\ref{tbl_profit_example_4},~\ref{tbl_profit_example_5}, ~\ref{tbl_profit_example_6},~\ref{tbl_profit_example_7},~\ref{tbl_profit_example_8}), we report the numerical results by taking the average of $50$ independent experiments to balance off the effect of random weight initialization. In each of these $50$ experiments we partition the testing period into consecutive sub-periods of length $n$ to which we apply our trained strategy.

 For the reader's convenience the applied \emph{Python}-codes are provided and can be found under \href{https://github.com/YINDAIYING/Deep-Robust-Statistical-Arbitrage}{https://github.com/YINDAIYING/Deep-Robust-Statistical-Arbitrage}.

\begin{rem}\label{rem_normalization}
As neural networks turn out to be extremely sensitive with respect to the scale of the input (compare e.g. \cite{ba2016layer, sola1997importance}), we decided in the implementation to normalize the inputs by scaling the stock prices of each asset separately such that the spot values of each asset equals to $100$, i.e. $S_{t_0}=(100,...,100)\in\mathbb{R}^d$ and therefore to consider a Wasserstein-ball around $\widehat{\PP}$ that takes these normalized spot values into account.
\end{rem}

\subsection{Comparison with the linear programming approach from \cite{lutkebohmert2021robust}} \label{comparison_lp}

In this example, we compare our approach, relying on Algorithm~\ref{algo_training_nn}, with the linear programming~(LP) approach  proposed in \cite[Section 5.2]{lutkebohmert2021robust}. Note that the LP-approach in \cite{lutkebohmert2021robust} is computationally tractable only in low dimensions, typically $d \leq 3$. To this end, we borrow the setting from \cite[Section 5.2]{lutkebohmert2021robust} and  also consider the price evolution of the \emph{EUROSTOXX 50} in a training period from January $1995$ until  August $2013$, and in a testing period  ranging from  September $2013$ until July $2018$, while assuming that $t_1-t_{0} = 11$, $t_2-t_1=10$, i.e., we face a trading period of approximately one month with one intermediate trade. We train strategies according to Algorithm~\ref{algo_training_nn} while\footnote{{Note that in this section only, to keep the results comparable with the LP-approach from \cite{lutkebohmert2021robust}, we decided to not include neither transaction costs, nor borrowing costs, nor bid-ask spreads.}} using different bounds $\underline{K}^1$ and $\overline{K}^1$ and then test the performance of the trained strategies for various values of $\underline{K}^1$ and $\overline{K}^1$. The performance measure is defined by the Sharpe ratio of the strategies applied to the testing period in dependence of $\overline{K}^1-\underline{K}^1$, averaged over all $50$ independent experiments. In Figure \ref{fig_NN_vs_LP}, the Sharpe ratio of the calculated strategy is plotted against the difference between the upper bound $\overline{K}^1$ and the lower bound $\underline{K}^1$ centered at the price of $100$. One can observe that the Sharpe ratio of the proposed approach is slightly higher than the Sharpe ratio of the LP-approach from \cite{lutkebohmert2021robust} once the bounds width exceeds a certain value (around $30$ in this example), reaches its highest value when the width equals a value of about $90$, and declines slightly when increasing the width of the bounds further.

As the best trading performance, i.e., the highest Sharpe ratio, can be observed when the asset values are assumed to be restricted by relatively small bounds, this provides strong evidence that the underlying setting, which assumes (large) bounds for stock prices, imposes no constraint for the detection of optimal trading strategies.


\begin{figure}[h!]
\begin{center}
\includegraphics[width=10cm]{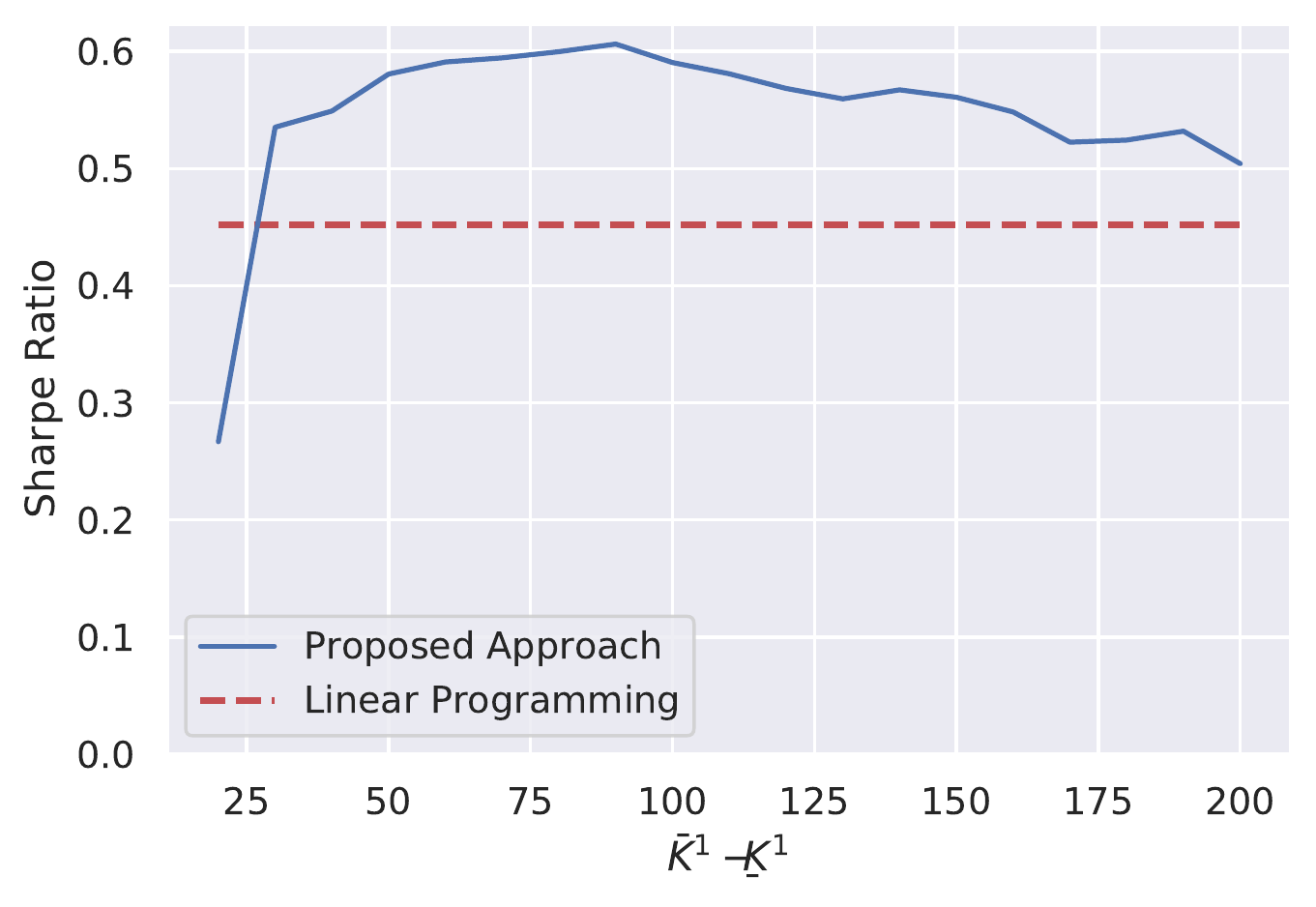}
\caption{We compare the Sharpe ratio of a trading strategy that was trained in line with Algorithm~\ref{algo_training_nn} with the Sharpe ratio of a strategy that was trained according to an LP-approach described in \cite[Section 5.2]{lutkebohmert2021robust}. We depict the Sharpe ratio in dependence of the width of the bounds $\overline{K}^1-\underline{K}^1$. The underlying security is the \emph{EUROSTOXX 50}, while the testing period ranges from September $2013$ until July $2018$.}\label{fig_NN_vs_LP}
\end{center}
\end{figure}

 In this example, the bounds \eqref{eq_s_underline_overline} mentioned in Remark~\ref{rem_choice_Omega} take the values $\underline{S}^1= 72.65$ and $\overline{S}^1=122.86$, respectively. We  set $\varepsilon=d=1$ and hence have  $(\overline{S}^1+\varepsilon)-(\underline{S}^1-\varepsilon)=52.21$. Figure~\ref{fig_NN_vs_LP} supports that the choice of these bounds is reasonable, since increasing the interval width further does not significantly improve the performance of our trading approach.

Note that Figure~\ref{fig_NN_vs_LP} also illustrates that both LP-approach and the neural network approach from Algorithm~\ref{algo_training_nn} perform remarkably well in the considered setting, while the neural network approach even marginally outperforms the LP-approach from \cite{lutkebohmert2021robust} which we here used as benchmark to compare. The drawback of the LP-approach however is that it is no more applicable in higher dimensions. We provide with the examples from Section~\ref{exa_large_number} empirical evidence that our approach (Algorithm~\ref{algo_training_nn}) can still be  applied successfully in higher dimensions in contrast to the LP-approach.

\subsection{Outperformance of the market}\label{exa_bad_market}
We consider $d=2$ underlying securities. These are the American stock market index $S\&P~500$ and the European market index \emph{EUROSTOXX}~$50$, respectively.
To estimate the ambiguity set $\mathcal{P}$ according to the methodology proposed in Section~\ref{sec_ambiguity_set_p} we consider historic daily data of both underlying indices from $1986/12/31$ to $2006/12/01$. This results in $N=5000$ observation dates. We set the number of future trading days to $n=9$ and consider daily trading. We train a trading strategy according to Algorithm~\ref{algo_training_nn} and test it on the period from $2006/12/04$ to $2013/01/25$, in which the considered indices perform remarkably bad\footnote{We refer to \cite{fischer2018deep,ghosh2021forecasting,krauss2017statistical} for detailed discussions of different trading periods over the last $30$ years.}, compare Figure~\ref{fig_test_period} which depicts both the evolution in the train period and the test period.  In practice, to account for recent changes in the underlying time series, it might be important to train the model with streamed data, that is, to fine-tune the model with the most recent incoming data, which is known as online-learning in the literature, see also \cite{de2018advances}. We adjust our framework to online-learning by performing another $5$ iterations of fine-tuning with the augmented data during the test period every time the trading window is shifted forward. Given that there are $150$ non-overlapping trading windows in the test period, the total number of iterations in the setting will be $100+5\times150=850$, compared to merely $100$ iterations without incorporating online-learning. Due to the increased computational costs, we only demonstrate online-learning in this  $S\&P~500$ and \emph{EUROSTOXX}~$50$ scenario, see also Table~\ref{tbl_profit_example_1}.
\begin{figure}[h!]
\begin{center}
\includegraphics[width=10cm]{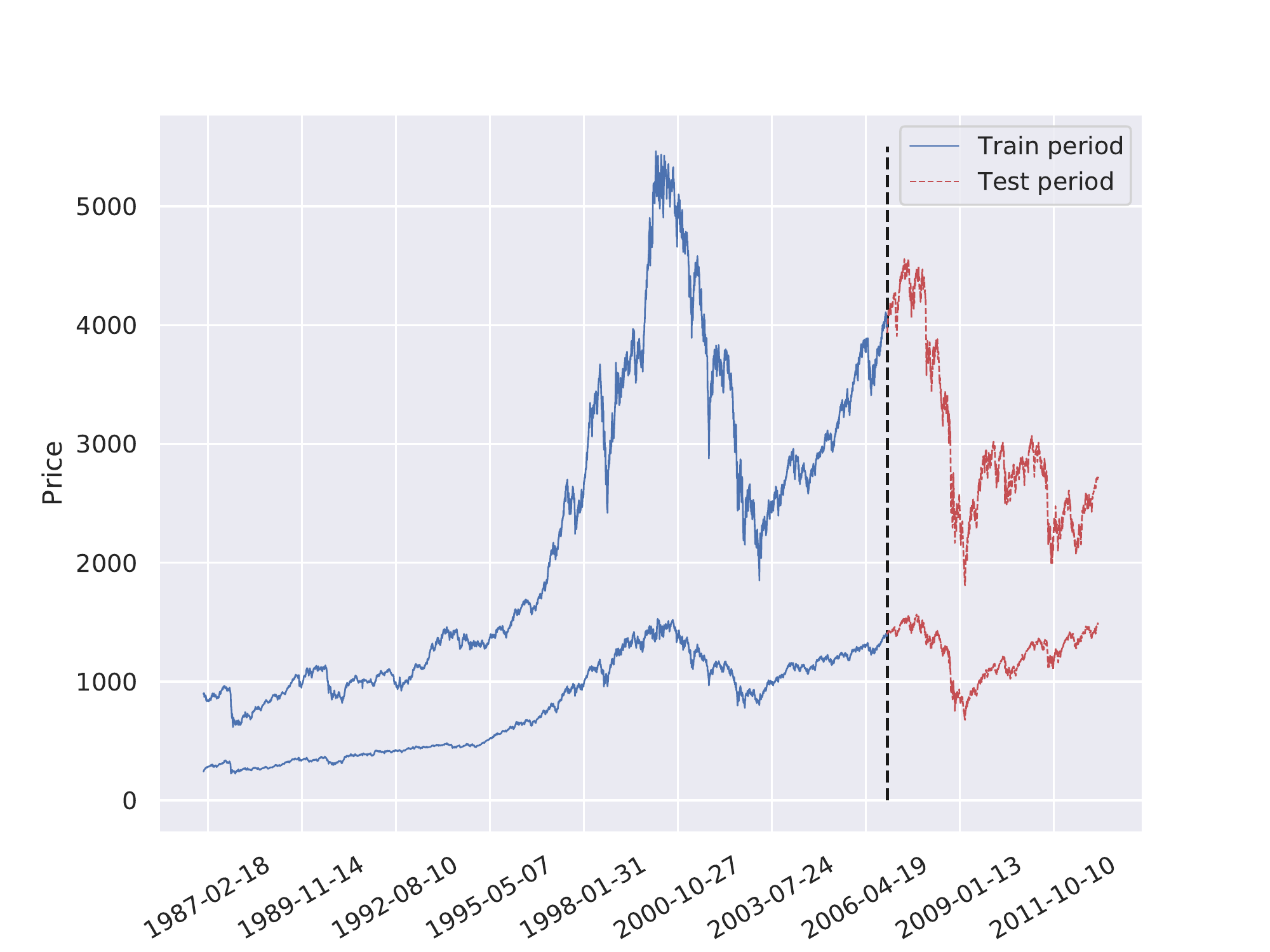}
\caption{The figure shows the evolution of the $S\&P~500$ (bottom) and the \emph{EUROSTOXX}~$50$ (top) in the training period from  $1986/12/31$ to $2006/11/30$ (blue) and in the testing period from $2006/11/31$ to $2013/01/25$ (red).}\label{fig_test_period}
\end{center}
\end{figure}

The results in Table~\ref{tbl_profit_example_1} showcase that our trained strategy clearly outperforms the market and therefore provide empirical evidence that the robust strategies determined by our approach are able to cut losses even in extremely difficult market scenarios. { Moreover, in Figure~\ref{equity_curve_stoxx} we depict the \emph{equity curve} of the trained strategies, i.e., the cumulated net profit of the trained strategy during the test period. As the profit depends on the random initialization we decided to depict the equity curve on a percentile basis after having evaluated $50$ experiments.}

\begin{table}[h!]
	\small
	\begin{center}
		\begin{tabular}{l>{\bfseries}c>{}c>{\bfseries}c>{}c>{\bfseries}c>{}c>{\bfseries}c} \toprule
			
			{ Transaction Costs}& 0 & $0$ \text{B}\&\text{H} & Prop.& Prop. $\text{B}\&\text{H}$ & P.S. & P.S. $\text{B}\&\text{H}$  &  {P.S. O.L.}\\
			
			\midrule
			Overall Profit& 8666.13 & -9932.61 &7974.68 & -11203.71&9293.87 & -9993.21 & 10875.71\\ 
			Average Profit\tablefootnote{Here, and in all other tables \emph{Average Profit} refers to overall profit divided by the number of periods with $n=9$ trading days on the testing data.}   &57.77 & -66.22 &53.16& -74.69& 61.96 & -66.62& 72.5\\
			\% of Profitable Trades &55.07 & 53.33 &53.65 & 52.67& 55.88 & 53.33& 57.21\\ 
			Max. Profit &2810.23 & 3375.2 & 2797.42 & 3367.86& 2956.79 & 3374.79& 2803.53\\ 
			Min. Profit &-845.17 & -5430.9 & -853.69 & -5438.6&-839.02&  -5431.31& -1393.25\\ 
			Sharpe Ratio & 0.9376 & -0.221 &0.864 &-0.249 & 0.9948 & -0.222& 0.9643\\ 
			Sortino Ratio& 2.361 & -0.285 &2.1681 &-0.321 & 2.5577 & -0.286& 1.8231\\  \bottomrule
		\end{tabular}
		\caption{The Table describes the success of the trained neural network statistical arbitrage strategies (bold) in the setting of Section~\ref{exa_bad_market}, where we consider trading in the  $S\&P~500$ and the \emph{EUROSTOXX} in a test period containing an overall declining market scenario which is depicted in Figure~\ref{fig_test_period}. The results are shown in dependence of { zero transaction costs ($0$), nonzero transaction costs with proportional~(Prop.), per share transaction costs~(P.S.), and online learning with per share transaction costs (P.S. O.L.), respectively}. Further, we compare the results with a buy-and-hold\tablefootnote{Note that we consider here the profits of buy-and-hold strategies which are held for $n=9$ days. Setting up and closing out the position both incurs transaction costs which are not negligible due to the relatively short length of the trading period. } strategy~(B$\&$H) with $\Delta_i^k = 10$ for all $i=0,\dots,9$ and all $k=1,2$. { We  also report the 10-day average profits of the one-time-buy-and-hold strategy of the above three scenarios respectively, i.e., without opening and closing the positions during each trading window\tablefootnote{Note that in the case without transaction costs and in the case with per share transaction costs, the one-time-buy-and-hold strategy leads to a worse outcome than the buy-and hold position which is closed every 9 days and than reopened the next day. This at first glance unintuitive result is the consequence of considering a \emph{bearish} testing period where closing out a position and avoiding exposure to a decreasing stock over night can be favourable even though transaction costs apply. }. (0 B\&H: $-74.07$, Prop. B\&H: $-74.13$, P.S. B\&H: $-74.07$)}.}
		\label{tbl_profit_example_1}
	\end{center}
\end{table}

\begin{figure}[h!]
	\begin{center}
		\includegraphics[width=10cm]{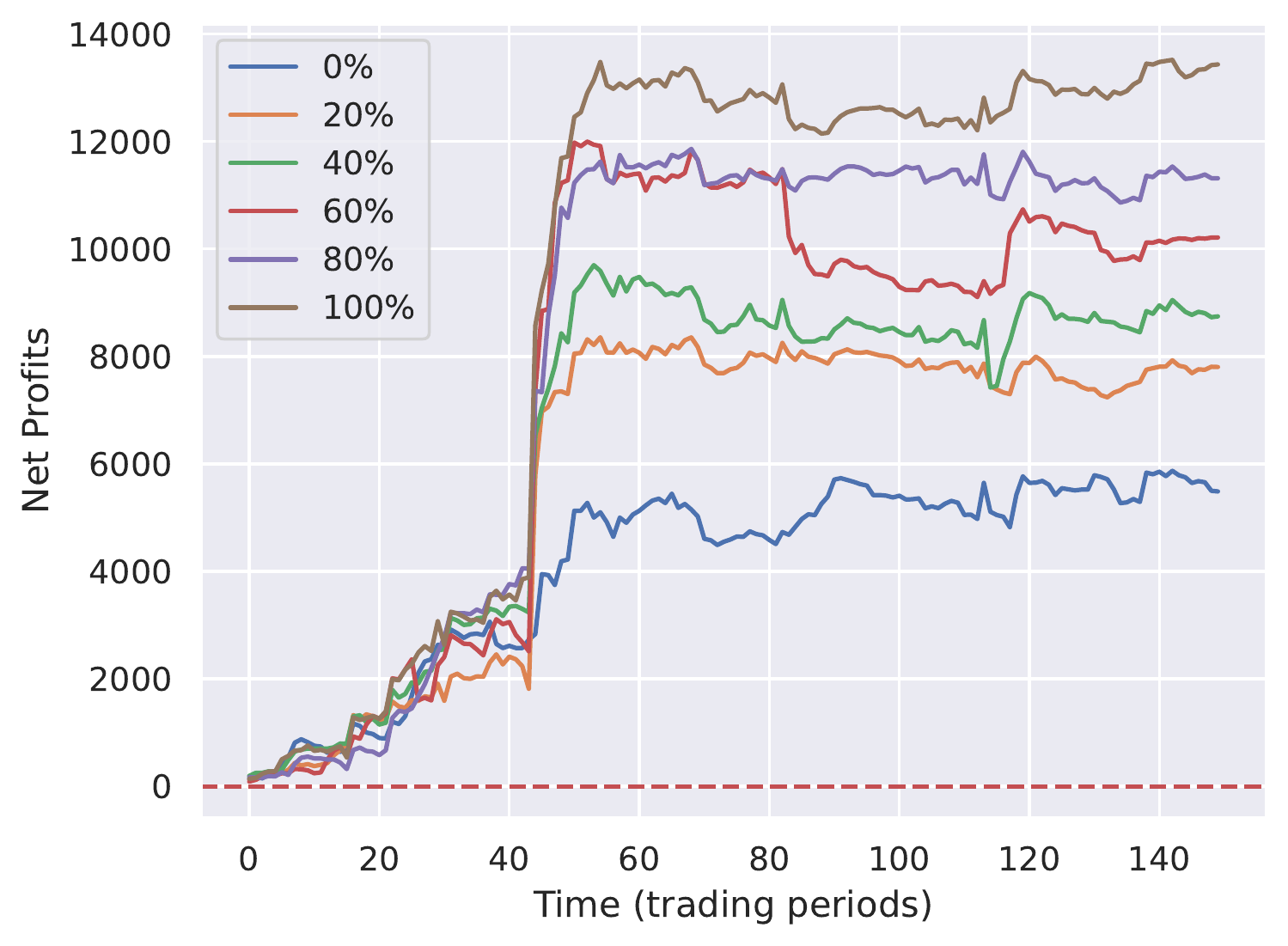}
		\caption{The figure shows the \emph{equity curves} of the trained strategies with per share transaction costs during the testing period depicted in Figure \ref{fig_test_period}. The y-axis represents the cumulated net profits as the trading window propagates. Since a total of 50 independent experiments are conducted, we plot the equity curves by quantiles, where the rank is based on the net profit of the experiment at the end of the testing period. }\label{equity_curve_stoxx}
	\end{center}
\end{figure}

\begin{figure}[!htb]
	\begin{minipage}{0.32\textwidth}
	\includegraphics[width=\linewidth]{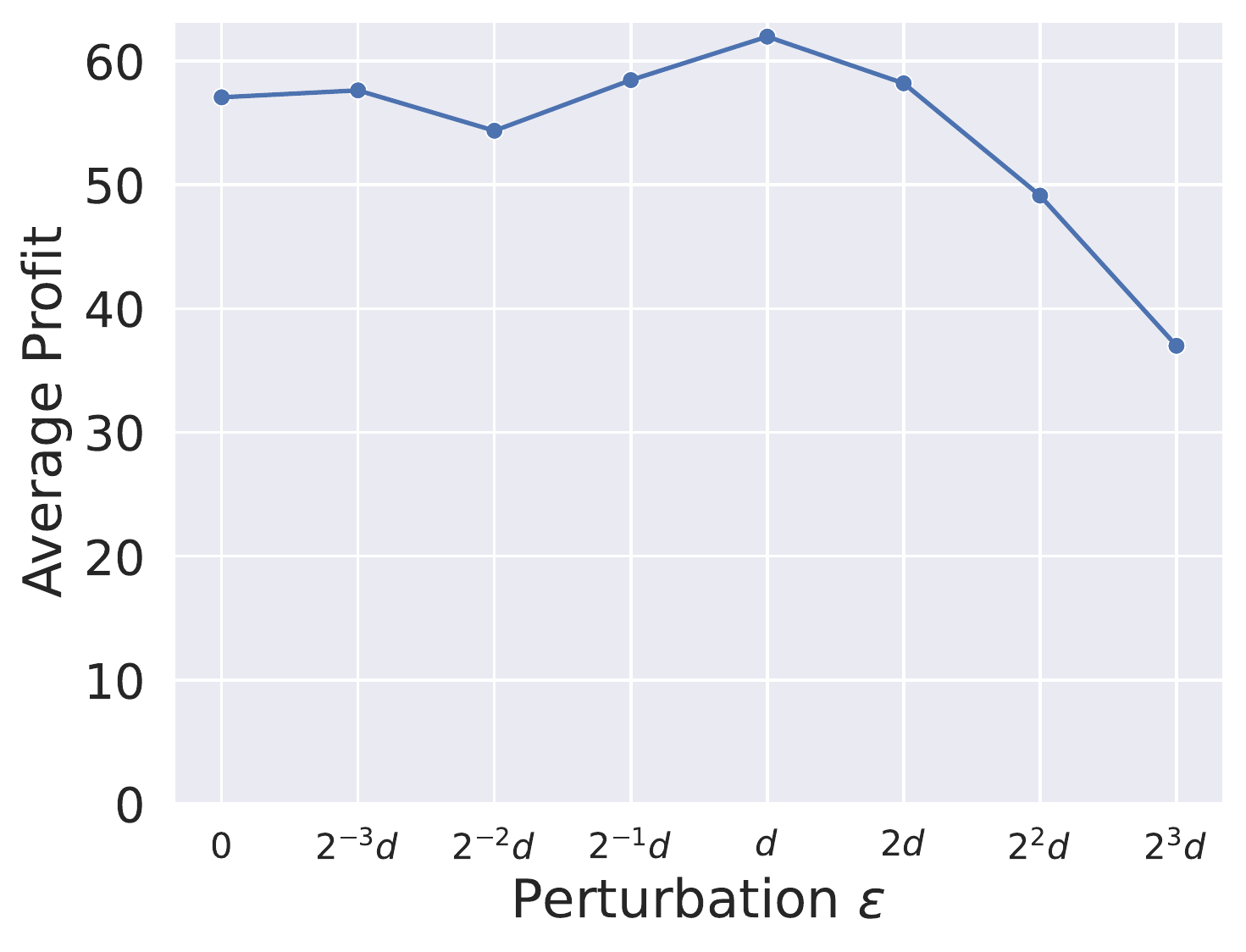}
	\end{minipage}\hfill
	\begin{minipage}{0.32\textwidth}
	\includegraphics[width=\linewidth]{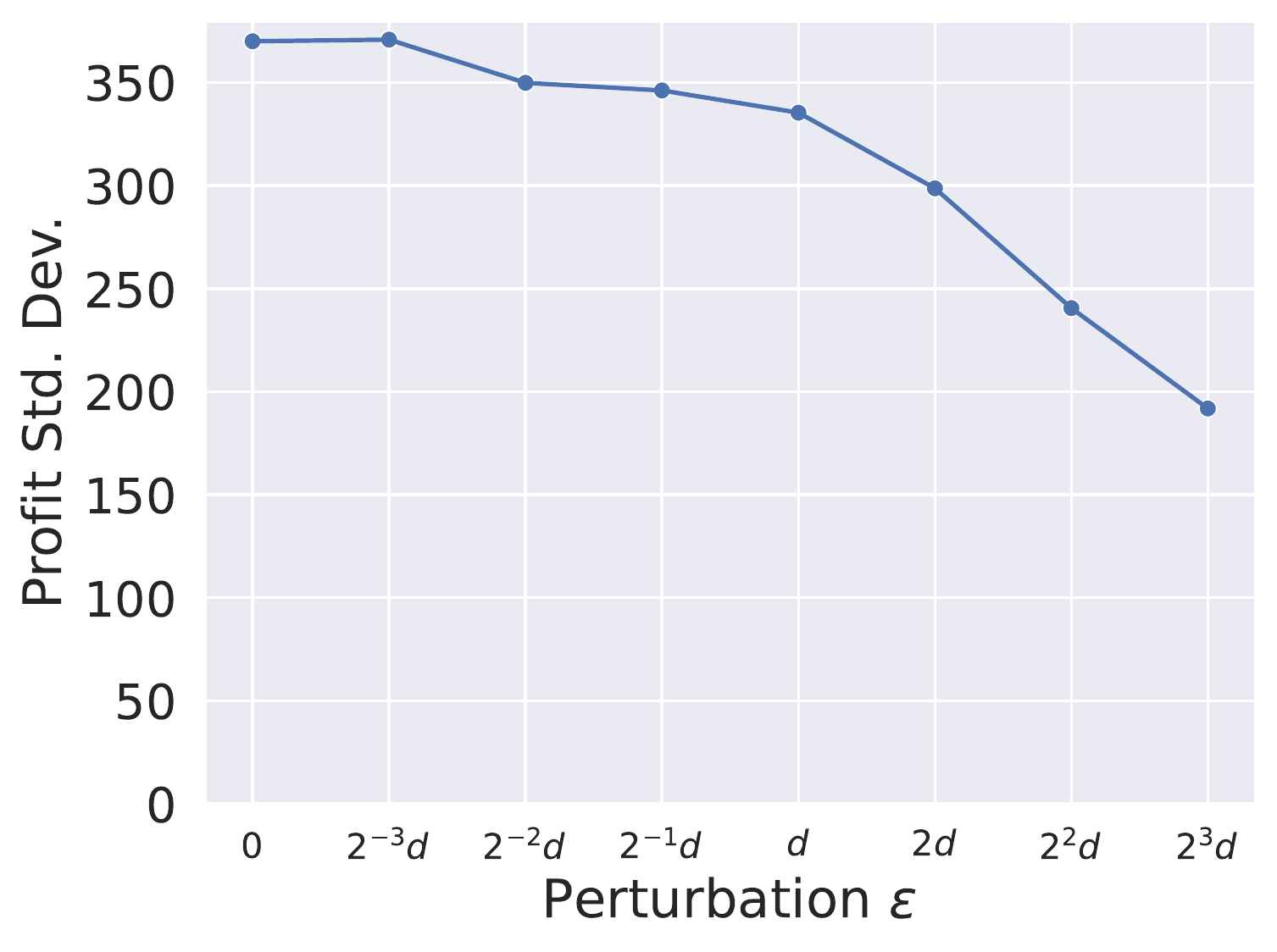}
	\end{minipage}\hfill
	\begin{minipage}{0.32\textwidth}%
	\includegraphics[width=\linewidth]{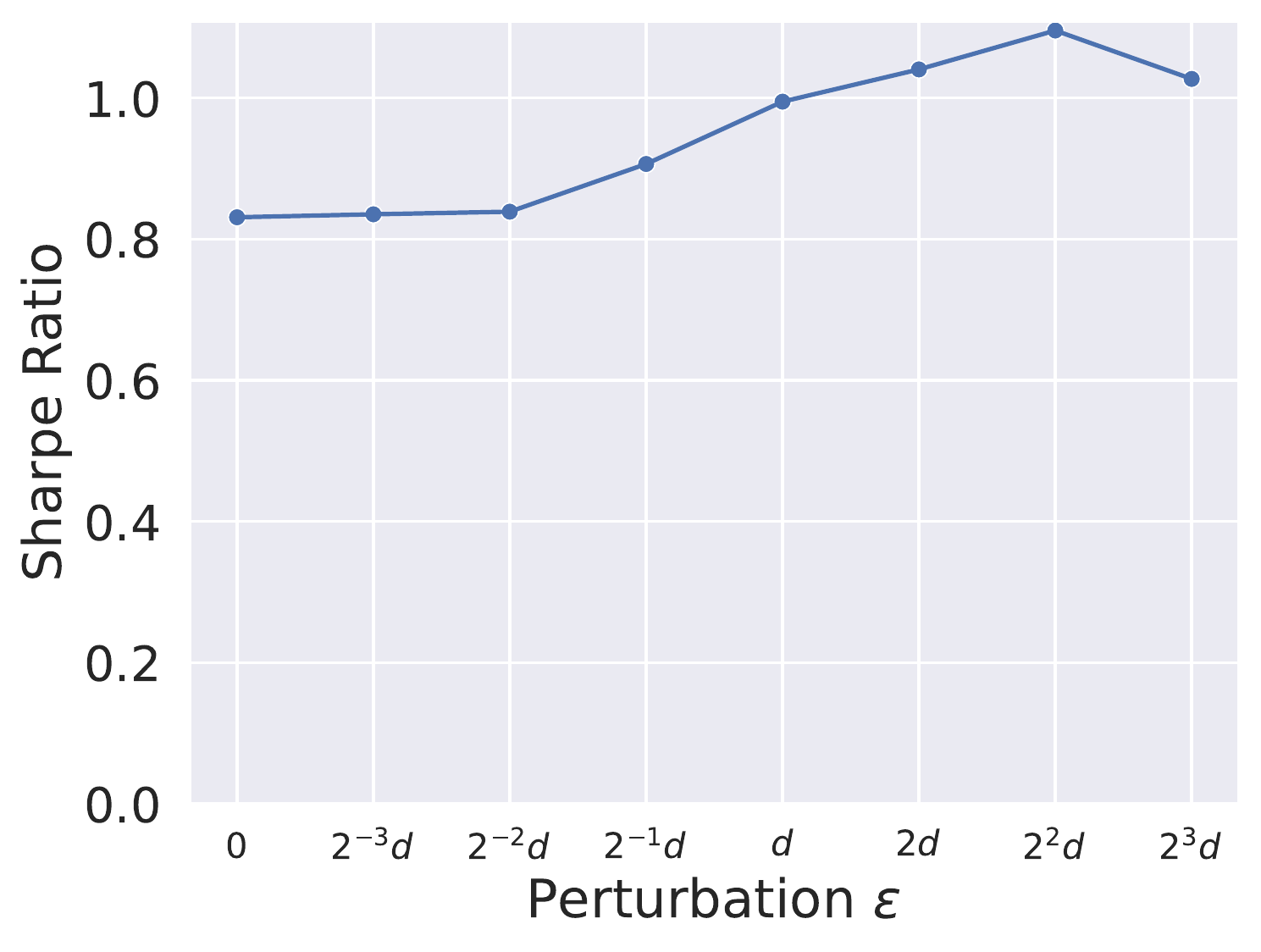}
	\end{minipage}
	\caption{These figures show the $\varepsilon$-sensitivity analysis in the scenario of trading the stocks $S\&P~500$ and \emph{EUROSTOXX} where per share transaction costs are considered. The results provide evidence that introducing a non-zero perturbation $\varepsilon$ is indeed beneficial for a better profit and sharpe ratio. Moreover, it can be seen that a larger $\varepsilon$, i.e. more uncertainty, makes the model more risk-averse, and therefore leads to a smaller standard deviation of the profits.}\label{2_assets_sensitivity_analysis}
\end{figure}

\subsection{Trading in a large number of stocks}\label{exa_large_number}
In this example we consider $d=10,20,30,40$, and $50$ securities, respectively. { The securities correspond to constituents of the  S\&{P}~500 index which were selected while respecting diversity across different industry sectors according to the \emph{Global Industry Classification Standard (GICS)}}\footnote{{The considered companies for all dimensions  can be found in the provided github-repository under} \href{https://github.com/YINDAIYING/Deep-Robust-Statistical-Arbitrage}{https://github.com/YINDAIYING/Deep-Robust-Statistical-Arbitrage}. In the case $d=10$ we consider the companies with ticker symbols: OKE, PG, RCL, SBUX, UNM, USB, VMC, WELL, WMB, XOM. }. { The diversity-based selection criterion was chosen to avoid the construction of a biased and unbalanced stock universe due to the predominance of companies from \emph{Information Technology} - sector in the US market.}
The training period ranges from $2000/01/03$ to $2015/11/25$ and the testing period from $2015/11/26$ to $2020/09/02$, compare Figure~\ref{fig_10_asset_test_period}, where we illustrate $d= 10$ securities. In Table~\ref{tbl_profit_example_3}, Table~\ref{tbl_profit_example_4}, Table~\ref{tbl_profit_example_5}, Table~\ref{tbl_profit_example_6}, and Table~\ref{tbl_profit_example_7}, we depict the results of the trained strategy evaluated on the testing period.
 { Moreover, in Figure~\ref{equity_curve_40_assets} we depict the \emph{equity curve} of the trained strategy for $d=40$ assets.}

\begin{figure}[h!]
\begin{center}
\includegraphics[width=10cm]{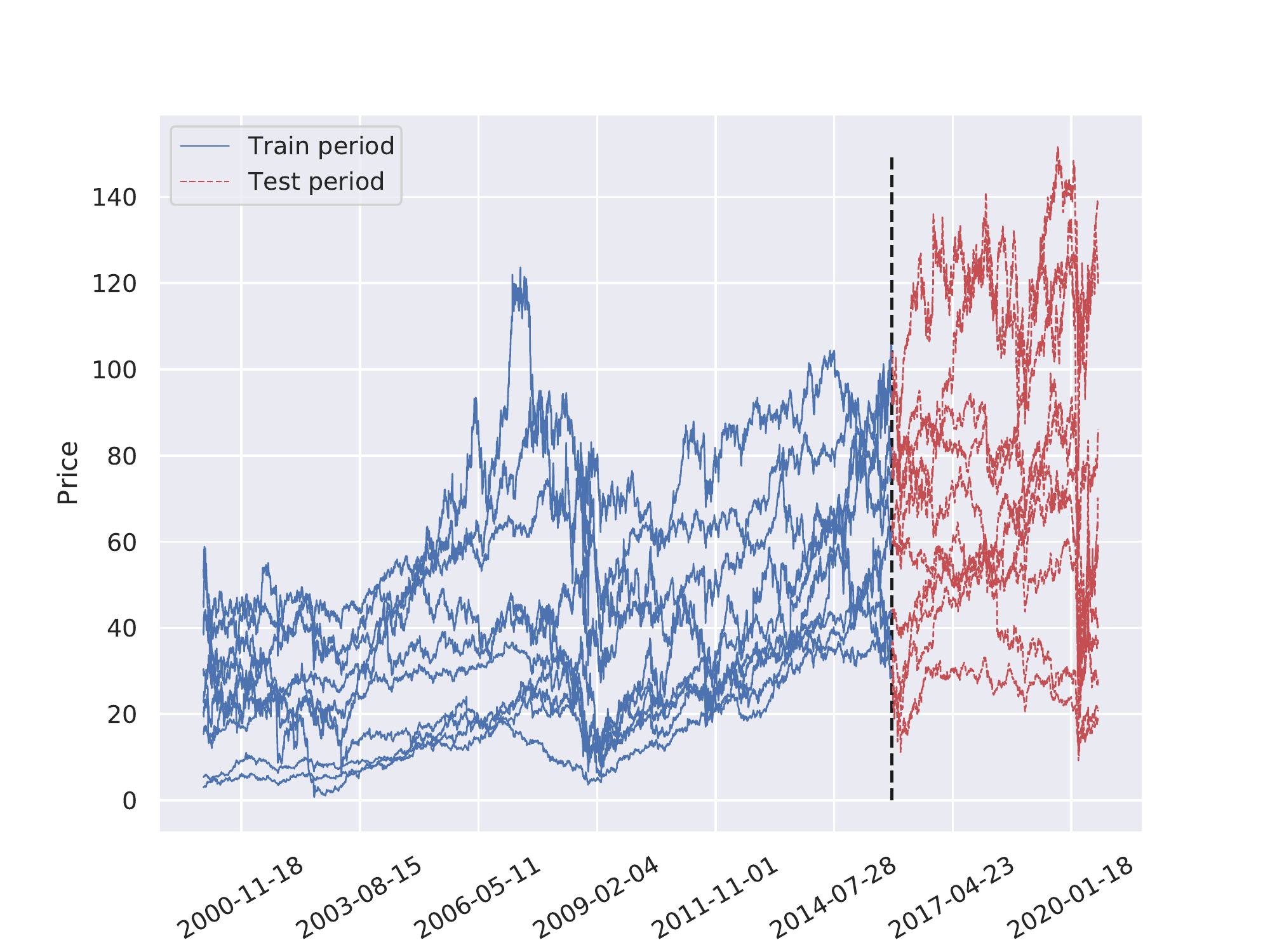}
\caption{The Figure shows the evolution of 10 securities in the training period from $2000/01/03$ to $2015/11/25$ (blue) and the testing period from $2015/11/26$ to $2020/09/02$ (red).}\label{fig_10_asset_test_period}
\end{center}
\end{figure}

\begin{table}[h!]
	\small
	\begin{center}
		\begin{tabular}{l>{\bfseries}c>{}c>{\bfseries}c>{}c>{\bfseries}c>{}c} \toprule
			
			{ Transaction Costs}& 0 & $0$ \text{B}\&\text{H} & Prop.& Prop. $\text{B}\&\text{H}$ & P.S. & P.S. $\text{B}\&\text{H}$ \\
			
			\midrule
			
			Overall Profit& 113.34 & -283.70& 66.38 & -449.06& 88.05 &  -523.70\\ 
			Average Profit  &0.94 & -2.36& 0.55 &-3.74 & 0.73 & -4.36\\
			\% of Profitable Trades &56.65&60.0 &55.15 & 60.0& 56.07 & 60.0\\ 
			Max. Profit  &221.97& 1006.58 &220.56 &1005.38 & 237.01 & 1004.58\\ 
			Min. Profit &-287.89 & -2027.22 &-288.94 & -2028.37& -316.92 & -2029.22\\
			Sharpe Ratio & 0.0979 & -0.043&0.0346 &-0.068 & 0.054 & -0.079\\ 
			Sortino Ratio& 0.1333 & -0.037&0.0646 & -0.058& 0.080 & -0.068\\  \bottomrule
		\end{tabular}
		\caption{We consider $d=10$ securities. The  table shows the outcome of trading strategies that were trained according to Algorithm~\ref{algo_training_nn} (bold) in the setting of Section~\ref{exa_large_number} in dependence of { zero transaction costs ($0$), nonzero transaction costs with proportional~(Prop.), per share transaction costs~(P.S.), and online learning with per share transaction costs (P.S. O.L.), respectively}, and in comparison with a buy-and-hold strategy (B$\&$H) with $\Delta_i^k=10$ for all $i=0,\dots,9$ and all $k=1,\dots,10$. { We  also report the 10-day average profits of the one-time-buy-and-hold strategy of the above three scenarios respectively, i.e., without opening and closing the positions during each trading window. (0 B\&H: -0.77, Prop. B\&H: -0.78, P.S. B\&H: -0.79)}}
		\label{tbl_profit_example_3}
	\end{center}
\end{table}

\begin{table}[h!]
	\small
	\begin{center}
		\begin{tabular}{l>{\bfseries}c>{}c>{\bfseries}c>{}c>{\bfseries}c>{}c} \toprule
			
			{ Transaction Costs}& 0 & $0$ \text{B}\&\text{H} & Prop.& Prop. $\text{B}\&\text{H}$ & P.S. & P.S. $\text{B}\&\text{H}$ \\
			
			\midrule
			
			Overall Profit& 372.38 & -386.45&165.60 & -815.91& 376.33 & -866.45\\ 
			Average Profit  &3.1 & -3.22&1.9 &-6.8 & 3.14 & -7.22\\
			\% of Profitable Trades &62.73 & 61.67&61.28 &61.67 & 64.28 & 61.67\\ 
			Max. Profit  &592.06 & 2313.86&588.82 &2310.66 & 657.57 & 2309.86\\ 
			Min. Profit &-1123.36 & -3968.04& -1126.8 & -3971.07& -1273.33 & -3972.04\\
			Sharpe Ratio & 0.0862 & -0.026&0.0311 & -0.054& 0.1077 & -0.058\\ 
			Sortino Ratio& 0.0821 & -0.024& 0.0296 & -0.051& 0.0858 & -0.054\\  \bottomrule
		\end{tabular}
		\caption{We consider $d=20$ securities. The  Table provides the outcome of trading strategies that were trained according to Algorithm~\ref{algo_training_nn} (bold) in the setting of Section~\ref{exa_large_number} in dependence of { zero transaction costs ($0$), nonzero transaction costs with proportional~(Prop.), per share transaction costs~(P.S.), and online learning with per share transaction costs (P.S. O.L.), respectively}. Moreover, the results are compared with a buy-and-hold strategy (B$\&$H) with $\Delta_i^k=10$ for all $i=0,\dots,9$ and all $k=1,\dots,20$. { We  also report the 10-day average profits of the one-time-buy-and-hold strategy of the above three scenarios respectively, i.e., without opening and closing the positions during each trading window. (0 B\&H: 5.41, Prop. B\&H: 5.39, P.S. B\&H: 5.38)}}
		\label{tbl_profit_example_4}
	\end{center}
\end{table}

\begin{table}[h!]
	\small
	\begin{center}
		\begin{tabular}{l>{\bfseries}c>{}c>{\bfseries}c>{}c>{\bfseries}c>{}c} \toprule
			
			{ Transaction Costs} & 0 & $0$ \text{B}\&\text{H} & Prop.& Prop. $\text{B}\&\text{H}$ & P.S. & P.S. $\text{B}\&\text{H}$ \\
			
			\midrule
			
			Overall Profit& 870.82 &-329.33 &603.89 & -902.33&809.36 & -1049.33\\ 
			Average Profit  &7.26 &-2.74 &5.03 &-7.52 &6.74 &  -8.74\\
			\% of Profitable Trades &61.85 &62.5 &59.93 & 61.67& 63.97 & 60.83\\ 
			Max. Profit  &942.39 & 3028.37&927.85 &3024.03&1135.97 & 3022.37\\ 
			Min. Profit &-1709.51 & -5077.4&-1712.0 & -5081.55& -2287.09 & -5083.4\\
			Sharpe Ratio & 0.1279 &-0.017 &0.0507 & -0.047&0.1276 & -0.054\\ 
			Sortino Ratio& 0.1253 &-0.016 &0.0386 & -0.044&0.1016 & -0.051\\  \bottomrule
		\end{tabular}
		\caption{We consider $d=30$ securities. The  Table provides the outcome of trading strategies that were trained according to Algorithm~\ref{algo_training_nn} (bold) in the setting of Section~\ref{exa_large_number} in dependence of { zero transaction costs ($0$), nonzero transaction costs with proportional~(Prop.), per share transaction costs~(P.S.), and online learning with per share transaction costs (P.S. O.L.), respectively}, and in comparison with a buy-and-hold strategy (B$\&$H) with $\Delta_i^k=10$ for all $i=0,\dots,9$ and all $k=1,\dots,30$. { We  also report the 10-day average profits of the one-time-buy-and-hold strategy of the above three scenarios respectively, i.e., without opening and closing the positions during each trading window. (0 B\&H: 11.82, Prop. B\&H: 11.78, P.S. B\&H: 11.77)}}
		\label{tbl_profit_example_5}
	\end{center}
\end{table}

\begin{table}[h!]
	\small
	\begin{center}
		\begin{tabular}{l>{\bfseries}c>{}c>{\bfseries}c>{}c>{\bfseries}c>{}c} \toprule
			
			{ Transaction Costs} & 0 & $0$ \text{B}\&\text{H} & Prop.& Prop. $\text{B}\&\text{H}$ & P.S. & P.S. $\text{B}\&\text{H}$ \\
			
			\midrule
			
			Overall Profit& 1430.99 &-571.02 &953.07& -1294.60&984.11& -1531.02\\ 
			Average Profit  &11.92 &-4.76 &7.94 &-10.79 &8.2 & -12.76\\
			\% of Profitable Trades &65.72 &62.5 &64.25 & 61.67& 65.07 & 61.67\\ 
			Max. Profit  &1449.09 & 3631.47&1446.8 &3626.01&1502.57 & 3623.47\\ 
			Min. Profit &-3022.75 & -6111.33&-3055.53 & -6116.53& -3322.79 & -6119.33\\
			Sharpe Ratio &0.1681 &-0.024 &0.1096 & -0.055&0.1064 & -0.065\\ 
			Sortino Ratio& 0.1329
			 &-0.022 &0.0879& -0.051& 0.0825 & -0.06\\  \bottomrule
		\end{tabular}
		\caption{We consider $d=40$ securities. The  Table provides the outcome of trading strategies that were trained according to Algorithm~\ref{algo_training_nn} (bold) in the setting of Section~\ref{exa_large_number} in dependence of { zero transaction costs ($0$), nonzero transaction costs with proportional~(Prop.), per share transaction costs~(P.S.), and online learning with per share transaction costs (P.S. O.L.), respectively}, and in comparison with a buy-and-hold strategy (B$\&$H) with $\Delta_i^k=10$ for all $i=0,\dots,9$ and all $k=1,\dots,40$. { We  also report the 10-day average profits of the one-time-buy-and-hold strategy of the above three scenarios respectively, i.e., without opening and closing the positions during each trading window. (0 B\&H: 11.70, Prop. B\&H: 11.66, P.S. B\&H: 11.63)}}
		\label{tbl_profit_example_6}
	\end{center}
\end{table}

\begin{figure}[h!]
	\begin{center}
		\includegraphics[width=10cm]{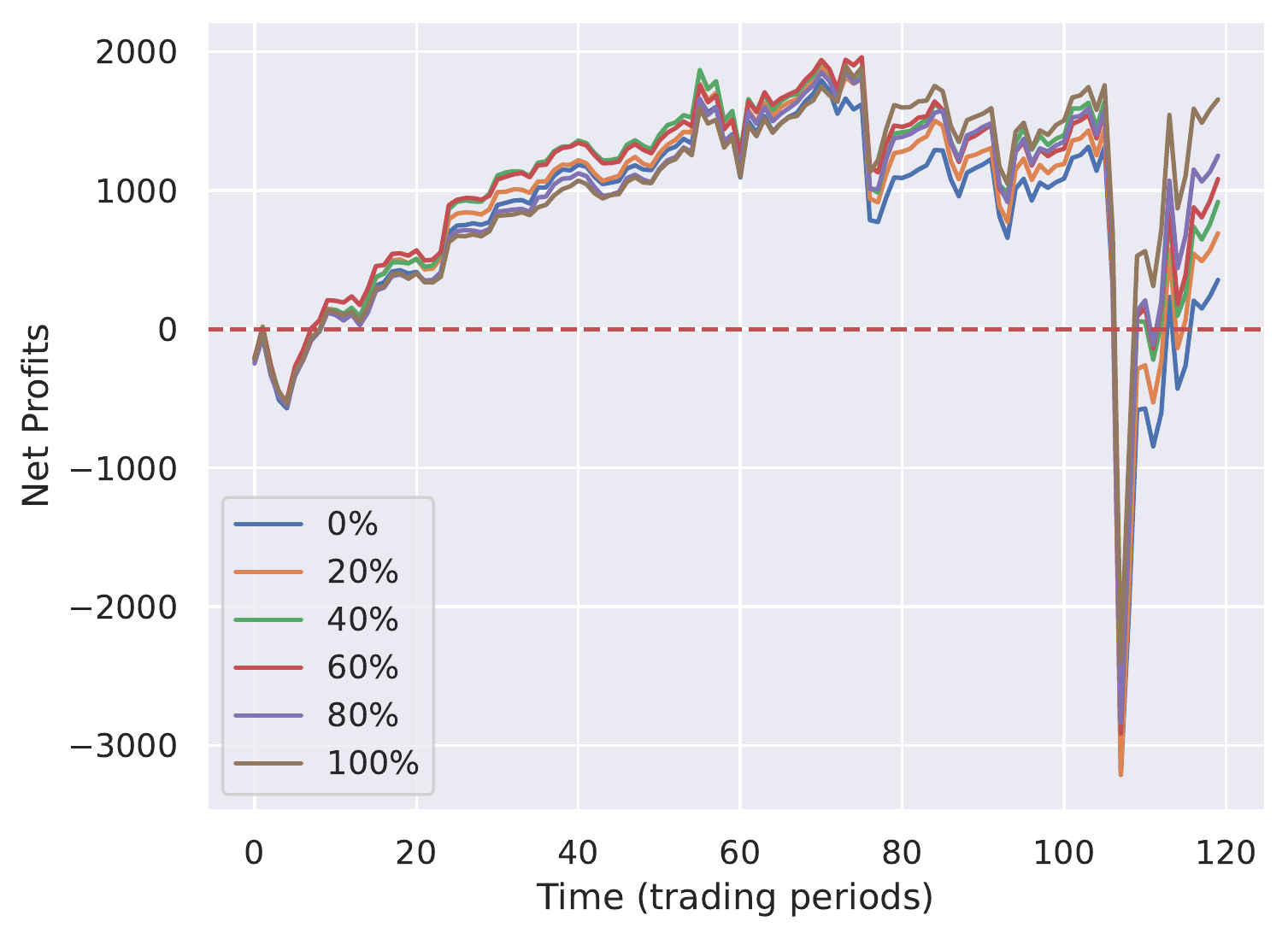}
		\caption{The figure shows the \emph{equity curves} of the trained strategies with per share transaction costs during the testing period of the $40$-assets-example (Table~\ref{tbl_profit_example_6}). The y-axis represents the cumulated net profits as the trading window propagates. Since a total of 50 independent experiments are conducted, we plot the equity curves by quantiles, where the rank is based on the net profit of the experiment at the end of the testing period.}\label{equity_curve_40_assets}
	\end{center}
\end{figure}

\begin{table}[h!]
	\small
	\begin{center}
		\begin{tabular}{l>{\bfseries}c>{}c>{\bfseries}c>{}c>{\bfseries}c>{}c} \toprule
			
			{ Transaction Costs}& 0 & $0$ \text{B}\&\text{H} & Prop.& Prop. $\text{B}\&\text{H}$ & P.S. & P.S. $\text{B}\&\text{H}$ \\
			
			\midrule 
			
			Overall Profit&816.32 &-687.36 &95.11 & -1546.64&60.72 & -1887.36\\ 
			Average Profit  &6.8 &-5.73 &0.79 &-12.89 &0.51 & -15.73\\
			\% of Profitable Trades &63.75 &62.5 &61.85 & 62.5& 63.15 & 62.5\\ 
			Max. Profit  &1732.6& 3970.81&1713.37 &3964.21&1849.47 & 3960.81\\ 
			Min. Profit &-3983.44 & -7107.61&-4011.3 & -7113.79& -4633.36 & -7117.61\\
			Sharpe Ratio & 0.0656 &-0.025 &-0.0045 & -0.057&0.0052 & -0.069\\ 
			Sortino Ratio& 0.0533 &-0.023 &-0.0027 & -0.052&0.0042 & -0.063\\  \bottomrule
		\end{tabular}
		\caption{We consider $d=50$ securities. The  Table provides the outcome of trading strategies that were trained according to Algorithm~\ref{algo_training_nn} (bold) in the setting of Section~\ref{exa_large_number} in dependence of { zero transaction costs ($0$), nonzero transaction costs with proportional~(Prop.), per share transaction costs~(P.S.), and online learning with per share transaction costs (P.S. O.L.), respectively}, and in comparison with a buy-and-hold strategy (B$\&$H) with $\Delta_i^k=10$ for all $i=0,\dots,9$ and all $k=1,\dots,50$. { We  also report the 10-day average profits of the one-time-buy-and-hold strategy of the above three scenarios respectively, i.e., without opening and closing the positions during each trading window. (0 B\&H: 18.71, Prop. B\&H: 18.65, P.S. B\&H: 18.62)}}
		\label{tbl_profit_example_7}
	\end{center}
\end{table}

\begin{figure}[!htb]
	\begin{minipage}{0.32\textwidth}
		\includegraphics[width=\linewidth]{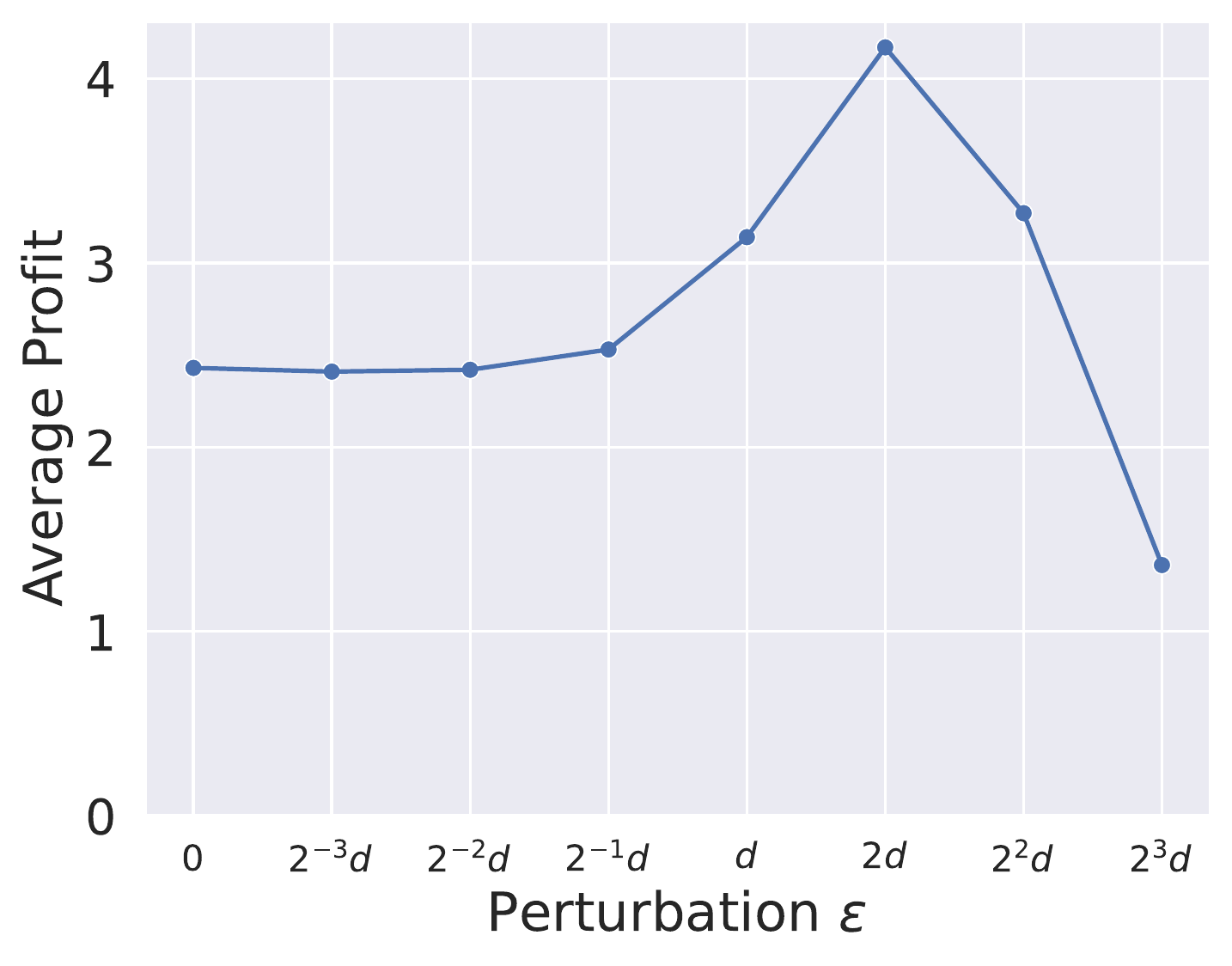}
	\end{minipage}\hfill
	\begin{minipage}{0.32\textwidth}
		\includegraphics[width=\linewidth]{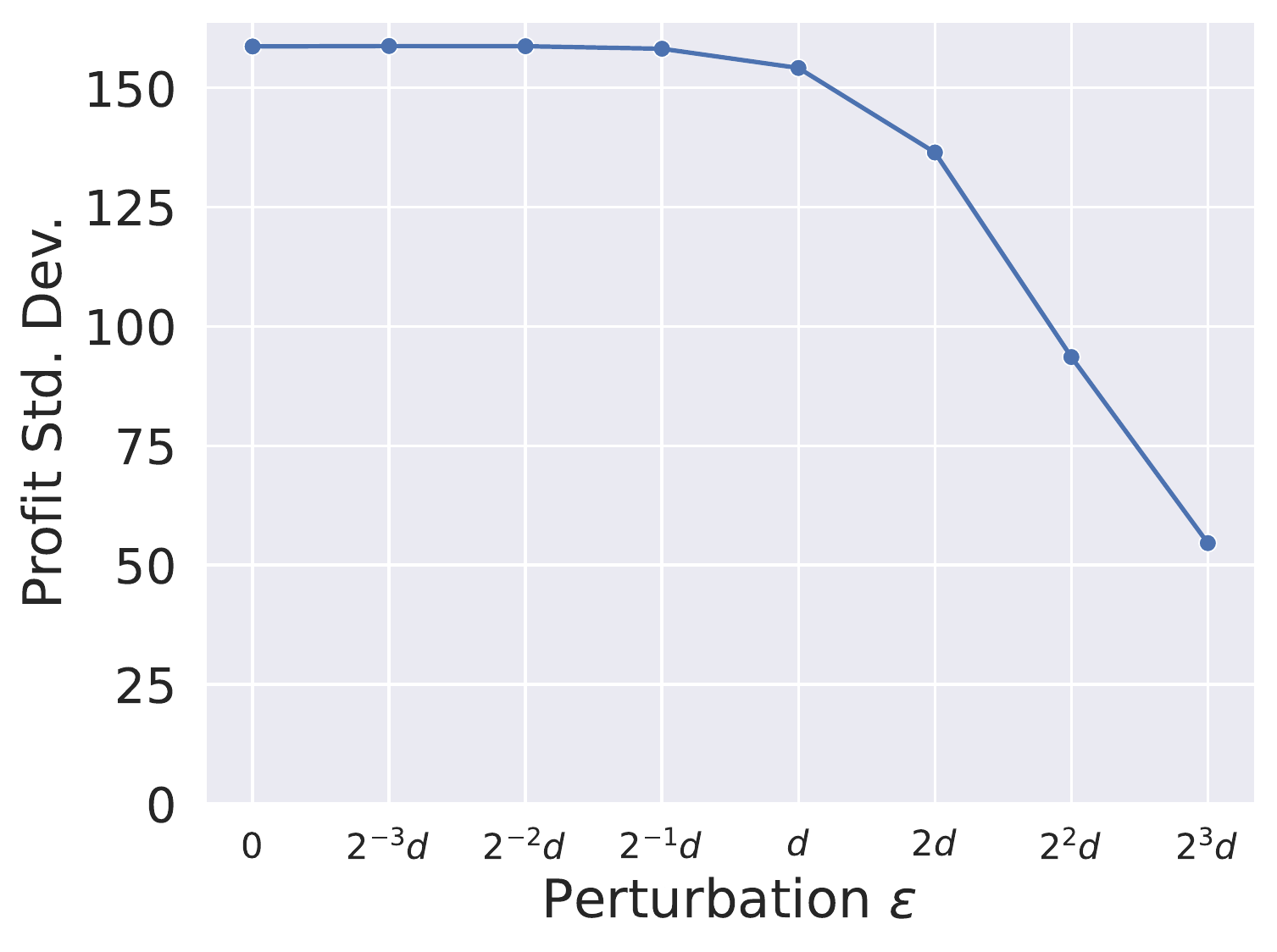}
	\end{minipage}\hfill
	\begin{minipage}{0.32\textwidth}%
		\includegraphics[width=\linewidth]{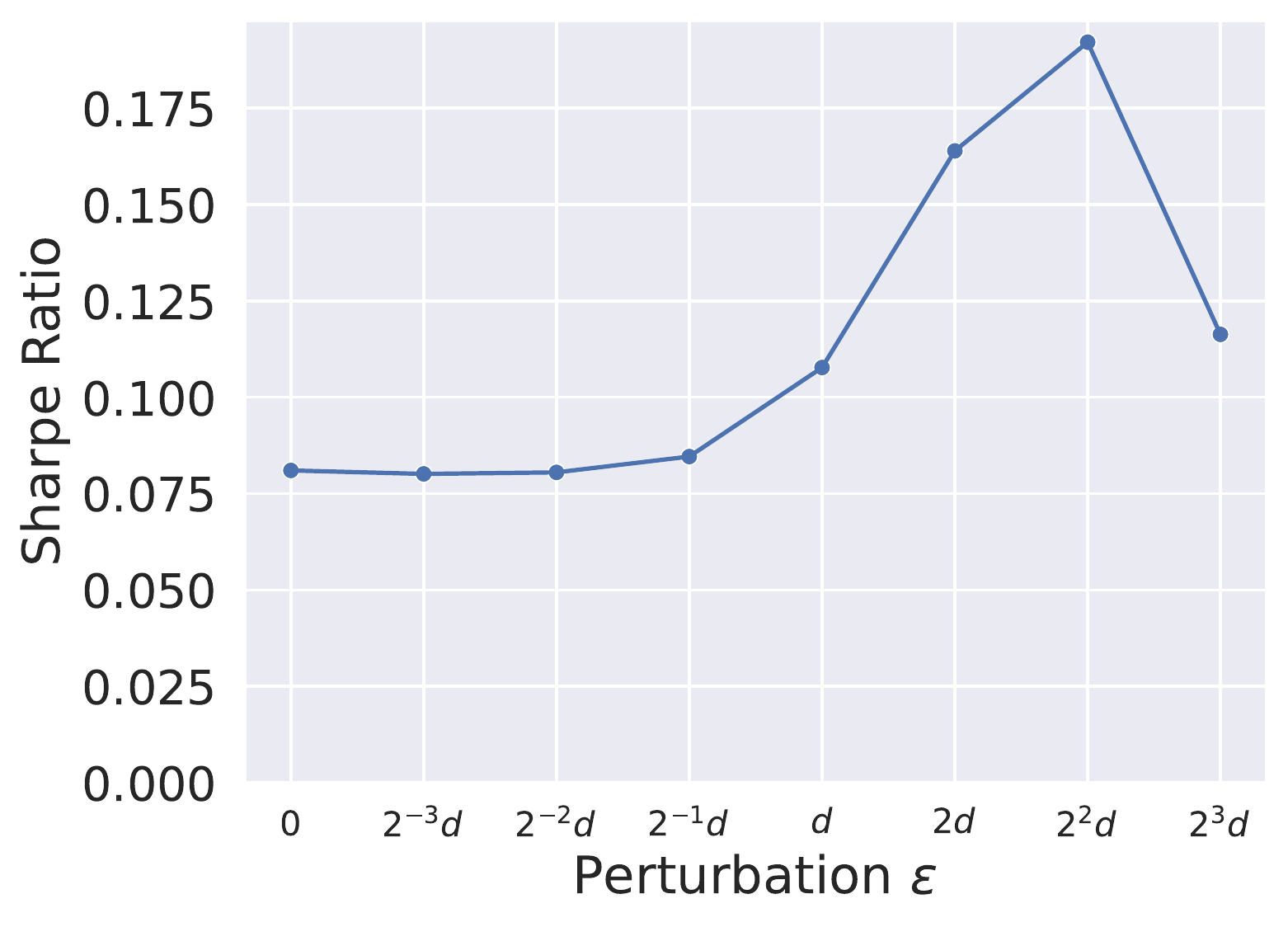}
	\end{minipage}
\begin{minipage}{0.32\textwidth}
	\includegraphics[width=\linewidth]{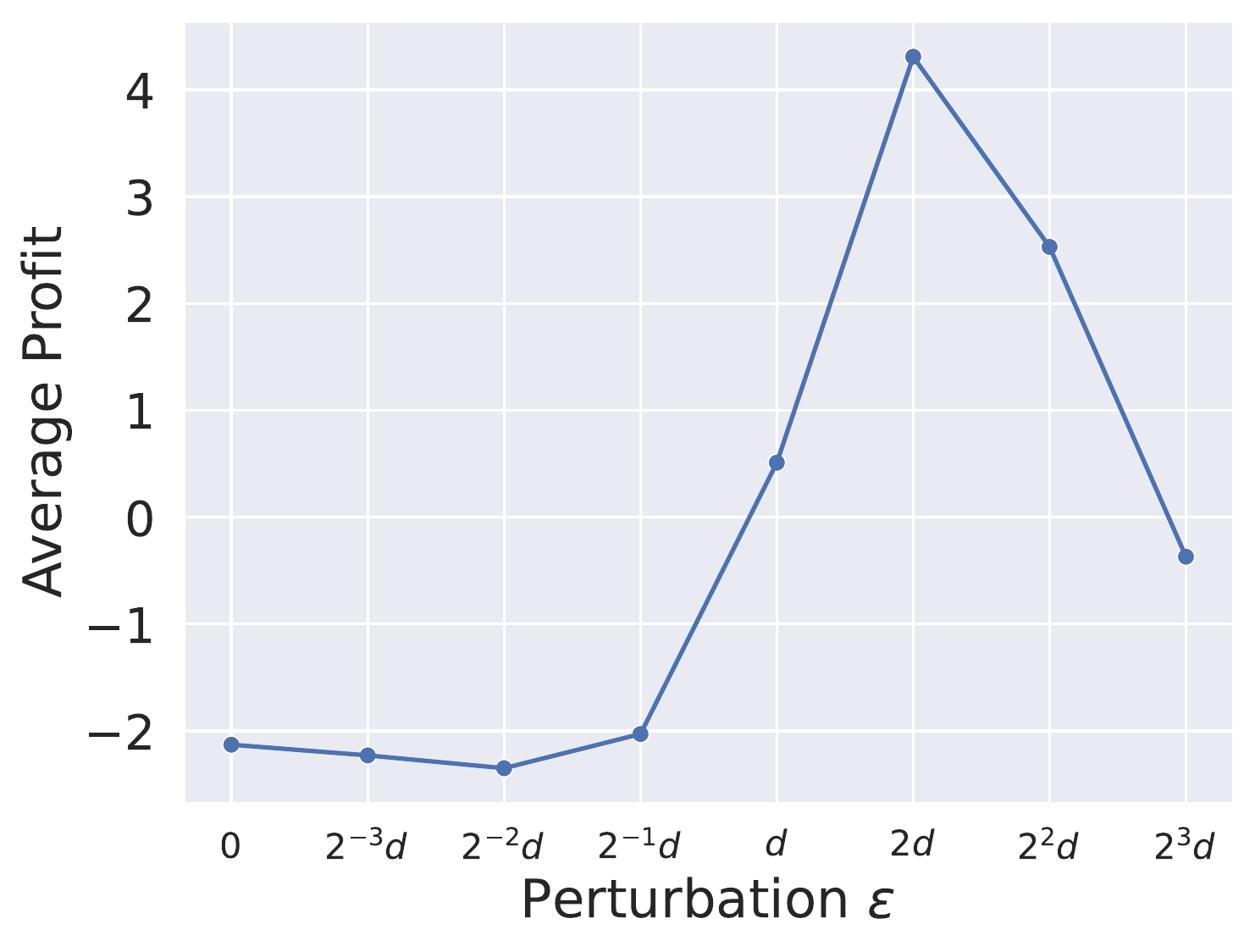}
\end{minipage}\hfill
\begin{minipage}{0.32\textwidth}
	\includegraphics[width=\linewidth]{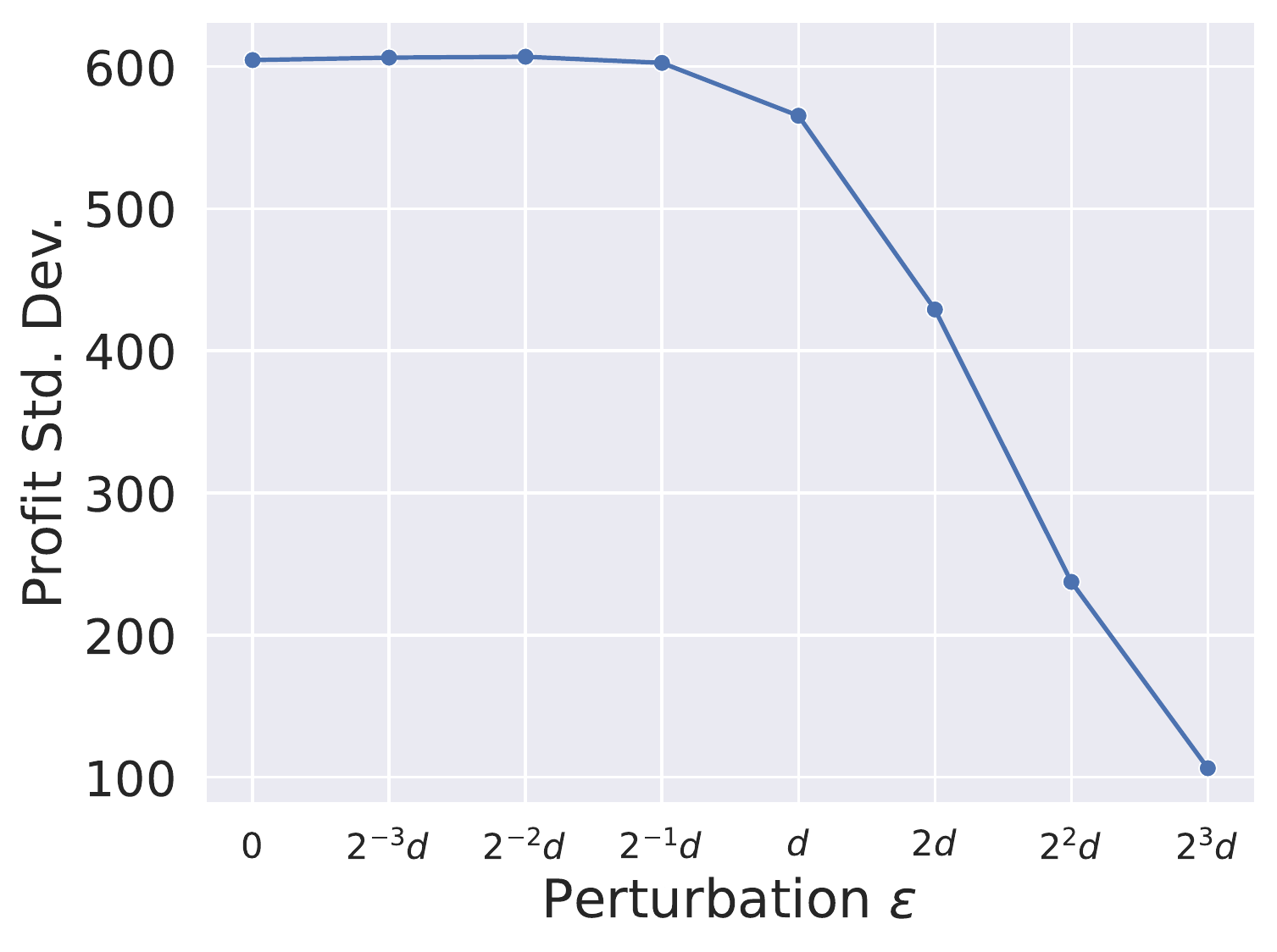}
\end{minipage}\hfill
\begin{minipage}{0.32\textwidth}%
	\includegraphics[width=\linewidth]{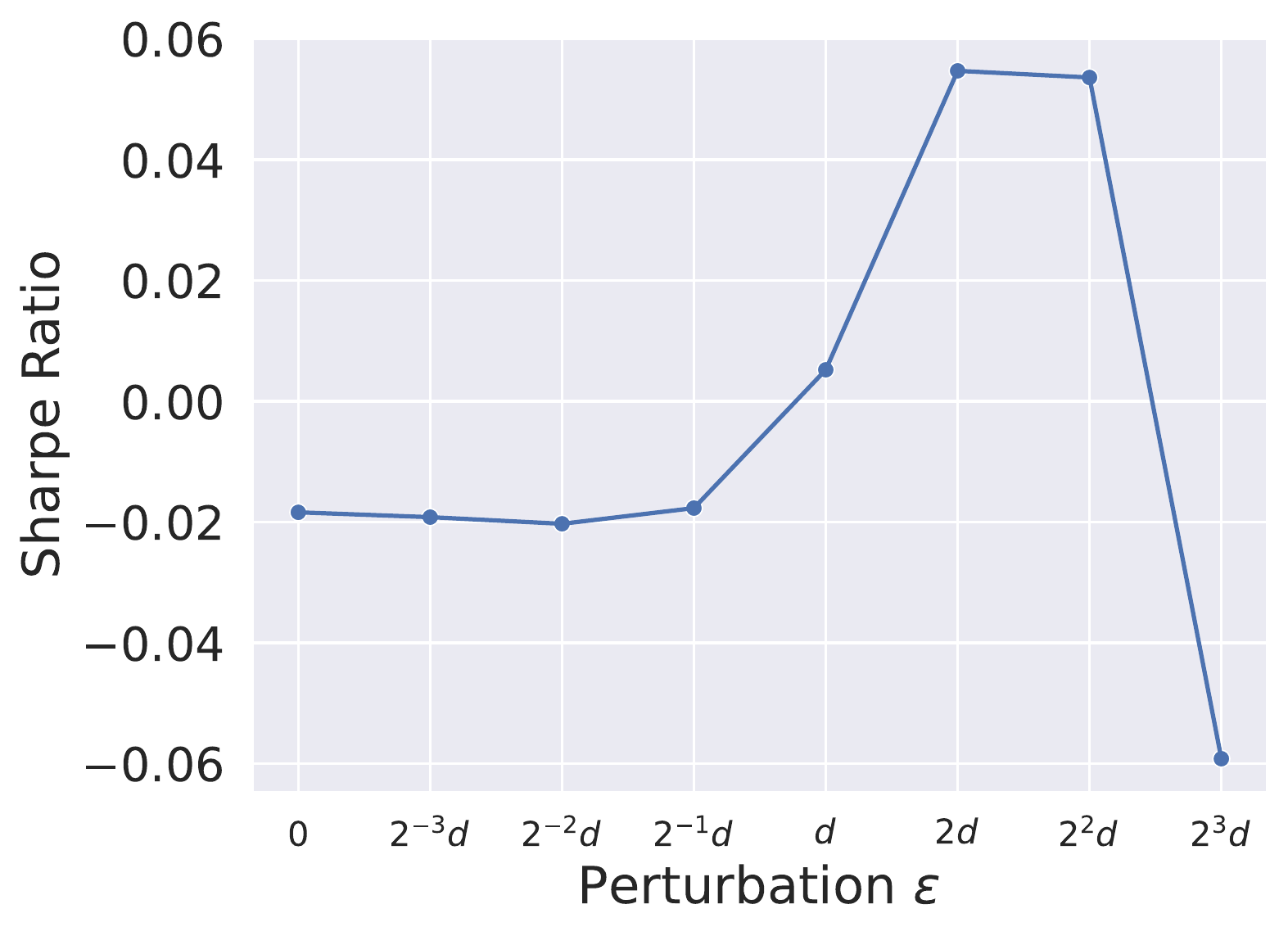}
\end{minipage}
	\caption{The top (resp.\,bottom) 3 figures show the $\varepsilon$-sensitivity analysis in the scenario of trading 20 (resp.\,50) assets where per share transactions cost are considered. Similar to Figure \ref{2_assets_sensitivity_analysis}, the results provide evidence that introducing a non-zero perturbation $\varepsilon$ is indeed beneficial for a better profit and sharpe ratio. Moreover, it can be seen that a larger $\varepsilon$, i.e. more uncertainty, makes the model more risk-averse, and therefore leads to a smaller standard deviation of the profits.}
\end{figure}

In particular, our results show that the approach can indeed be applied to high-dimensional data in contrast to conventional approaches such as linear programming.  Moreover, our approach clearly beats the market in all considered cases.

\subsection{Outperformance of pairs trading strategies} \label{exa_pairs_trading}
In a third example we show that with our approach using Algorithm~\ref{algo_training_nn} it is possible to trade profitably in a market environment where classical pairs trading approaches fail. To this end, we consider $d=2$ securities which show a high degree of mutual dependence. More specifically, we study the stocks of \emph{ExxonMobil} (Ticker symbol: XOM) and of \emph{BP p.l.c} (Ticker symbol: BP). The training period ranges from $1987/10/22$ to $2005/08/24$ and indeed, an augmented Dickey-Fuller test (compare e.g. \cite{fuller2009introduction}) indicates with a $p$-value of $ 0.0129$ that the spread $1.15 \cdot$XOM-BP is stationary and thus XOM and BP are indeed cointegrated. However, as depicted in Figure \ref{fig_test_period_divergence}, in the testing period that ranges from $2005/08/25$ to $2013/08/06$, the cointegration relationship breaks down, i.e., the spread diverges, and thus conventional pairs trading approaches fail. However, as the results in Table~\ref{tbl_profit_example_8} {and the equity curve in Figure~\ref{equity_curve_bp}} indicate, a trading strategy that was trained according to Algorithm~\ref{algo_training_nn} performs even in this environment remarkably well. This provides further evidence for the robustness of our presented approach. 
\begin{figure}[h!]
\begin{center}
\includegraphics[width=10cm]{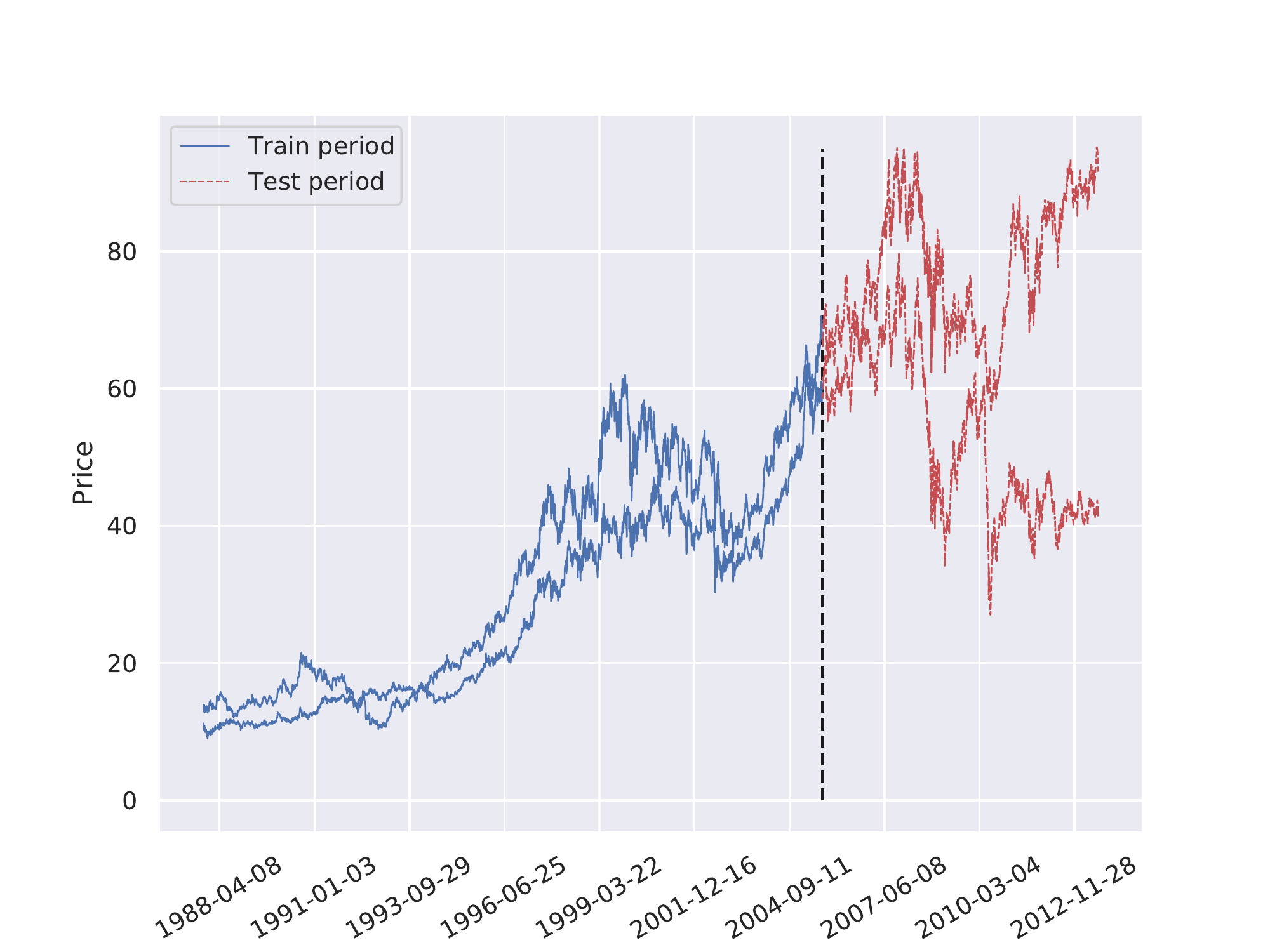}
\caption{This Figure shows the evolution of $XOM$ and $BP$ in the training period from $1987/10/22$ to $2005/08/23$ and the testing period from $2005/08/24$ to $2013/08/06$. In particular, we can observe that the cointegration relationship breaks down during the testing period.}\label{fig_test_period_divergence}
\end{center}
\end{figure}

\begin{table}[h!]
	\small
	\begin{center}
		\begin{tabular}{l>{\bfseries}c>{}c>{\bfseries}c>{}c>{\bfseries}c>{}c} \toprule
			
			{ Transaction Costs}& 0 & $0$ \text{B}\&\text{H} & Prop.& Prop. $\text{B}\&\text{H}$ & P.S. & P.S. $\text{B}\&\text{H}$ \\
			
			\midrule
			
			Overall Profit& 82.97 & -149.40& 66.59 & -202.26&62.03 & -230.20\\ 
			Average Profit &0.42 & -0.75& 0.33 & -1.01&0.31 & -1.15\\
			\% of Profitable Trades &54.37 &53.5 & 53.17 &53.0 &53.13 & 53.0\\
			Max. Profit  &43.77 & 121.0& 43.35 & 120.67&41.78 & 120.6\\ 
			Min. Profit &-17.74 & -221.7&-17.92 & -221.94&-18.25 & -222.1\\
			Sharpe Ratio & 0.4093 &-0.073 & 0.3216 & -0.099&0.3091 & -0.112\\ 
			Sortino Ratio& 0.8722 & -0.095&0.7125 &-0.129 &0.6951 & -0.146\\  \bottomrule
		\end{tabular}
		\caption{The Table shows the performance of the neural network statistical arbitrage strategies, trained according to Algorithm~\ref{algo_training_nn}, in the setting of Section~\ref{exa_pairs_trading}, where we test trading in $d=2$ securities in a testing period in which the spread, that was tested to be cointegrated in the training period, diverges. The profits in bold are depicted in dependence of { zero transaction costs ($0$), nonzero transaction costs with proportional~(Prop.), per share transaction costs~(P.S.), and online learning with per share transaction costs (P.S. O.L.), respectively}, and are compared with the outcomes of a buy-and-hold strategy (B$\&$H) with $\Delta_i^k=10$ for all $i=0,\dots,9$ and all $k=1,2$. { We  also report the 10-day average profits of the one-time-buy-and-hold strategy of the above three scenarios respectively, i.e., without opening and closing the positions during each trading window. (0 B\&H: 0.36, Prop. B\&H: 0.36, P.S. B\&H: 0.36)}}
		\label{tbl_profit_example_8}
	\end{center}
\end{table}

\begin{figure}[h!]
	\begin{center}
		\includegraphics[width=10cm]{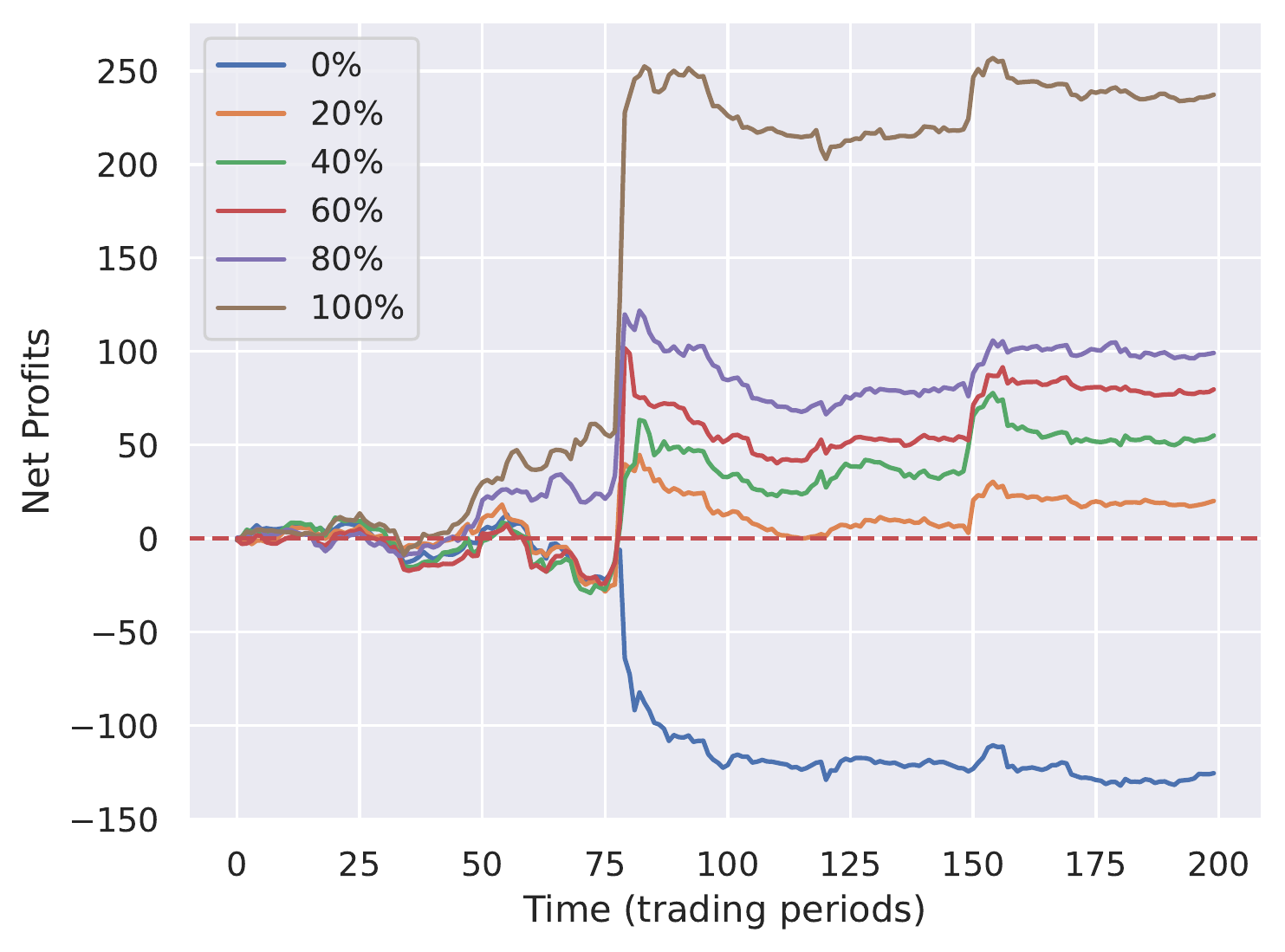}
		\caption{The figure shows the \emph{equity curves} of the trained strategies with per share transaction costs during the testing period shown in Figure \ref{fig_test_period_divergence}. The y-axis represents the cumulated net profits as the trading window propagates. Since a total of 50 independent experiments are conducted, we plot the equity curves by quantiles, where the rank is based on the net profit of the experiment at the end of the testing period. }\label{equity_curve_bp}
	\end{center}
\end{figure}

\section{Conclusion and Future Research}\label{sec_conclusion}
In this paper we have presented a deep-learning based approach to the computation of profitable trading strategies. In the main result in Theorem~\ref{thm_summary} we prove, in particular, the convergence of a penalized functional involving strategies that can be represented by neural networks towards the minimal conditional super-replication price of profitable trading strategies. This allows theoretically to determine profitable trading strategies by optimizing neural networks. To ensure also the practical applicability of our approach, we have introduced an algorithm in Section~\ref{sec_application} for the determination of a data-driven ambiguity set $\mathcal{P}$ of underlying probability measures as well as an algorithm for the computation of $\mathcal{P}$-robust statistical arbitrage strategies. In Section~\ref{sec_real_world_examples} we then have shown in several empirical examples that our approach indeed leads to profitable trading outcomes, even in scenarios where it is naturally difficult for trading strategies to perform well.

Analogue to the setting from \cite{lutkebohmert2021robust}, the presented approach can be extended in a natural way by trading in liquid options and by considering more general $\sigma$-algebras $\mathcal{G}$, leading to the notion of $\mathcal{P}$-robust $\mathcal{G}$-arbitrage strategies instead of $\mathcal{P}$-robust statistical arbitrage strategies. In this case, the $\sigma$-algebra $\mathcal{G}$ may require a different approximation as the case $\mathcal{G}=\sigma(S_{t_n})$, which was considered in Proposition~\ref{prop_filtration}. Moreover, the determination of the ambiguity set $\mathcal{P}$ in Section~\ref{sec_ambiguity_set_p}, even though entirely data-driven and straightforward, could also be replaced by approaches that include, for example, a more subtle weighting of past returns or even by a Bayesian approach in which one updates sequentially an assumed prior distribution contingent on observed returns.
Eventually we remark that all of our results, and in particular Theorem~\ref{thm_summary}, not only enable the approximation of $\mathcal{P}$-robust $\mathcal{G}$-arbitrage strategies, but also allow to detect strategies that conditionally super-replicate some payoff $\Phi$ and thus allow to involve mispriced derivatives (with payoff $\Phi$) into trading. 
We leave these extensions and its effects on the outcomes of the resultant trading strategies for future research.

\section{Proofs}\label{sec_proofs}
In this section we provide the proofs from the mathematical statements in Section~\ref{sec_main} and Section~\ref{sec_application}.
\begin{proof}[Proof of Lemma~\ref{lem_compactness}] Let $B>0$, $L>0$ and pick a sequence $( h_{c^{(k)},\Delta^{(k)}})_{k \in \N} \subset \HM$.
Note that, by definition of $\HM$, for all $k\in \N$, it holds $|c^{(k)}| \leq B$ as well as $|{\Delta_0^{j}}^{(k)}| \leq B$ for all $j=1,\dots,d$. Thus, according to the Bolzano--Weierstrass theorem, there exists a convergent subsequence of $(c^{(k)})_{k \in \N}$ as well as for $({\Delta_0^{j}}^{(k)})_{k \in \N}$ for each $j=1,\dots,d$. Moreover, for all  $i=1,\dots,n$, $j=1,\dots,d$ and for all $k \in \N$, we have $ \| {\Delta_i^j}^{(k)} \|_{\infty,\Omega_i}\leq B$ and that ${\Delta_i^j}^{(k)}$ is $L$-Lipschitz. Hence, by the Arzelà–-Ascoli theorem, for all $i=1,\dots,n$, $j=1,\dots,d$ there exists a subsequence of $({\Delta_i^j}^{(k)})_{k\in \N}$ converging uniformly on $\Omega_i$. Inductively we obtain therefore a subsequence such that $(c^{(k_l)})_{l \in \N}$ converges against some $c\in \R$, such that for each $j=1,\dots,d$ the subsequence $({\Delta^j_0}^{(k_l)})_{l\in \N}$ converges to some $\Delta_0^j\in \R$, and such that $({\Delta_i^j}^{(k_l)})_{l\in \N}$ converges uniformly for each $i=1,\dots,n$, $j=1,\dots,d$ against some $L$-Lipschitz function ${\Delta_i^j}$ . Due to the assumed continuity of the {trading costs ${c_{\operatorname{trans}}}_i^j, {c_{\operatorname{spread}}}_i^j, {c_{\operatorname{short}}}_i^j$}, for all $i=0,\dots,n,~j=1,\dots,d$, this implies the uniform convergence 
\[
\lim_{l \rightarrow \infty} h_{c^{(k_l)},\Delta^{(k_l)}}= h_{c,\Delta}\text{ on } \Omega.
\]
Moreover, it holds $h_{c,\Delta}\in \HM$, since 
\begin{align*}
|c|&=\lim_{l \rightarrow \infty} |c^{(k_l)}| \leq B,\\
\|\Delta_i^j\|_{\infty,\Omega_i}&=\lim_{l \rightarrow \infty}  \|{\Delta_i^j}^{(k_l)}\|_{\infty,\Omega_i}\leq B \text{ for all } i=1,\dots,n,~j=1,\dots,d,\\
|\Delta_0^j|&=\lim_{l \rightarrow \infty} |{\Delta_0^j}^{(k_l)}|\leq B \text{ for all } j=1,\dots,d.
\end{align*}
\end{proof}
\begin{proof}[Proof of Lemma~\ref{lem_attainment}]
Let $\G \subseteq\sigma(S)$ be some $\sigma$-algebra. We  define the map 
\begin{align*}
F: \HM &\to \R\\
h_{c,\Delta} &\mapsto c + k \sum_{\PP \in \mathcal{P}} \int_{\Omega} \beta\big(\E_{\PP}[\Phi(S)-h_{c,\Delta}(S) ~|~ \G]\big) \D \PP.
\end{align*}
The map $F$ is continuous with respect to the uniform topology induced by $\|\cdot \|_{\infty,\Omega}$ due to the dominated convergence theorem and by the assumed continuity of $\beta$. Let $\left(h_{c^{(i)},\Delta^{(i)}}\right)_{i \in \N}\subset \HM$ be a sequence with 
$\lim_{i \rightarrow \infty} F\left(h_{c^{(i)},\Delta^{(i)}}\right)=\Gamma_{B,L,k}(\Phi,
\G)$. Then, by the compactness of $\HM$, as stated in Lemma~\ref{lem_compactness}, there exists a subsequence $\left(h_{c^{(i_l)},\Delta^{(i_l)}}\right)_{l \in \N}$ and some ${h_{c,\Delta}} \in \HM$ such that 
\[
\Gamma_{B,L,k}(\Phi,
\G)=\lim_{l \rightarrow \infty} F\left(h_{c^{(i_l)},\Delta^{(i_l)}}\right)=F\left({h_{c,\Delta}}\right),
\]
where the last equality is a consequence of the continuity of $F$.
\end{proof}
\begin{proof}[Proof of Proposition~\ref{prop_convergence}]
Let $\G \subseteq\sigma(S)$ be some $\sigma$-algebra. First, we note that the condition $\|\Phi\|_{\infty,\Omega} \leq B$ ensures the existence of some strategy $h_{c,\Delta} \in \HM$ such that
\[
\E_{\PP}[h_{c,\Delta}(S)~|~\G] \geq \E_{\PP}[\Phi(S)~|~\G] ~~~\PP\text{-a.s. for all } \PP \in \mathcal{P}.
\]
Thus, due to Assumption~\ref{asu_beta}, we have for all $k \in \N$
\begin{align*}
\Gamma_{B,L,k}(\Phi,\G)&\leq\inf_{h_{c,\Delta} \in \HM} \bigg\{c~+ k \sum_{\PP \in \mathcal{P}} \int_{\Omega} \beta\big(\E_{\PP}[\Phi(S)-h_{c,\Delta}(S) ~|~ \G]\big) \D \PP ~\bigg|\\
&\hspace{4cm}\E_{\PP}[h_{c,\Delta}(S)~|~\G] \geq \E_{\PP}[\Phi(S)~|~\G] ~~~\PP\text{-a.s. for all } \PP \in \mathcal{P}\bigg\} \\
&=\inf_{h_{c,\Delta} \in \HM} \bigg\{ c ~~~\bigg|~ \E_{\PP}[h_{c,\Delta}(S)~|~\G] \geq \E_{\PP}[\Phi(S)~|~\G] ~~~\PP\text{-a.s. for all } \PP \in \mathcal{P}\bigg\} = \Gamma_{B,L}(\Phi,\G).
\end{align*}
By the assumption that $\beta \geq 0$ and by the above inequality we obtain for each $k \in \N$
\[
-B \leq \Gamma_{B,L,k}(\Phi,\G) \leq \Gamma_{B,L,k+1}(\Phi,\G) \leq \cdots \leq \Gamma_{B,L}(\Phi,\G)\leq B < \infty.
\]
This means in particular that $\Gamma_{B,L,\infty}(\Phi,\G):=\lim_{k \rightarrow \infty} \Gamma_{B,L,k}(\Phi,\G)$ exists and
\begin{equation}\label{eq_proof_ineq_1}
\Gamma_{B,L,\infty}(\Phi,\G)\leq \Gamma_{B,L}(\Phi,\G).
\end{equation}
Moreover, by Lemma~\ref{lem_attainment}, for each $k \in \N$ there exists $h_{c^{(k)},\Delta^{(k)}}\in \HM$ which minimizes $\Gamma_{B,L,k}(\Phi,\G)$, i.e., we have that
\begin{equation}\label{eq_chain_1}
c^{(k)}~+ k \sum_{\PP \in \mathcal{P}} \int_{\Omega} \beta\big(\E_{\PP}[\Phi(S)-h_{c^{(k)},\Delta^{(k)}}(S) ~|~ \G]\big) \D \PP = \Gamma_{B,L,k}(\Phi,\G)\leq B.
\end{equation}
According to Lemma~\ref{lem_compactness}, $\HM$ is compact, thus there exists a uniformly convergent subsequence $(h_{c^{(k_l)},\Delta^{(k_l)}})_{l\in \N}$ with limit $h_{c,\Delta} \in \HM$. Hence, by \eqref{eq_chain_1} we have 
\begin{align}
-B \leq &\lim _{l \rightarrow \infty} c^{(k_l)}~+ k_l \sum_{\PP \in \mathcal{P}} \int_{\Omega}  \beta\big(\E_{\PP}[\Phi(S)-h_{c^{(k_l)},\Delta^{(k_l)}}(S) ~|~ \G]\big) \D \PP \notag \\
=& c~+ \lim _{l \rightarrow \infty} k_l \cdot \sum_{\PP \in \mathcal{P}} \int_{\Omega}  \beta\big(\E_{\PP}[\Phi(S)-h_{c^{(k_l)},\Delta^{(k_l)}}(S) ~|~ \G]\big) \D \PP\leq B<\infty.  \label{eq_lower_infinty}
\end{align}
Since $\beta$ is continuous and nonnegative the dominated convergence theorem and the boundedness of the expression in \eqref{eq_lower_infinty} imply that 
\begin{equation}\label{eq_sum_zero}
0=\lim_{l \rightarrow \infty}\sum_{\PP \in \mathcal{P}} \int_{\Omega} \beta\big(\E_{\PP}[\Phi(S)-h_{c^{(k_l)},\Delta^{(k_l)}}(S) ~|~ \G]\big) \D \PP=\sum_{\PP \in \mathcal{P}} \int_{\Omega}  \beta\big(\E_{\PP}[\Phi(S)-h_{c,\Delta}(S) ~|~ \G]\big) \D \PP.
\end{equation}
By Assumption~\ref{asu_beta}, this in turn can only hold if
\begin{equation}\label{eq_inequality_p_tohold}
\E_{\PP}[h_{c,\Delta}(S) ~|~\G] \geq \E_{\PP}[\Phi(S)~|~\G]  ~~~\PP\text{-a.s. for all }  \PP \in \mathcal{P},
\end{equation}
By the validity of \eqref{eq_inequality_p_tohold}, we see that 
\begin{equation}\label{eq_d_geq_gamma}
c \geq \Gamma_{B,L}(\Phi,\G).
\end{equation}
Further, since $\beta \geq 0$, we have with \eqref{eq_chain_1} that for all $l \geq 1$ 
\begin{align}
\label{eq_l_0} -B \leq c^{(k_l)}\leq &c^{(k_l)}~+ k_{l} \sum_{\PP \in \mathcal{P}} \int_{\Omega} \beta\big(\E_{\PP}[\Phi(S)-h_{c^{(k_l)},\Delta^{(k_l)}}(S)  ~|~ \G]\big) \D \PP = \Gamma_{B,L,{k_l}}(\Phi,\G)  < \infty.
\end{align}
We consider in \eqref{eq_l_0} the limit $l \rightarrow \infty$ and obtain with \eqref{eq_sum_zero} that
\begin{equation}\label{eq_d_lim_delta}
c \leq \lim_{l\rightarrow \infty} \Gamma_{B,L,{k_l}}(\Phi,\G)=\Gamma_{B,L,\infty}(\Phi,\G).
\end{equation}
Thus, combining \eqref{eq_proof_ineq_1},~\eqref{eq_d_lim_delta}, and \eqref{eq_d_geq_gamma} yields the following inequalities
\[
\Gamma_{B,L}(\Phi,\G) \geq \Gamma_{B,L,\infty}(\Phi,\G) \geq  c \geq \Gamma_{B,L}(\Phi,\G),
\]
which implies that $\Gamma_{B,L}(\Phi,\G)=\Gamma_{B,L,\infty}(\Phi,\G)$.
\end{proof}

\begin{proof}[Proof of Lemma~\ref{lem_universal_approx}]
Let $B>0$, $L>0$, $0 < \varepsilon < B$, $i \in \{1,\dots,n\}$, and let $f:\Omega_i\rightarrow \R^d$ be some function such that $f_j:=\pi_j \circ f$ is $L$-Lipschitz with $\|f_j\|_{\Omega_i,\infty,1} \leq  B$ for all $j=1,\dots,d$. We define for all $j=1,\dots,n$ the truncated function $\widetilde{f}_j:\Omega_i \rightarrow \R$ by
\[
\widetilde{f}_j:=\max\left\{\min\left\{f_j,B-\frac{\varepsilon}{2\sqrt{d}}\right\},-B+\frac{\varepsilon}{2\sqrt{d}}\right\}
\]
which is $L$-Lipschitz, since for all $x,y \in \Omega_i$ we have by construction
\[
\left|\widetilde{f}_j(x)-\widetilde{f}_j(y)\right|\leq|{f}_j(x)-{f}_j(y)|\leq L \|x-y\|_{id}.
\]
Then, according to a version of the universal approximation theorem for Lipschitz functions in the form of \cite[Theorem 1]{eckstein2020lipschitz}, there exists for all $j=1,\dots,d$ some neural network $g_j:\R^{id} \rightarrow \R$, $g_j \in \mathfrak{N}_{id,1}$,  which is $L$-Lipschitz continuous, and which fulfils
\begin{equation}\label{eq_f-n_eps_2}
\|g_j-\widetilde{f}_j\|_{\infty,\Omega_i,1} \leq  \frac{\varepsilon}{2\sqrt{d}}.
\end{equation}
Thus,  by \eqref{eq_f-n_eps_2}, and since by construction $\|\widetilde{f}_j\|_{\infty,\Omega_i}\leq B-\frac{\varepsilon}{2\sqrt{d}}$, we have for all  $j=1,\dots,d$ that
\[
\|g_j\|_{\infty,\Omega_i,1}\leq \|\widetilde{f}_j\|_{\infty,\Omega_i,1}+\|g_j-\widetilde{f}_j\|_{\infty,\Omega_i,1}\leq  B-\frac{\varepsilon}{2\sqrt{d}}+\frac{\varepsilon}{2\sqrt{d}}=B.
\]
Moreover, we observe for all $j=1,\dots,d$ that 
\[
\|g_j-{f}_j\|_{\infty,\Omega_i,1}\leq \|g_j-\widetilde{f}_j\|_{\infty,\Omega_i,1}+\|\widetilde{f}_j-{f}_j\|_{\infty,\Omega_i,1}\leq  \varepsilon/\sqrt{d}.
\]
Then by defining $g:=(g_1,\dots,g_d) $ we obtain
\[
\|g-f\|_{\infty, \Omega_i} \equiv \|g-f\|_{\infty, \Omega_i,d} \leq \sqrt{\sum_{i=1}^d \|g_i-f_i\|_{\infty, \Omega_i,1}^2} \leq \varepsilon.
\]
Finally, note that $g \in \mathfrak{N}_{i,B,L}$, as the set $\mathfrak{N}_{i,B,L}$ also contains neural networks which are not fully connected, which can be seen by setting the weights of connections from a neural network in $\mathfrak{N}_{i,B,L}$, which do not appear in partially connected networks, to be $0$.
\end{proof}

\begin{proof}[Proof of Proposition~\ref{prop_neuralnetworks}]
Let $\G \subseteq\sigma(S)$ be some $\sigma$-algebra, and let $k \in \N$. First note that as $\HM^{\mathcal{N}\mathcal{N}} \subseteq \HM$ we have $\Gamma_{B,L,k}^{\mathcal{N}\mathcal{N}}(\Phi) \geq \Gamma_{B,L,k}(\Phi)$. Hence it remains to prove that $\Gamma_{B,L,k}^{\mathcal{N}\mathcal{N}}(\Phi) \leq \Gamma_{B,L,k}(\Phi)$. To that end, according to Lemma~\ref{lem_attainment} there exists some $h_{c,\Delta} \in \HM$ such that 
\[
c~+ k \sum_{\PP \in \mathcal{P}} \int_{\Omega} \beta\big(\E_{\PP}[\Phi(S)-h_{c,\Delta}(S) ~|~ \G]\big) \D \PP=\Gamma_{B,L,k}(\Phi,\mathcal{G}).
\]
By Lemma~\ref{lem_universal_approx} there exists a sequence $(\Delta^{(m)})_{m \in \N}:=\left(({\Delta_i^j}^{(m)})_{i,j}\right)_{m \in \N}$ such that for all $i=1,\dots,n-1$ we have $\left({\Delta_i^1}^{(m)},\dots,{\Delta_i^d}^{(m)}\right) \in \mathfrak{N}_{i,B,L}$ for all $m\in \N$ and such that 
\[
\left({\Delta_i^1}^{(m)},\dots,{\Delta_i^d}^{(m)}\right) \rightarrow \left({\Delta_i^1},\dots,{\Delta_i^d}\right) \text{ uniformly on }\Omega_i \text{ for }m \rightarrow \infty.
\]
This implies particularly for all $i=1,\dots,n$, $j=1,\dots,d$ that ${\Delta_i^j}^{(m)}\rightarrow\Delta_i^j$ uniformly on $\Omega_i$ for $m \rightarrow \infty$.  By the dominated convergence theorem, with the continuity of $\beta$ and the continuity of the {trading costs ${c_{\operatorname{trans}}}_i^j, {c_{\operatorname{spread}}}_i^j, {c_{\operatorname{short}}}_i^j$ }for all $i=0,\dots,n$, $j=1,\dots,d$,  we have that 
\begin{align*}
\lim_{m \rightarrow \infty} &c~+ k \sum_{\PP \in \mathcal{P}} \int_{\Omega} \beta\big(\E_{\PP}[\Phi(S)-h_{c,\Delta^{(m)}}(S) ~|~ \G]\big) \D \PP\\
=&c~+ k \sum_{\PP \in \mathcal{P}} \int_{\Omega} \beta\big(\E_{\PP}[\Phi(S)-h_{c,\Delta}(S) ~|~ \G]\big) \D \PP=\Gamma_{B,L,k}(\Phi,\mathcal{G}).
\end{align*}
Thus, as $h_{c,\Delta^{(m)}} \in \HM^{\mathcal{N}\mathcal{N}}$ for all $m\in \N$, we obtain
\[
\Gamma_{B,L,k}^{\mathcal{N}\mathcal{N}}(\Phi)\leq \lim_{m \rightarrow \infty} \left(c~+ k \sum_{\PP \in \mathcal{P}} \int_{\Omega} \beta\big(\E_{\PP}[\Phi(S)-h_{c,\Delta^{(m)}}(S) ~|~ \G]\big) \D \PP \right) = \Gamma_{B,L,k}(\Phi,\mathcal{G}).
\]
\end{proof}

\begin{proof}[Proof of Proposition~\ref{prop_filtration}]
We first show that the following equality holds $\PP^U$-almost surely:\footnote{More precisely, we show that the following holds for $\PP^U$-almost all realizations of $(\mathcal{E}_i)_{i \in \N}$ following \eqref{defn_E_i} and \eqref{defn_A_i}.}
\begin{equation}\label{eq_claim_proof_filtration}
\sigma(S_{t_n})=\sigma \left(\bigcup_{i=1}^{\infty}\mathcal{F}_{i}\right).
\end{equation}
Note that for all $i\in\mathbb{N}$ it holds
\begin{equation}\label{f_i_in_sigma_s_tn}
\mathcal{F}_{i}=\left\{S_{t_n}^{-1}(A)~\middle|~ A \in \sigma(\mathcal{E}_i)\right\}\subseteq \left\{S_{t_n}^{-1}(A)~\middle|~ A \in \mathcal{B}\bigg([\underline{K}^1,\overline{K}^1]\times \cdots \times [\underline{K}^d,\overline{K}^d]\bigg)\right\}=\sigma(S_{t_{n}}).
\end{equation}
Hence, it follows by \eqref{f_i_in_sigma_s_tn} that
$
\bigcup_{i=1}^{\infty}\mathcal{F}_{i}\subseteq \sigma(S_{t_n}),
$
and in particular 
\[
\sigma\left(\bigcup_{i=1}^{\infty}\mathcal{F}_{i}\right)\subseteq\sigma(S_{t_{n}}).
\]

This proves one inclusion of \eqref{eq_claim_proof_filtration}. Next, for the other inclusion, pick for $j=1,\dots,d$ some arbitrary  numbers $a_j\in [\underline{K}^j,\overline{K}^j)$.
Then we aim to show that $\PP^U$-almost surely 
\begin{equation}\label{interval_in_sigma_partition}
(a_1,\overline{K}^1]\times\dots\times(a_d,\overline{K}^d] \in \sigma \left(\bigcup_{i=1}^\infty \mathcal{E}_i\right).
\end{equation}
Let $a^{(i)}_j$, $b^{(i)}_j$ for $j=1,...,d$, $i\in \N$  be the random variables which are $\PP^U$-distributed, i.e, 
\[
 a_j^{(i)} \sim \mathcal{U}\left([\underline{K}^j,\overline{K}^j]\right), \qquad  b_j^{(i)}\equiv \overline{K}^j,\qquad j=1,...,d,~~i\in \N,
 \]
where $ a_j^{(i)}$ is independent of  $a_k^{(l)}$ for $(j,i) \neq (k,l)$.
%
Choose $i_0\in \N$ so large such that $\overline{K}^j-a_j>i^{-\frac{1}{d}}$ for all $i \geq i_0$ and  for all $j=1,\dots,d$. Then, we define events
\[
E_i:=\bigg\{ a_j<{a^{(i)}_j}<a_j+i^{-\frac{1}{d}}~~\text{for all}~~j=1,\dots,d\bigg\}, ~i\in\mathbb{N}.
\]
It follows now from $a_j^{(i)} \sim \mathcal{U}\left([\underline{K}^j,\overline{K}^j]\right)$ under $\PP^U$ that the following equality holds for all $i \geq i_0$
\begin{align*}
\PP^U\left(a_j<a^{(i)}_j<a_j+i^{-\frac{1}{d}}\right)& = \frac{1}{\overline{K}^j-\underline{K}^j}\cdot i^{-\frac{1}{d}}
\end{align*}
for each $j=1,\dots,d$. Thus, by using the independence of $(a_j^{(i)})_{j=1,\dots,d}$, we have for all $i \geq i_0$
\begin{align*}
\PP^U(E_i)=\prod_{j=1}^{d}\PP^U\left(a_j<a^{(i)}_j<a_j+i^{-\frac{1}{d}}\right)=\prod_{j=1}^{d}\frac{1}{(\overline{K}^j-\underline{K}^j)}i^{-\frac{1}{d}}=\frac{1}{i} \cdot \prod_{j=1}^{d}\frac{1}{(\overline{K}^j-\underline{K}^j)}.
\end{align*}
This implies that $\sum_{i=1}^{\infty}\PP^U(E_i)=\infty$. By the Borel--Cantelli lemma and by the fact that the events $\{E_i, i \in \N\}$ are independent, we have $\PP^U\left(\limsup\limits_{i\rightarrow\infty}{E_i}\right)=1$. Thus, $\PP^U$-almost surely we have 
\begin{equation}\label{eq_aj_bj_condition}
a_j<a^{(i)}_j<a_j+i^{-\frac{1}{d}}~~\text{for all}~~j=1,\dots,d
\end{equation}
for infinitely many $i\in \N$. Therefore, $\PP^U$-a.s., there exists a subsequence $\left(\left(a_j^{(i_k)}\right)_{j=1,\dots,d}\right)_{k \in \N}$ which fulfils \eqref{eq_aj_bj_condition} for every $k\in \N$. Hence, $\PP^U$-a.s., according to the Bolzano--Weierstrass theorem, there exists a subsequence which converges, and due to \eqref{eq_aj_bj_condition} we have $a_j^{(i_k)} \downarrow a_j$ for $k \rightarrow \infty$ for all $j=1,\dots,d$. We hence get
\[
(a_1,\overline{K}^1]\times \cdots \times(a_d,\overline{K}^d] = \bigcup_{k\in \N} (a_1^{(i_k)},\overline{K}^1] \times \cdots \times (a_d^{(i_k)},\overline{K}^d] \in  \sigma \left(\bigcup_{i=1}^\infty \mathcal{E}_i\right).
\]
Next, note that $\mathcal{B}([\underline{K}^1,\overline{K}^1]\times \cdots \times [\underline{K}^d,\overline{K}^d])$ is generated by 
\[
\bigg\{(a_1,\overline{K}^1]\times\dots\times(a_d,\overline{K}^d]~\bigg|~\underline{K}^i\leq a_i < \overline{K}^i,~i=1,\dots,d\bigg\}.
\]
Therefore, by using \eqref{interval_in_sigma_partition} we obtain $\PP^U$-a.s. that
\begin{align*}
\sigma(S_{t_n})&=\left\{S_{t_n}^{-1}(A)~\middle|~ A \in \sigma \left(\bigg\{(a_1,\overline{K}^1]\times\dots\times(a_d,\overline{K}^d]~\bigg|~\underline{K}^i\leq a_i < \overline{K}^i,~i=1,\dots,d\bigg\}\right)\right\}\\
&\subseteq\left\{S_{t_n}^{-1}(A)~\middle|~ A \in \sigma \left(\bigcup_{i=1}^{\infty}\mathcal{E}_i\right)\right\}\\
&=\sigma\left(\bigcup_{i=1}^{\infty}\left\{S_{t_n}^{-1}(A)~\middle|~ A \in \sigma\left(\mathcal{E}_i\right)\right\}\right)=\sigma\left(\bigcup_{i=1}^{\infty}\mathcal{F}_{i}\right).
\end{align*}
Thus, the claim \eqref{eq_claim_proof_filtration} is proved. 

Next, we show (ii). To this end, let $\Phi:\Omega\rightarrow \R$ be Borel measurable with $\|\Phi\|_{\infty,\Omega} \leq B$, let $h_{c,\Delta}(S) \in \HM$, and $\PP \in \mathcal{P}$. Since $\Phi$ and $h_{c,\Delta}(S)$ are both bounded, it holds $\Phi-h_{c,\Delta}(S) \in L^1(\mathbb{P})$.
Moreover, by definition, $(\mathcal{F}_{i})_{i \in \N}$ is an increasing sequence of $\sigma$-algebras. Given $\sigma(S_{t_n})=\sigma \left(\bigcup_{i=1}^{\infty}\mathcal{F}_{i}\right)$, it follows from Lévy's zero-one law that,
$$
\lim_{i\to\infty}\mathbb{E}_{\mathbb{P}}[\Phi(S)-h_{c,\Delta}(S)~|~\mathcal{F}_{i}]=\mathbb{E}_{\mathbb{P}}\left[\Phi(S)-h_{c,\Delta}(S)~\middle|~\sigma \left(\bigcup_{i=1}^{\infty}\mathcal{F}_{i}\right)\right]=\mathbb{E}_{\mathbb{P}}[\Phi(S)-h_{c,\Delta}(S)~|~\sigma(S_{t_n})],
$$
where, by Lévy's zero-one law, the convergence holds both $\mathbb{P}$-almost surely and in $L^1(\mathbb{P})$.

\end{proof}

\begin{proof}[Proof of Theorem~\ref{thm_summary}]
Let $B>0$, $L>0$. Moreover, let $\sigma(S_{t_n})=\sigma \left(\bigcup_{i=1}^{\infty}\mathcal{F}_{i}\right)$, which holds true $\PP^U$-a.s. according to Proposition~\ref{prop_filtration}. By Lemma~\ref{lem_attainment}, for all $k \in \N$ and all $i \in \N$ there exists some strategy $h_{c^{(i)},\Delta^{(i)}} \in \HM$ such that
\[
\Gamma_{B,L,k}(\Phi,\mathcal{F}_i)=c^{(i)}~+ k \sum_{\PP \in \mathcal{P}} \int_{\Omega} \beta\big(\E_{\PP}[\Phi(S)-h_{c^{(i)},\Delta^{(i)}}(S)  ~|~ \mathcal{F}_i]\big) \D \PP,
\]
and some $h_{c_{\infty},\Delta_{\infty}} \in \HM$ s.t.
\[
\Gamma_{B,L,k}(\Phi,\sigma(S_{t_n}))=c_{\infty}~+ k \sum_{\PP \in \mathcal{P}} \int_{\Omega} \beta\big(\E_{\PP}[\Phi(S)-h_{c_{\infty},\Delta_{\infty}}(S)  ~|~ \sigma(S_{t_n})]\big) \D \PP.
\]
According to Lemma~\ref{lem_compactness}, $\HM$ is compact, thus there exists a subsequence of $(h_{c^{(i)},\Delta^{(i)}} )_{i \in \N}$ (labelled identically) such that
$(h_{c^{(i)},\Delta^{(i)}} )_{i \in \N}$ converges uniformly for $i \rightarrow \infty$ against some $ h_{c,\Delta}\in \HM$.
Then, it follows with the dominated convergence theorem, the continuity of the {trading costs}, the continuity of $\beta$, and with Proposition~\ref{prop_filtration} that
\begin{align}
\Gamma_{B,L,k}(\Phi,\sigma(S_{t_n})) &\leq c~+ k \sum_{\PP \in \mathcal{P}} \int_{\Omega} \beta\big(\E_{\PP}[\Phi(S)-h_{c,\Delta}(S)  ~|~ \sigma(S_{t_n})]\big) \D \PP \notag \\
&=\lim_{i \rightarrow \infty} c~+ k \sum_{\PP \in \mathcal{P}} \int_{\Omega} \beta\big(\E_{\PP}[\Phi(S)-h_{c,\Delta}(S)  ~|~ \mathcal{F}_i]\big) \D \PP\notag \\
&=\lim_{i \rightarrow \infty}\lim_{j \rightarrow \infty} c^{(j)}~+ k \sum_{\PP \in \mathcal{P}} \int_{\Omega} \beta\big(\E_{\PP}[\Phi(S)-h_{c^{(j)},\Delta^{(j)}}(S)  ~|~ \mathcal{F}_i]\big) \D \PP. \label{eq_limilimj}
\end{align}
By the tower property of the conditional expectation and by Jensen's inequality, as $\beta$ is convex, it holds for every $j \geq i$, by using $\mathcal{F}_i \subseteq \mathcal{F}_j$, that
\begin{align}
\int \beta\big(\E_{\PP}[\Phi(S)-h_{c^{(j)},\Delta^{(j)}}(S)  ~|~ \mathcal{F}_i]\big) \D \PP&=\int_{\Omega} \beta\bigg(\E_{\PP}\left[\E_{\PP}[\Phi(S)-h_{c^{(j)},\Delta^{(j)}}(S)  ~|~ \mathcal{F}_j]~\middle|~\mathcal{F}_i\right]\bigg) \D \PP\notag \\
&\leq \int_{\Omega} \E_{\PP}\left[\beta\big(\E_{\PP}[\Phi(S)-h_{c^{(j)},\Delta^{(j)}}(S)  ~|~ \mathcal{F}_j]\big)~\middle|~\mathcal{F}_i\right] \D \PP\notag \\
&= \int_{\Omega} \beta\big(\E_{\PP}[\Phi(S)-h_{c^{(j)},\Delta^{(j)}}(S)  ~|~ \mathcal{F}_j]\big) \D \PP.
\end{align}
Thus, we obtain by \eqref{eq_limilimj} that
\begin{align*}
\Gamma_{B,L,k}(\Phi,\sigma(S_{t_n}))&\leq \lim_{j \rightarrow \infty} c^{(j)}~+ k \sum_{\PP \in \mathcal{P}} \int_{\Omega} \beta\big(\E_{\PP}[\Phi(S)-h_{c^{(j)},\Delta^{(j)}}(S)  ~|~ \mathcal{F}_j]\big) \D \PP\\
&=  \lim_{j \rightarrow \infty}\Gamma_{B,L,k}(\Phi,\mathcal{F}_j)\\
&\leq \lim_{j \rightarrow \infty} c_{\infty}~+ k \sum_{\PP \in \mathcal{P}} \int_{\Omega} \beta\big(\E_{\PP}[\Phi(S)-h_{c_{\infty},\Delta_{\infty}}(S)  ~|~ \mathcal{F}_j]\big) \D \PP\\
&= c_{\infty}~+ k \sum_{\PP \in \mathcal{P}} \int_{\Omega} \beta\big(\E_{\PP}[\Phi(S)-h_{c_{\infty},\Delta_{\infty}}(S)  ~|~ \sigma(S_{t_n})]\big) \D \PP=\Gamma_{B,L,k}(\Phi,\sigma(S_{t_n})).
\end{align*}
This means we have shown that for each $k\in \N$
\[
\lim_{j \rightarrow \infty}\Gamma_{B,L,k}(\Phi,\mathcal{F}_j) = \Gamma_{B,L,k}(\Phi,\sigma(S_{t_n})).
\]
Therefore we conclude with Proposition~\ref{prop_neuralnetworks}, and Proposition~\ref{prop_convergence} that
\[
\lim_{k \rightarrow \infty}\lim_{i \rightarrow \infty}\Gamma_{B,L,k}^{\mathcal{N}\mathcal{N}}(\Phi,\mathcal{F}_i)=\lim_{k \rightarrow \infty}\lim_{i \rightarrow \infty}\Gamma_{B,L,k}(\Phi,\mathcal{F}_i)=\lim_{k \rightarrow \infty} \Gamma_{B,L,k}(\Phi,\sigma(S_{t_n})) = \Gamma_{B,L}(\Phi,\sigma(S_{t_n})).
\]
\end{proof}
\begin{proof}[Proof of Lemma~\ref{lem_wasserstein}]
We set $x_\ell:=\left(S_{t_0}\cdot \frac{Y_{s_{\ell+1}}}{Y_{s_\ell}},\dots,S_{t_0}\cdot \frac{Y_{s_{\ell+n}}}{Y_{s_\ell}}\right)\in \R^{nd}$ for $\ell=1,\dots,N-n$ and define
\[
\widetilde{\pi}:= \frac{1}{N-n}\sum_{\ell=1}^{N-n}\delta_{(x_{\ell},x_{\ell}+\tau_\ell)}\in \Pi(\widehat{\PP},\widehat{\PP}_\tau) \subset \mathcal{M}_1(\R^{nd}\times \R^{nd}).
\]
Then it holds by construction that
\begin{align*}
\mathcal{W}(\widehat{\PP},\widehat{\PP}_\tau) &= \inf_{\pi \in \Pi(\widehat{\PP},\widehat{\PP}_\tau)}\int_{\R^{nd}\times \R^{nd}} \| u-v \|_{nd} \,\D\pi(u,v)\\
&\leq \int_{\R^{nd}\times \R^{nd}} \| u-v \|_{nd} \,\D\widetilde{\pi}(u,v)= \frac{1}{N-n}\sum_{\ell=1}^{N-n} \|\tau_\ell \|_{nd} < \varepsilon.
\end{align*}
This shows $\widehat{\PP}_\tau \in \mathcal{B}_{\varepsilon}(\widehat{\PP})$.
\end{proof}

\begin{proof}[Proof of Remark~\ref{rem_choice_Omega}]
We write  $\tau_\ell = \left({\tau_\ell}^{i,j} \right)_{i=1,\dots,n, \atop j=1,\dots,d}$ for $\ell=1,\dots,N-n$. Note that \eqref{eq_ineq_tau_epsilon} implies $\left|{\tau_\ell}^{i,j}\right|< \varepsilon \leq \delta$ for all $\ell =1,\dots,N-n$, $i=1,\dots,n$, $j=1,\dots,d$. Hence, it follows 
\[
\underline{K}^j = \underline{S}^j-\delta \leq S_{t_0}^j \cdot \frac{Y_{s_{\ell+i}}^j }{Y_{s_\ell}^j} +{\tau_\ell}^{i,j} \leq \overline{S}^j+\delta = \overline{K}^j ,
\]
for all $\ell =1,\dots,N-n$, $i=1,\dots,n$, $j=1,\dots,d$, which yields the assertion that $\widehat{\PP}_\tau \in \mathcal{M}_1(\Omega)$. 
\end{proof}

\section*{Acknowledgments}
\noindent
Financial support by the MOE AcRF Tier 1 Grant \emph{RG74/21} and by the  Nanyang Assistant Professorship Grant (NAP Grant) \emph{Machine Learning based Algorithms in Finance and Insurance} is gratefully acknowledged. 
\bibliographystyle{plain} 
\bibliography{literature}
\end{document}